\documentclass[11pt]{article}

\usepackage[numbers]{natbib} 

\usepackage{titlesec}
\usepackage{mathtools}
\usepackage{bbm}
\usepackage{xcolor}
\usepackage{amsfonts,amsmath,amssymb,color,amsthm,boxedminipage,url,fullpage,xparse}
\definecolor{weborange}{rgb}{.8,.3,.3}
\definecolor{webblue}{rgb}{0,0,.8}
\definecolor{internallinkcolor}{rgb}{0,.5,0}
\definecolor{externallinkcolor}{rgb}{0,0,.5}

\usepackage{amsthm} 
\usepackage{stackengine}
\usepackage{comment}
\usepackage{enumerate,paralist}
\usepackage[labelfont=bf]{caption}
\usepackage{aliascnt}

\ifdefined\asdvi
\usepackage[dvips,
hyperfootnotes=false,
colorlinks=true,
urlcolor=externallinkcolor,
linkcolor=internallinkcolor,
filecolor=externallinkcolor,
citecolor=internallinkcolor,
breaklinks=true,
pdfstartview=FitH,
pdfpagelayout=OneColumn]{hyperref} 
\else
\usepackage[
hyperfootnotes=false,
colorlinks=true,
urlcolor=externallinkcolor,
linkcolor=internallinkcolor,
filecolor=externallinkcolor,
citecolor=internallinkcolor,
breaklinks=true,
pdfstartview=FitH,
pdfpagelayout=OneColumn]{hyperref}
\fi

\usepackage{amsfonts,amsmath,amssymb,amsthm,boxedminipage,color,url,fullpage}
\usepackage[normalem]{ulem}
\usepackage{cleveref}
\usepackage{xspace}
\usepackage{xstring}
\usepackage{tikz}
\usetikzlibrary{arrows}
\usetikzlibrary{decorations.markings}
\usepackage{subcaption}
\usepackage{tabularx}
\usepackage{bbold}
\usepackage{dsfont}
\usepackage{titlesec}
\usepackage{enumitem}
\usepackage{nicefrac}

\providecommand{\remove}[1]{}
\newcommand{\Draft}[1]{\ifdefined\IsDraft \texttt{ #1} \fi}

\ifdefined\IsDraft
 \newcommand{\authnote}[2]{{\bf [{\color{red} #1's Note:} {\color{blue} #2}]}}
\else
 \newcommand{\authnote}[2]{}
\fi

\ifdefined\IsProofRead
\newcommand{\ToPR}[1]{{\bf [{\color{blue} #1}]}}
\else
\newcommand{\ToPR}[1]{}
\fi

\def\endoc{\end{document}}

\newcommand{\sdotfill}{\textcolor[rgb]{0.8,0.8,0.8}{\dotfill}} 

\newenvironment{protocol}{\begin{proto}}{\vspace{-\topsep}\sdotfill\end{proto}}
\newenvironment{algorithm}{\begin{algo}}{\vspace{-\topsep}\sdotfill\end{algo}}

\newcommand{\aka} {also known as \ }
\newcommand{\resp}{resp.,\ }
\newcommand{\ie} {i.e.,\ }
\newcommand{\eg} {e.g.,\ }
\newcommand{\wrt} {with respect to\ }
\newcommand{\cf}{{cf.,\ }}

\newcommand{\ceil}[1]{\left\lceil #1 \right\rceil}

\newcommand{\iprod}[1]{\langle #1 \rangle_2}
\newcommand{\set}[1]{\ens{#1}}

\newcommand{\floor}[1]{\left \lfloor#1 \right \rfloor}

\newcommand{\eqdef}{:=}

\newcommand{\N}{{\mathbb{N}}}

\newcommand{\F}{{\mathbb F}}

\newcommand{\zo}{\set{0,1}}
\newcommand{\zn}{{\zo^n}}
\newcommand{\zl}{{\zo^\ell}}
\newcommand{\zm}{{\zo^m}}
\newcommand{\zs}{{\zo^\ast}}

\newcommand{\getsr}{\mathbin{\stackrel{\mbox{\tiny R}}{\gets}}}
\newcommand{\xor}{\oplus}

\newcommand{\eps}{\varepsilon}

\newcommand{\la}{\gets}

\newcommand{\poly}{\operatorname{poly}}
\newcommand{\polylog}{\operatorname{polylog}}

\newcommand{\GF}{\operatorname{GF}}
\newcommand{\Exp}{\operatorname*{E}}

\newcommand{\Hall}{\operatorname{H}}
\newcommand{\HSh}{\operatorname{H}}

\newcommand{\Hmax}{\operatorname{H_0}}
\newcommand{\Hsam}{\operatorname{H}}

\newcommand{\Hmin}{\operatorname{H_{\infty}}}
\newcommand{\HRen}{\operatorname{H_2}}
\newcommand{\Haccsam}{\operatorname{AccH}}
\newcommand{\Hrealsam}{\operatorname{RealH}}

\newcommand{\negl}{\operatorname{neg}}

\newcommand{\Supp}{\operatorname{Supp}}

\newcommand{\MathFam}[1]{\mathcal{#1}}

\newcommand{\FFam}{\MathFam{F}}

\newcommand{\Dom}{\MathFam{D}}
\newcommand{\Rng}{\MathFam{R}}

\newcommand{\MathAlg}[1]{\mathsf{#1}}

\newcommand{\is}{{i^\ast}}

\newcommand{\Com}{\ensuremath{\MathAlg{Com}}\xspace}

\newcommand{\WB}{{\MathAlg{WH}}}
\newcommand{\SB}{{\MathAlg{SH}}}

\newcommand{\len}{k}

\newcommand{\class}[1]{\mathrm{#1}}

\newcommand{\NP}{\class{NP}}

\newcommand{\SZK}{\class{SZK}}

\renewcommand{\cref}{\Cref}

\newtheorem{theorem}{Theorem}[section]

\newaliascnt{lemma}{theorem}
\newtheorem{lemma}[lemma]{Lemma}
\aliascntresetthe{lemma}
\crefname{lemma}{Lemma}{Lemmas}

\newaliascnt{claim}{theorem}
\newtheorem{claim}[claim]{Claim}
\aliascntresetthe{claim}
\crefname{claim}{Claim}{Claims}

\newaliascnt{corollary}{theorem}

\aliascntresetthe{corollary}
\crefname{corollary}{Corollary}{Corollaries}

\newaliascnt{construction}{theorem}
\newtheorem{construction}[construction]{Construction}
\aliascntresetthe{construction}
\crefname{construction}{Construction}{Constructions}

\newaliascnt{fact}{theorem}

\aliascntresetthe{fact}
\crefname{fact}{Fact}{Facts}

\newaliascnt{proposition}{theorem}
\newtheorem{proposition}[proposition]{Proposition}
\aliascntresetthe{proposition}
\crefname{proposition}{Proposition}{Propositions}

\newaliascnt{conjecture}{theorem}

\aliascntresetthe{conjecture}
\crefname{conjecture}{Conjecture}{Conjectures}

\newaliascnt{definition}{theorem}
\newtheorem{definition}[definition]{Definition}
\aliascntresetthe{definition}
\crefname{definition}{Definition}{Definitions}

\newaliascnt{remark}{theorem}
\newtheorem{remark}[remark]{Remark}
\aliascntresetthe{remark}
\crefname{remark}{Remark}{Remarks}

\newaliascnt{notation}{theorem}

\aliascntresetthe{notation}
\crefname{notation}{Notation}{Notation}

\newaliascnt{proto}{theorem}

\newtheorem{proto}[proto]{Protocol}

\aliascntresetthe{proto}
\crefname{proto}{protocol}{protocols}

\newaliascnt{algo}{theorem}
\newtheorem{algo}[algo]{Algorithm}
\aliascntresetthe{algo}
\crefname{algo}{algorithm}{algorithms}

\newaliascnt{expr}{theorem}

\aliascntresetthe{expr}
\crefname{experiment}{experiment}{experiments}

\newcommand{\Stepref}[1]{Step~\ref{#1}}

\def\FullBox{$\Box$}
\def\qed{\ifmmode\qquad\FullBox\else{\unskip\nobreak\hfil
\penalty50\hskip1em\null\nobreak\hfil\FullBox
\parfillskip=0pt\finalhyphendemerits=0\endgraf}\fi}

\def\qedsketch{\ifmmode\Box\else{\unskip\nobreak\hfil
\penalty50\hskip1em\null\nobreak\hfil$\Box$
\parfillskip=0pt\finalhyphendemerits=0\endgraf}\fi}

\newcommand{\Sa}{\mathsf{S}}
\newcommand{\Ra}{\mathsf{R}}

\newcommand{\Sc}{\Sa}
\newcommand{\Rc}{\Ra}

\newcommand{\cp}{\operatorname{CP}}

\newcommand{\ens}[1]{\left\{#1\right\}}
\newcommand{\size}[1]{\left|#1\right|}

\newcommand{\sindist}{\approx_s}

\newcommand{\view}{\operatorname{view}}

\newcommand{\indist}{\approx}

\newcommand{\Uni}{{\mathord{\mathcal{U}}}}

\newcommand{\prob}[1]{\mathsf{\textsc{#1}}}

\newcommand{\SD}{\prob{SD}}

\newcommand{\I}{\mathcal{I}}
\newcommand{\cl}{\mathcal{L}}
\newcommand{\ck}{\mathcal{K}}

\newcommand{\state}{{\sf state}}
\newcommand{\State}{{\sf State}}
\newcommand{\ppt}{{\sc ppt}\xspace}
\newcommand{\pptm}{{\sc pptm}\xspace}

\newcommand{\h}{{\cal{H}}}

\newcommand{\cJ}{\mathcal{J}}

\newcommand{\kmax}{{k_{\textsc{acc}}}}
\newcommand{\kreal}{{k_{\textsc{real}}}}

\newcommand{\Gc}{\mathsf{G}}
\newcommand{\Gb}{\mathbb{G}}

\newcommand{\Ac}{\mathsf{A}}

\newcommand{\Dc}{\mathsf{D}}

\newcommand{\Inv}{\mathsf{Inv}}

\newcommand{\cs}{{\cal{S}}}

\newcommand{\Tableofcontents}{
\thispagestyle{empty}
\pagenumbering{gobble}
\clearpage
\tableofcontents
\thispagestyle{empty}
\clearpage
\pagenumbering{arabic}
}

\renewcommand{\Pr}{{\mathsf {Pr}}}
\newcommand{\ppr}[2]{\Pr_{#1}\left[#2\right]}
\newcommand{\pr}[1]{\ppr{}{#1}}

\newcommand{\Ex}{{\mathsf E}}
\newcommand{\eex}[2]{\Ex_{#1}\left[#2\right]}
\newcommand{\ex}[1]{\eex{}{#1}}
\newcommand{\wlg} {without loss of generality\xspace}

\newcommand{\vect}[1]{{ \boldsymbol #1}} 

\newcommand{\G}{G}

\newcommand{\brk}[1]{[ #1 ]}
\newcommand{\prn}[1]{(#1)}
\newcommand{\idx}[1]{#1}

\newcommand\Renyi{R\'enyi\xspace}

\newcommand{\cY}{{\cal{Y}}}
\newcommand{\cX}{{\cal{X}}}

\newcommand{\seq}[1]{\langle #1 \rangle}

\newcommand{\vZ}{\textbf{Z}}

\newcommand{\Fail}{\mathsf{Fail}}

\newcommand{\false}{\XMathAlg{false}}
\newcommand{\true}{\XMathAlg{true}}

\newcommand{\XMathAlg}[1]{\MathAlg{#1}\xspace}

\newcommand{\tth}[1]{$#1$'th\xspace}
\newcommand{\ith}{\tth{i}}
\newcommand{\jth}{\tth{j}}

\newcommand{\kld}[2]{\mathsf{D}\bigl({#1 \,\|\, #2} \xspace\bigr)}

\newcommand{\ColFinder}{\MathAlg{ColFinder}}

\newcommand{\Gs}{{\widetilde{\Gc}}}

\newcommand{\tT}{{\widetilde{T}}}
\newcommand{\Ts}{\tT}
\newcommand{\tR}{{\widetilde{R}}}
\newcommand{\tr}{{\widetilde{r}}}
\newcommand{\tY}{{\widetilde{Y}}}

\newcommand{\hT}{{\widehat{T}}}
\newcommand{\hY}{{\widehat{Y}}}

\newcommand{\hR}{{\widehat{R}}}

\newcommand{\Gbs}{{\widetilde{\Gb}}}
\newcommand{\Tbs}{{\widetilde{\mathbb{T}}}}

\newcommand{\vy}{{\vect{y}}}
\newcommand{\vz}{{\vect{z}}}
\newcommand{\vx}{{\vect{x}}}

\newcommand{\vX}{{\vect{X}}}

\newcommand{\vR}{{\vect{R}}}

\newcommand{\vt}{\textbf{t}}
\newcommand{\vT}{T}

\newcommand{\Rs}{{\widetilde{\Rc}}}
\newcommand{\Ss}{{\widetilde{\Sc}}}
\newcommand{\InvG}{\Inv^{\Gs}}

\newcommand{\vY}{{\bf{Y}}}

\newcommand{\ee}{w}

\newcommand{\kmaxp}{{k'_{\textsc{acc}}}}
\newcommand{\krealp}{{k'_{\textsc{real}}}}

\newcommand{\mt}{{\tilde m}}

\newcommand{\Ypl}{Y^{\prn \ee} } 
 
\newcommand{\hYpl}{\widehat{Y}^{\prn \ee}}

\newcommand{\Yl}{\widetilde{Y}}

\newcommand{\tRR}{R}	
\newcommand{\tYY}{Y}		
\newcommand{\vv}{v}

\newcommand{\Sss}{{\widehat{\Sc}}}
\newcommand{\pp}{c}

\newcommand{\Inote}[1]{\authnote{Iftach}{#1}}

\newcommand{\Gbl}{\Gc^{\brk \ee}}

\newcommand{\Gt}{\G^{\seq \vv}}
\newcommand{\RNum}[1]{\uppercase\expandafter{\romannumeral #1\relax}}

\def\Fail{\mathsf{Fail}}

\title{Inaccessible Entropy  I:\\ Inaccessible Entropy Generators and Statistically Hiding Commitments from One-Way Functions\thanks{A preliminary version, with a  different notion of  accessible entropy, appeared in \cite{HaitnerReVaWe09}.}
	\Draft{\\\small \sc Working Draft: Please Do Not Distribute}
   }

\author{Iftach Haitner\thanks{School of Computer Science, Tel Aviv University. E-mail:
 \texttt{iftachh@tauex.tau.ac.il}. Research supported by ERC starting grant 638121 and US-Israel BSF grant 2010196. Member of the  Check Point Institute for Information Security.} \and  Omer Reingold\thanks{Computer Science Department, Stanford University, reingold@stanford.edu. Research supported by US-Israel BSF grant 2006060.} \and Salil Vadhan\thanks{School of Engineering and Applied Sciences and Center for Research on Computation and Society, Harvard University. E-mail: \texttt{salil@eecs.harvard.edu}. Work done in part while visiting U.C. Berkeley, supported by  the Miller Institute for Basic Research in Science and a Guggenheim Fellowship.  Also supported by NSF grant CNS-0831289 and US-Israel BSF grant 2006060.} \and Hoeteck Wee\thanks{ENS, Paris, France. E-mail:
 \texttt{wee@di.ens.fr}. Part of this work was done while a post-doc at Columbia University, supported in part by NSF Grants CNS-0716245 and SBE-0245014.}}

\begin{document}
\begin{titlepage}

\maketitle

\begin{abstract}
	
We put forth a new computational notion of entropy, measuring the
(in)feasibility of sampling high-entropy strings that are consistent with a
given generator. Specifically, the \ith output block of a generator
$\Gc$ has {\em accessible entropy} at most $k$  if the following holds:  when conditioning on its  prior coin tosses, no polynomial-time strategy $\Gs$ can generate valid output  for $\Gc$'s  \ith output block with entropy greater than $k$.  A generator has {\em inaccessible entropy} if the total accessible entropy (summed over the blocks) is noticeably smaller than the real entropy of $\Gc$'s output.

As an application of the above notion, we improve upon the result of \citeauthor*{HaitnerNgOnReVa09} [Sicomp '09],  presenting a much simpler and more efficient construction of statistically hiding commitment schemes from arbitrary one-way functions.
\end{abstract}

\vfill

\thispagestyle{empty}
\end{titlepage}

\Tableofcontents

\sloppy
\section{Introduction}\label{sec:intro}

Computational analogues of information-theoretic notions have given rise to some of the most interesting phenomena in the theory of computation. For example, \emph{computational indistinguishability}, a computational analogue of statistical  indistinguishability introduced by \citet{GoldwasserMi84}, enabled the bypassing of \citet{Shannon67}'s impossibility results on perfectly secure encryption~\cite{Shannon67} and provided the basis for the computational theory of pseudorandomness~\cite{BlumM82,Yao82B}. \emph{Pseudoentropy}, a computational analogue of entropy introduced by \citet*{HastadImLeLu99}, was the key to their fundamental result that established  the equivalence of pseudorandom generators and one-way functions and has become a basic concept in complexity theory and cryptography. The above notions were further refined in  \cite{HaitnerReVa13,VadhanZheng2012}, leading to to much simpler and more efficient constructions of pseudorandom generators based on one-way functions.

In this work, we introduce another computational analogue of entropy  we call {\em accessible entropy}. We use this notion to build a simpler construction of  statistically hiding  and computationally binding  commitment schemes from arbitrary one-way functions, a construction that is significantly  simpler and more efficient than the previous construction of \citet{HaitnerNgOnReVa09}. Before describing accessible entropy (and the complementary notion of {\em inaccessible entropy}), we review the standard information-theoretic notion of entropy and the computational notion of pseudoentropy  \cite{HastadImLeLu99}.

\subsection{Entropy and Pseudoentropy}

Recall that the {\em entropy} of a random variable $X$ is defined to be $\HSh(X)\eqdef \eex{x\getsr X}{\log \frac1 {\pr{X=x}}}$, which measures the number of ``bits of randomness'' in $X$ (on average). We will refer to $\HSh(X)$ as the {\em real entropy} of $X$ to contrast with the computational analogues that we study.

\citet{HastadImLeLu99} define the  {\em pseudoentropy}  of a random variable $X$ to be (at least) $k$ if there exists a random variable $Y$ of entropy (at least) $k$ such that $X$ and $Y$ are computationally indistinguishable. Pseudoentropy is interesting and useful since, assuming one-way  functions exist,   there exist random variables whose pseudoentropy is larger than their real entropy. For example, the output of a pseudorandom generator $\Gc \colon \zo^\ell\mapsto \zn$ on a uniformly random seed has entropy at most $\ell$, but has pseudoentropy $n$ (by definition). \citet{HastadImLeLu99} proved that from {\em any} efficiently samplable distribution $X$ whose pseudoentropy is noticeably larger than its real entropy, it is possible to construct a pseudorandom generator. By showing, in addition, how to construct such a distribution $X$ from any one-way function, \citet{HastadImLeLu99} effectively proved their theorem that the existence of one-way functions implies the existence of pseudorandom generators.

Notions of pseudoentropy as above  are only useful as a \emph{lower} bound on the ``computational entropy'' in a distribution. Indeed, it can be shown that \emph{every} distribution on $\zn$ is computationally indistinguishable from a distribution of entropy at most $\polylog n$.  While several other computational analogues of entropy have been studied in the literature (\cf \cite{BarakShWi}), all are  meant to capture the idea that a distribution ``behaves like'' one of higher entropy. In this paper, we explore a way in which a distribution can ``behave like'' one of much {\em lower} entropy.

\subsection{Inaccessible Entropy}\label{sec:IAE}
We begin by defining accessible entropy, from which we can derive inaccessible entropy. The notion of {\em accessible entropy} presented below is useful as an \emph{upper} bound on computational entropy. We motivate the idea of accessible entropy with an example. Let $\G$ be the following two-block generator:

\begin{algorithm}[Generator $\G$]
	
	\item Let $m\ll n$ and let   $\h  = \set{h\colon \zn\mapsto \zm}$ be  a \textit{family of collision-resistant hash functions}.\footnote{Given $h \getsr\h$, it is infeasible to find distinct $x,x' \in \zn$  with $h(x) = h(x')$.}

	\item On public parameter $h \getsr \h$:
	\begin{enumerate}
		
		\item Sample $x\getsr \zn$.
		
		\item Output $y=h(x)$.
		\item Output $x$.
	\end{enumerate}
\end{algorithm}

Information-theoretically,  conditioned on  $h$ and its first output block $y$, the second output block of $\G$ (namely $x$) has entropy at least $n-m$.  This is since  $(h,y=h(x))$ reveals at most  $m$ bits of information about $x$. The collision-resistance property of $h$, however, implies  that given the {\em state} of $\G$ after it outputs its first block $y$, there is at most one consistent value of $x$ that can be computed in polynomial time  with non-negligible probability. (Otherwise, we would be able to find two distinct messages $x\neq x'$ such that $h(x)=h(x')$.) This holds even if $\G$ is replaced by any polynomial-time cheating strategy $\Gs$. Thus, there is ``real entropy'' in $x$ (conditioned on  $h$ and the first output of $\G$), but it is ``computationally inaccessible'' to $\Gs$, for which the entropy of  $x$ is effectively $0$.

We generalize this basic idea to  consider both the real and accessible entropy accumulated over several blocks of a  generator. Consider an \emph{$m$-block generator} $\G\colon \zn \mapsto (\zs)^m$, and let $(Y_1,\ldots,Y_m)$ be random variables denoting the $m$ output blocks generated by applying $\G$ over randomness $U_n$ (no public parameters are given). We define the {\em real entropy} of $\G$ as $\HSh(G(U_n))$, the Shannon entropy of $\G(U_n)$, which is equal to

$$\sum_{i\in [m]} \HSh(Y_i\mid Y_{< i})$$ 

\noindent
for $\HSh(X\mid Y)=\Exp_{y\getsr Y}[\HSh(X\mid _{Y=y})]$, which is the standard notion of (Shannon) conditional entropy, and $Y_{< i} = Y_1 \ldots, Y_{ i-1}$.

To define {\em accessible entropy}, consider the following  \ppt (probabilistic polynomial-time) cheating algorithm $\Gs$:  before outputting the \ith block, it tosses some fresh random coins $r_i$, and uses them to  output a string $y_i$. We restrict our attention to \emph{$\G$-consistent} (adversarial) generators --- $\Gs$'s  output is always in the support of $\G$ (though it might be distributed differently). Now, let $(R_1,Y_1,\ldots,R_m,Y_m)$ be random variables
corresponding to a random execution of $\Gs$. We define the {\em accessible entropy} achieved by $\Gs$ to be

$$\sum_{i\in [m]} \HSh(Y_i\mid Y_{<i}, R_{<i})= \sum_{i\in [m]} \HSh(Y_i\mid R_{<i}).$$

Namely,  we compute the entropy conditioned not just on the previous output blocks $Y_{<i}$ (which are  determined by $R_{<i}$), as  done when computing the real entropy of $\G$, but also on the local state of $\Gs$ prior to outputting the \ith block (which \wlg equals its coin tosses $R_{<i}$). For a given $\G$, we define its \emph{accessible entropy} as the maximal accessible entropy achieved by a $\G$-consistent, polynomial-time generator $\Gs$.\footnote{The above informal definitions are simplified or restricted compared to our actual definitions in several ways. In particular, for some of our reductions it is beneficial to work with real {\em min-}entropy and/or accessible {\em max-}entropy rather than real and accessible Shannon entropy as defined above, and formulating conditional versions of these measures is a bit more delicate.}  Thus, in contrast to pseudoentropy, accessible entropy is useful for expressing
the idea that the ``computational entropy'' in a distribution is {\em smaller} than its real entropy.
We refer to the difference $\text{(real entropy)}-\text{(accessible entropy)}$ as the {\em inaccessible entropy}
of the protocol.

It is important to note that if we put no  restrictions on the computational power of a $\G$-consistent $\Gs$, then its  accessible entropy can always be as high as the real entropy of $\Gc$; to generate its \ith block $y_i$, $\Gs$ samples $x$ uniformly at random from the set $\set{x'\colon \Gc(x')_1 = y_1,\ldots,\Gc(x')_{i-1} = y_{i-1}}$. This strategy, however, is not always possible for a computationally bounded $\Gs$.

The collision resistance example given earlier provides evidence   that when allowing public parameters,  there are efficient generators whose computationally accessible entropy is much smaller than their real Shannon entropy. Indeed, the real entropy of the generator we considered above is $n$ (namely, the total entropy in $x$), but its accessible entropy is at most $m+\negl(n) \ll n$, where $m$ is the output length of the collision-resistant hash function.

\subsection{Constructing  an Inaccessible Entropy Generator from One-Way Functions}
It turns out that  we do not need to assume collision resistance or use pubic parameters  to have a generator whose real entropy is significantly larger than its accessible entropy; any one-way function can be used to construct an inaccessible entropy generator (without public parameters).  In particular, we prove  the following result:
 
 \begin{theorem}[Inaccessible entropy generator from one-way functions, informal]\label{thm:IAEfromOEFInf}
  For a function $f\colon \zn \mapsto\zn$, let $\Gc$ be the  $(n+1)$-block generator defined by 
\begin{align*}
\Gc(x) = (f(x)_1,f(x)_2,\ldots,f(x)_n,x).
\end{align*}
Assuming $f$ is a one-way function, then the accessible entropy of $\Gc$ is at most  $n-\log n$. Since the real entropy of $\G(U_n)$ is $n$, it follows that $\Gc$   has $\log n$ bits of inaccessible entropy.

\end{theorem}

Interestingly,  the definition of  $\Gc$ used in the above theorem  is the same as the construction of a next-block pseudoentropy generator from a one-way function used by \citet{VadhanZheng2012}, except that we have broken  $f(x)$, other than   $x$, into one-bit blocks.

To show that $\Gc$ has accessible entropy at most $n-\log n$, we show that the existence of an efficient algorithm whose accessible entropy is too high yields an efficient inverter for $f$ (contradicting its one-wayness). Assume for simplicity that $f$ is a permutation and that there exists an efficient algorithm whose accessible entropy is as high as the real entropy of $\Gc$ (\ie $n$). Let $Y_i$ denote the $i$'th output block of $\Gs$ and let $R_i$ be the coins used by $\Gs$ to produce $Y_i$. By assumption, $\sum_{i\in [n]}\HSh(Y_i\mid R_{<i}) = \HSh(f(U_n)) = n$. Since $\HSh(Y_i\mid R_{<i}) \leq 1$ for every $i\in [n]$, it follows that $\HSh(Y_i\mid R_{<i}) =1$ for every $i\in [n]$.

The above observation yields the following strategy to invert $f$ on input $y$: for $i=1$ to $n$, keep sampling values for $R_{i}$ until the induced value of $Y_i$ is equal to $y_i$. The high accessible entropy of each output block of $\Gs$  yields that the expected number of samples per output bit is two. Therefore, with high probability, after a polynomial number of samples the above process terminates successfully, making the first $n$ output bits of $\Gs$ to be equal to $y$. This implies an efficient inverter for $f$, since when the above process is done, the definition of $\Gs$ yields that its ``justification'' for  the $n$'th output block is the preimage of $y$.

\subsection{Statistically Hiding Commitment from One-way Functions}
As an application of the above accessible entropy notion, we present  a much simpler and more efficient construction of statistically hiding commitment schemes from arbitrary one-way functions.   Our construction builds such a commitment from a generator with noticeable  accessible entropy whose existence is guaranteed by \cref{thm:IAEfromOEFInf}. The resulting scheme  conceptually unifies the construction of statistically \emph{hiding} commitments from one-way functions with the construction of statistically \emph{binding} commitments from one-way functions (the latter being due to \cite{HastadImLeLu99,Naor91}): the first step of both
constructions is to obtain a gap between real entropy and ``computational entropy'' (pseudoentropy in the case of statistical binding and accessible entropy in the case of statistical hiding), which is then amplified by repetitions and finally combined with various forms of hashing.

We start by describing the above mentioned commitment schemes, and then explain how to construct them from an inaccessible entropy generator. 

\subsubsection{Commitment Schemes}

A {\em commitment scheme} is the cryptographic analogue of a safe. It is a two-party protocol between a {\em sender} $\Sc$ and a {\em receiver} $\Rc$ that consists of two stages. The {\em commit stage} corresponds to putting an object in a safe and locking it. In this stage, the sender ``commits'' to a private message $m$. The {\em reveal stage} corresponds to unlocking and opening the safe. In this stage, the sender ``reveals'' the message $m$ and ``proves'' that it was the value committed to in the commit stage (without loss of generality, by revealing coin tosses consistent with $m$ and the transcript of the commit stage).

Commitment schemes have two security properties. The {\em hiding} property informally states that at the end of the commit stage, an adversarial receiver has learned nothing about the message $m$, except with negligible probability. The {\em binding} property states that after the commit stage, an adversarial sender cannot output valid openings for two distinct messages, except with negligible probability. Both of these security properties come in two flavors --- {\em statistical}, where we require security even against a computationally unbounded adversary, and {\em computational}, where we only require security against feasible (e.g.
polynomial-time) adversaries.

Statistical security is preferable to computational security, but it is easy to see that commitment schemes that are both statistically hiding
and statistically binding do not exist.  Instead we have to settle for one of the two properties being statistical and the other being computational.
Statistically binding (and computationally hiding) commitments have been well understood for a long time. Indeed, \citet{Naor91} showed
how to build a two-message statistically binding commitment using any pseudorandom generator;  thus, in combination with
the construction of pseudorandom generators from any one-way function \cite{HastadImLeLu99}, we obtain two-message statistically binding commitments
from the minimal assumption that one-way functions exist.

In contrast, our understanding of {\em statistically hiding} commitments has lagged behind.   \citet{HaitnerNgOnReVa09} have shown that statistically hiding commitment schemes can be constructed from any one-way function. Their construction, however, is very complicated and inefficient.  In this paper, we show that these two types of commitments are closely connected to the notion of inaccessible entropy, that is, with protocols having a gap between real entropy and accessible entropy.

Consider a statistically hiding commitment scheme
in which the sender commits to a message of length $k$, and suppose we run the protocol with the message $m$ chosen uniformly at random
in $\zo^k$. Then, by the statistical hiding property, the {\em real entropy} of the message $m$ after the commit stage is $k-\negl(n)$.
On the other hand, the computational binding property states that the {\em accessible entropy} of $m$ after the commit stage is at most $\negl(n)$. Our main technical contribution is the converse to the above observation.

\subsubsection{Statistically Hiding Commitment Schemes from  Inaccessible Entropy  Generators}\label{sec:intro:SHCfroAE}

\begin{theorem}[inaccessible entropy to commitment, informal]\label{thm:gap-to-commit-intro}
	Assume there exists an $m$-block efficient generator with real entropy $n$ and accessible entropy at most $k(1- \delta)$. Then there exists a   $\Theta(m^2/\delta^2)$-round  statistically hiding commitment scheme.  Moreover, if we allow the protocol to use nonuniform advice,  the round complexity is reduced to  $\Theta(m/\delta)$. 
\end{theorem}

We prove \cref{thm:gap-to-commit-intro} via few modular steps:
\begin{description}
	\item[Entropy equalization] Using sequential repetition with a ``random offset'', we convert the generator into one for which we  {\em know} the real entropy in each block  (rather than just knowing the total entropy) and there remains a noticeable gap between the real entropy and the accessible entropy. 
	
	\item[Gap amplification] Taking the direct product of the generator (\ie invoking it many times in  parallel)   has the effect of (a) converting the real entropy to real {\em min-}entropy, and (b) amplifying the gap between the real entropy and accessible entropy.
	
	\item[Constructing the commitment scheme] By applying a constant-round hashing protocol in each round (based on the interactive hashing protocol of \cite{DingHRS04} and universal one-way hash functions \cite{NaorYu89,Rompel90}), we obtain a receiver public-coin {\em weakly binding} statistically  hiding commitment scheme.  The  latter commitment is then  amplified into a full-fledged commitment  using  parallel repetition.

\end{description}

\paragraph{Statistically hiding commitments from one-way functions.}
Combining \cref{thm:IAEfromOEFInf,thm:gap-to-commit-intro} yields that the minimal assumption that one-way functions exist  implies the existence of statistically hiding commitment schemes,  reproving \cite{HaitnerNgOnReVa09}.

\begin{theorem}[Statistically hiding commitments from one-way functions, informal]\label{thm:oef-to-commit-intro}
Assume one-way functions exist.  Then there exists an  $\Theta(n^2/\log n^2)$-round statistically hiding commitment scheme, where $n$ is the input length of the one-way function.  Moreover, if we allow the protocol to use nonuniform advice, the round complexity is reduced to  obtain $\Theta(n/\log n)$. 
\end{theorem}
The above improves upon the protocol  of \cite{HaitnerNgOnReVa09},  which  has a large unspecified polynomial number of rounds, and the nonuniform variant meets the  lower bound of \cite{HaitnerHRS15} for ``black-box constructions'' (such as the one we used to prove the theorem).

\subsection{Related Work}\label{sec:relatedWork}
A preliminary version of  \citet{HaitnerReVaWe09} uses a more general, and more complicated, notion of accessible entropy in which the accessible entropy of \emph{protocols} rather than generators is measured.   This latter notion is used in that paper  to show that if $\NP$ has constant-round interactive proofs that are black-box zero knowledge under parallel composition,  then there exist constant-round statistically hiding commitment schemes. A subsequent work of \citet{HaitnerHolReVaWe10} uses a simplified version of accessible entropy to present a simpler and more efficient construction of \emph{universal one-way functions} from any one-way function.  One of the two  inaccessible entropy generators considered in \cite{HaitnerHolReVaWe10}, they considered for this construction, is very similar to the  construction of inaccessible entropy  generators discussed above. The notion of inaccessible entropy, of the simpler variant appearing in   that work, is in a sense implicit in the work of \citet{Rompel90}, who first  showed  how to base universal one-way functions on any one-way function. Very recently, \citet{BitanskyHKY18} used the reduction given in this paper from an inaccessible entropy generator to statistically hiding commitment, to construct constant-round statistically hiding commitments from \textit{distributional collision resistance hash functions}, a relaxation of collision resistance hash functions known to exist assuming  average hardness  of the class $\SZK$. Finally, a simplified presentation  of the definitions and reductions appearing in this paper can be found in  \cite{HaitnerVadhan17}.

\subsection*{Paper Organization}
Standard  notations and definitions are given in \cref{sec:prelim};  we also give there several useful facts about the conditional entropy of sequences of random variables. Formal definitions of the  real and accessible entropy of a generator are given in \cref{sec:realvsacc}.  In \cref{sec:IAEGenFromOWF}, we show how to construct inaccessible entropy generators from one-way functions.  In \cref{sec:manipulatingAE},  we  develop tools  to manipulate the real and accessible entropy of such generators. In \cref{sec:GapToCom}, we construct a statistically hiding  commitment scheme from a generator that has a (noticeable) gap between its real and accessible entropy, and then use this reduction together with the results of \cref{sec:IAEGenFromOWF,sec:manipulatingAE} to construct a statistically hiding  commitment scheme from one-way functions. Finally, in \cref{sec:ATEtoOWF} we prove that the existence of inaccessible entropy generators implies that of one-way functions (namely, we prove the converse of the main result of \cref{sec:IAEGenFromOWF}).

\section{Preliminaries}\label{sec:prelim}

\subsection{Notation}\label{sec:prelim:notation}
We use calligraphic letters to denote sets, upper-case for random variables, lower-case for values,  bold-face for vectors, and sans serif for algorithms (\ie Turing machines).  For $n\in\N$, let $[n]= \set{1,\ldots,n}$. For vector $\vy= (y_1,\ldots,y_n)$ and $\cJ \subseteq [n]$, let $\vy_\cJ = (y_{i_1},\ldots, y_{i_{\size{\cJ}}})$, where $i_1<\ldots <i_{\size{\cJ}}$ are the elements of  $\cJ$.  Let $\vy_{<j} = \vy_{[j-1]} = (y_1,\ldots,y_{j-1})$ and $\vy_{\le j} = \vy_{[j]} = (y_1,\ldots,y_j)$. Both notations naturally extend to an ordered list of elements that is embedded in  a larger vector (\ie given $(a_1,b_1,\ldots,a_n,b_n$), $a_{<3}$ refers to the vector $(a_1,a_2)$). Let $\poly$ denote the set of all positive polynomials.  A function $\nu \colon \N \mapsto [0,1]$ is \textit{negligible}, denoted $\nu(n) = \negl(n)$, if $\nu(n)<1/p(n)$ for every  $p\in\poly$ and large enough $n$. 


\subsection{Random Variables}
Let $X$ and $Y$ be random variables taking values in a discrete universe $\Uni$.
We adopt the convention that, when the same random variable appears multiple times in an expression, all occurrences refer to the same instantiation. For example, $\Pr[X=X]$ is 1.
For an event $E$, we write $X|_E$ to denote the random variable $X$ conditioned on $E$. We let $\ppr{X|Y}{x|y}$ stand for $\pr{X=x\mid Y=y}$. The {\em support} of a random variable $X$, denoted $\Supp(X)$, is defined as $\set{ x \colon \Pr[X=x]> 0}$.  Let $U_n$ denote a random variable that is uniform over $\zn$.  For $t\in \N$, let $X^{\prn{t}} =(X^{\idx{1}},\ldots,X^{\idx{t}})$, where $X^{\idx 1},\ldots,X^{\idx t}$ are independent copies of $X$.  


We write $X\equiv Y$ to indicate that $X$ and $Y$ are identically distributed. We write $\Delta(X,Y)$ to denote the {\em statistical difference} (\aka variation distance) between $X$ and $Y$, \ie
$$\Delta(X,Y) = \max_{T\subseteq \Uni} \size{\Pr[X\in T]-\Pr[Y\in T]}.$$
If $\Delta(X,Y)\leq \eps$ [\resp $\Delta(X,Y) > \eps$], we say that $X$ and $Y$ are {\em $\eps$-close} [\resp $\eps$-far]. Two random variables $X = X(n) $ and $Y = Y(n)$ are \emph{statistically  indistinguishable}, denoted $X \sindist Y$, if for any unbounded algorithm $\Dc$, it holds that $\size{\Pr[\Dc(1^n,X(n)) = 1] - \Pr[\Dc(1^n,Y(n)) = 1]} = \negl(n)$.\footnote{This is equivalent to to requiring  that $\Delta(X(n),Y(n)) = \negl(n)$.}  Similarly, $X$ and $Y$ are   computationally  indistinguishable, denoted   $X \indist Y$], if $\size{\Pr[\Dc(1^n,X(n)) = 1] - \Pr[\Dc(1^n,Y(n)) = 1]} = \negl(n)$ for every \ppt $\Dc$.
 
The {\sf KL-divergence} (\aka {\sf ~Kullback-Leibler divergence} and {\sf relative
	entropy}) between two distributions $P,Q$ over a discrete domain  $\cX$
is defined by
\begin{align*}
\kld{P}{Q} \eqdef \sum_{x\in\cX}P(x)\log\frac{P(x)}{Q(x)} = \Ex_{x \getsr P}\log\frac{P(x)}{Q(x)},
\end{align*}
letting $0\cdot\log\frac00 = 0$, and if there exists $x\in\cX$ such that
$P(x)>0=Q(x)$ then $\kld{P}{Q}\eqdef \infty$.

\subsection{Entropy Measures}\label{sec:measures}
We  refer to several measures of entropy. The relation and motivation of these measures is best understood by considering a notion that we will refer to as the sample-entropy: for a random variable $X$ and $x\in \Supp(X)$, the \emph{sample-entropy} of $x$ \wrt $X$ is the quantity  
$$\Hsam_X(x) \eqdef \log \tfrac 1{\Pr[X=x]},$$
letting $\Hsam_X(x) = \infty$ for $x\notin \Supp(X)$, and  $2^{-\infty} = 0$.

The sample-entropy measures the amount of ``randomness" or ``surprise" in the specific sample $x$, assuming that $x$ has been generated according to $X$. Using this notion, we can define the {\em Shannon entropy} $\HSh(X)$ and {\em min-entropy} $\Hmin(X)$ as follows:
\begin{align*}
\HSh(X) &\eqdef \Exp_{x\getsr X}[\Hsam_X(x)],\\
\Hmin(X) &\eqdef \min_{x\in \Supp(X)} \Hsam_X(x).
\end{align*}

The \emph{collision probability} of $X$ is defined by 
$$\cp(X) \eqdef \sum_{x\in \Supp(X)} \ppr{X}{x}^2 = \ppr{(x,x') \getsr X^2}{x= x'}$$
 and its  {\em \Renyi entropy} is defined by 
\begin{align*}
\HRen(X) \eqdef - \log \cp(X).
\end{align*}
We will also discuss the {\em max-entropy} $\Hmax(X) \eqdef \log |\Supp(X)|$. The term ``max-entropy'' and its relation to the sample-entropy will be made apparent below.

It can be shown that $\Hmin(X) \leq \HRen(X) \leq \HSh(X) \leq \Hmax(X)$ with each inequality being an equality if and only if $X$ is flat (uniform on its support). Thus, saying that $\Hmin(X)\geq k$ is a strong way of saying that $X$ has ``high entropy'' and $\Hmax(X)\leq k$ a strong way of saying that $X$ has ``low entropy''. 

The following fact quantifies the probability  that the sample-entropy is larger than the max-entropy.
\begin{lemma}\label{prop:ShanonoToSample}
	For   random variable $X$ it holds that 
	\begin{enumerate}
		
		\item $\eex{x\getsr X}{2^{\Hall_X(x)}} = \size{\Supp(X)}$.
		
		\item $\ppr{x \getsr  X}{\HSh_X(x) > \log \frac1\eps + \Hmax(X)} <\eps$, for any $\eps >0$.
	\end{enumerate}
\end{lemma}
\begin{proof}
	For the first item, compute 
	\begin{align*}
	\eex{x\getsr X}{2^{\Hall_X(x)}}&= \sum_{x\in \Supp(X)} 2^{-\Hall_X(x)} \cdot 2^{\Hall_X(x)}\\
	&=\sum_{x\in \Supp(X)} 1\\
	& = \size{\Supp(X)}. 
	\end{align*}
	The second item follows by the first item and Markov inequality. 
	\begin{align*}
	\ppr{x \getsr  X}{\HSh_X(x) > \log \frac1\eps + \Hmax(X)} &= \ppr{x \getsr  X}{2^{\HSh_X(x)} > \frac 1\eps \cdot  \size{\Supp(X)}}\\
	& < \eps.
	\end{align*}
\end{proof}

The following fact quantifies the contribution of unlikely events with high sample-entropy to the overall entropy of a random variable. 

\begin{lemma}\label{prop:HighAEContribution}
	Let  $X$ be a random variable with $\ppr{x\getsr X}{\HSh_X(x) > k} \le \eps \in [0,1]$. Then $\Hall(X) \le  (1 - \eps)k + \eps  \cdot (\Hmax(X) - \log \eps) \le k + \eps \cdot \Hmax(X) +1$.
\end{lemma}
\begin{proof}
Let $\cY = \set{x\in \Supp(X) \colon \HSh_X(x) > k}$ and 	let $Y = X|_{X \in \cY}$. Note that
\begin{align}
\Hall(Y) = \eex{y\gets Y}{\HSh_Y(y)} =  \eex{y\gets Y}{\HSh_X(y)} + \log \ppr{x\getsr X}{\HSh_X(x) > k}.
\end{align}
Let  $\eps'= \ppr{x\getsr X}{\HSh_X(x) > k}$ (hence, $ \eps' \le \eps$). We conclude that 
\begin{align*}
\Hall(X)  &= \sum_{x\in \Supp(X) \setminus \cY} \pr{X=x} \cdot \Hall_X(x)  + \sum_{x\in \cY} \pr{X=x} \cdot \Hall_X(x) \\
&\le (1 - \eps')k + \eps'  \cdot \sum_{x\in \cY} \pr{Y=x} \cdot \Hall_X(x) \\
&=  (1 - \eps')k + \eps'  \cdot (\Hall(Y) - \log \eps')\\
&\le (1 - \eps')k + \eps'  \cdot (\Hmax(X) - \log \eps')\\
&\le (1 - \eps)k + \eps  \cdot (\Hmax(X) -  \log \eps)\\
&\le k + \eps \cdot \Hmax(X) +1.
\end{align*}
\end{proof}

\paragraph{Conditional entropies.}
We will also be interested in conditional versions of entropy. For jointly distributed random variables $(X,Y)$ and $(x,y)\in \Supp(X,Y)$, we define
the {\em conditional sample-entropy} to be $\Hsam_{X| Y}(x| y) = \log \tfrac 1{\ppr{X|Y}{x|y}}=\log \tfrac 1{\Pr[X=x\mid Y=y]}$.
Then the standard {\em conditional Shannon entropy} can be written as
$$\HSh(X\mid Y) = \Exp_{(x,y)\getsr (X,Y)} \left[\Hsam_{X\mid Y}(x\mid y)\right]
= \Exp_{y\getsr Y} \left[\HSh(X|_{Y=y})\right] = \HSh(X,Y)-\HSh(Y).$$

The following known lemma  states that  conditioning on a ``short" variable is unlikely to change the sample-entropy significantly.\footnote{We could save a few bits, as compared to the result below, by considering the \emph{average} sample-entropy induced by the conditioning, and in particular  the \textit{average min-entropy} of the variable \cite{DodisORS2008}. Doing so, however, will only reduce the running  time and communication size of our commitment scheme by a constant, so we preferred to stay with the simpler statement  below.}
\begin{lemma}\label{prop:CondNotReduce}
	Let $X$ and $Y$ be random variables, let $k = \Hmin(X)$, and let $\ell = \Hmax(Y)$. Then, for any  $t>0$, it holds that
	$$\ppr{(x,y)\getsr(X,Y)}{\Hall_{X|Y}(x|y) <  k - \ell -t }< 2^{-t}.$$
\end{lemma}
\begin{proof}
	For $y\in \Supp(Y)$, let $\cX_y = \set{x\in \Supp(X) \colon \Hall_{X|Y}(x|y) <  k - \ell -t }$. We have $\size{\cX_y} < 2^{k - \ell -t }$. Hence, $\size{\cX = \bigcup_{y\in \Supp(Y)} \cX_y} < 2^\ell \cdot  2^{k - \ell -t } =  2^{k -t }$. It follows that
	\begin{align*}
	\ppr{(x,y)\getsr(X,Y)}{\Hall_{X|Y}(x|y) <  k - \ell -t } \le \ppr{(x,y)\getsr(X,Y)}{x\in \cX} < 2^{-k} \cdot 2^{k -t } = 2^{-t}.
	\end{align*}
\end{proof}

\paragraph{Smoothed entropies.}
The following lemma will allow us to  think of a  random variable $X$ whose sample-entropy is high, with high probability, as if it has high min-entropy (\ie as if its sample-entropy function is ``smoother'', with no peaks). 

\begin{lemma}\label{lemma:smoothEntropies}~
	Let $X$ and $Y$ be random variables and let $\eps >0$.
	\begin{enumerate}
		\item Suppose $\ppr{x\getsr X}{\Hsam_X(x) \ge k} \geq 1-\eps$. Then $X$ is $\eps$-close to a random variable $X'$ with $\Hmin(X')\geq k$.
		
		\item Suppose $\ppr{(x,y)\getsr (X,Y)}{\Hsam_{X|Y}(x|y) \ge k} \geq 1-\eps$. Then $(X,Y)$ is $\eps$-close to a random variable $(X',Y')$ with  $\Hsam_{X'|Y'}(x|y) \ge k$ for any $(x,y) \in \Supp(X',Y')$. Further, $Y'$ and $Y$ are identically distributed.  
	\end{enumerate}
\end{lemma}
\begin{proof}
	For the first item, we modify $X$ on an $\eps$ fraction of the probability space (corresponding to when $X$ takes on a value $x$ such that
	$\Hsam_X(x)\geq k$)  to bring all probabilities to be smaller  than or equal to $2^{-k}$.
	
	The second item is proved via similar means, while when changing $(X,Y)$, we do so  without changing the ``$Y$'' coordinate.  
\end{proof}

\paragraph{Flattening Shannon entropy.} It is well known that the Shannon entropy of a random variable can be converted to min-entropy (up to small statistical distance) by taking independent copies of this variable.

\begin{lemma}[\cite{Yang15}, Theorem 3.14]\label{lem:flattening}
	Let $X$ be a random variable taking values in a universe $\Uni$, let $t\in \N$, and let $0 < \eps \le 1/e^2$. Then with probability at least $1-\eps$ over $x\getsr X^{\prn{t}}$,
	$$\Hsam_{X^{\prn{t}}}(x) - t \cdot \Hall(X) \geq - O\left(\sqrt{t\cdot \log \tfrac1\eps} \cdot \log (|\Uni|\cdot t)\right).$$
\end{lemma}
We will make use of the following ``conditional variant''  of \cref{lem:flattening}:  
\begin{lemma}\label{lem:flatteningCond}
	Let $X$ and $Y$ be jointly distributed random variables where $X$ takes values in a universe $\Uni$, let $t\in \N$, and  let $0 < \eps \le 1/e^2$. Then with probability at least $1-\eps$  over $(x,y)\gets (X',Y') = (X,Y)^{\prn{t}}$,
	$$\Hsam_{X'|Y'}(x\mid y) - t\cdot \HSh(X\mid Y)  \ge - O\left(\sqrt{t\cdot \log \tfrac1\eps} \cdot \log (|\Uni|\cdot t)\right).$$
\end{lemma}
\begin{proof}
Follows the same line  as the proof  of \cref{lem:flattening}, by considering the random variable $\Hall_{X|Y}(X|Y)$ instead of $\Hall_{X}(X)$.  
\end{proof}

\paragraph{Sequence of random variables.}
For  measuring the real and accessible entropy of a generator, we will measure  the entropy of  a subset of random variables, conditioned on the previous elements in the sequence,  
The following lemma generalizes  \cref{prop:HighAEContribution} to settings that come up  naturally when measuring   the accessible entropy of a generator (as we do in \cref{sec:manipulatingAE}).
\begin{definition}\label{def:HallSum}
	For a $t$-tuple random variable $\vX = (X_1,\ldots,X_t)$, $\vx \in \Supp(X)$ and $\cJ \subseteq [t]$,  let
	$$\Hall_{\vX,\cJ}(\vx) = \sum_{i\in \cJ} \Hall_{\vX_i| \vX_{<i}}(\vx_i | \vx_{<i}).$$
\end{definition}
The following fact is immediate by the chain rule.
\begin{proposition}\label{prop:ShanonoToSampleSeq}
	Let $\vX= (X_1,\ldots,X_t)$  be a sequence of random variables and let $\cJ \subseteq [t]$. Then, $\eex{\vx\getsr \vX}{\Hall_{\vX,\cJ}(\vx)} =  \sum_{j\in \cJ} \Hall(X_j | \vX_{<j}) \le \Hmax(\vX_\cJ)$.
\end{proposition}
\begin{proof}
	Compute
	\begin{align*}
	\eex{\vx\getsr \vX}{\Hall_{\vX,\cJ}(\vx)} &= \eex{\vx\getsr \vX} {\sum_{j\in \cJ}  \Hsam_{X_j \mid \vX_{<j}}(\vx_j | \vx_{<j})}\\
	&= \sum_{j\in \cJ}   \eex{\vx\getsr \vX} {\Hsam_{X_j \mid \vX_{<j}}(\vx_j | \vx_{<j})}\\
	&= \sum_{j\in \cJ} \HSh(X_j|\vX_{<j})\\
	&\le \sum_{j\in \cJ} \Hmax(X_j) \le  \Hmax(\vX_\cJ).
	\end{align*}
\end{proof}

The following lemma generalizes  \cref{prop:HighAEContribution} to such a sequence of random variables.
\begin{lemma}\label{prop:HighAEContributionSeq}
	Let $\vX= (X_1,\ldots,X_t)$  be a sequence of random variables and let $\cJ \subseteq [t]$.  If  $\ppr{\vx\getsr \vX}{\Hall_{\vX,\cJ}(\vx)  > k} \le \eps \in [0,1]$, then $\eex{\vx\getsr \vX}{\Hall_{\vX,\cJ}(\vx)} \le (1 - \eps)k + \eps  \cdot (\Hmax(\vX_{\cJ}) - \log 1/\eps) \le k +\eps \cdot  \Hmax(\vX_\cJ) + 1$.
\end{lemma}
\begin{proof}
	The proof is similar to that of \cref{prop:HighAEContribution}. Let $\cY = \set{\vx\in \Supp(\vX) \colon\Hall_{\vX,\cJ}(\vx)  > k }$, let $\vY = \vX|_{\vX\in \cY}$, and  for $\vy_{\le k} \in \Supp(\vY_{< j})$ let $\eps_{\vy_{\le j}} =  \frac{\pr{\vY_j = \vy_j \mid \vY_{< j} = \vy_{< j}}}{\pr{\vX_j = \vy_j \mid \vX_{< j} = \vy_{< j}}}$.  Compute
	
\remove{	
	The distribution   $\vY$ can be  generated by the following process: sample $\vx \getsr  \vX|_{\vX\in \cY}$ and set $Y_1 = \vx_1$, then sample  $\vx \getsr  \vX|_{\vX\in \cY, \vX_1 = Y_1}$ and set $Y_2 = \vx_2$, and so on. It follows that  
	for every $\vy \in \Supp(\vY_{\le j})$: 
		\begin{align*}
	\lefteqn{\pr{\vX_j = \vy_j \mid \vX_{< j} = \vy_{< j}}}\\
	  &\ge \pr{\vX \in \Supp(\vY) \mid X_{< j} = \vy_{ <j} } \cdot \pr{\vX_j = \vy_j \mid \vX_{< j} = \vy_{< j} \land \vX \in \Supp(\vY)}\\
		&= \eps_{\vy_{<j}} \cdot \pr{\vY_j = \vy_j \mid \vY_{< j} = \vy_{< j}}
			\end{align*}
	and  thus
		\begin{align}
	\eps_{\vy_{<j}}  \le  \frac{\pr{\vX_j = \vy_j \mid \vX_{< j} = \vy_{< j}}}{\pr{\vY_j = \vy_j \mid \vY_{< j} = \vy_{< j}}}
	\end{align}
}

	\begin{align}
	\eex{\vy\getsr \vY}{\Hall_{\vY,\cJ}(\vy)}&= \eex{\vy\getsr \vY}{ \sum_{j\in \cJ} \Hall_{\vY_i| \vY_{<j}}(\vy_j | \vy_{<j})}\nonumber\\
	&= \eex{\vy\getsr \vY}{ \sum_{j\in \cJ} \Hall_{\vX_j| \vX_{<j}}(\vy_j | \vy_{<j}) + \log  \tfrac{\pr{\vY_j = \vy_j \mid \vY_{< j} = \vy_{< j}}}{\pr{\vX_j = \vy_j \mid \vX_{< j} = \vy_{\le j}}}}\nonumber\\
	&=  \eex{\vy\getsr \vY}{ \sum_{j\in \cJ} \Hall_{\vX_j| \vX_{<j}}(\vy_j | \vy_{<j}) + \log \eps_{\vy_{\le j}} } \nonumber \\
	&= \eex{\vy\getsr \vY}{ \sum_{j\in \cJ} \Hall_{\vX_j| \vX_{<j}}(\vy_j | \vy_{<j})}  + \eex{\vy\getsr \vY}{ \sum_{j\in \cJ}  \log \eps_{\vy_{\le j}} },
	\end{align}
and note that
	\begin{align}
	\eex{\vy\getsr \vY}{ \sum_{j\in \cJ}  \log \eps_{\vy_{\le j}} } &\ge 	\eex{\vy\getsr \vY}{ \sum_{j \in [t]}  \log \eps_{\vy_{\le j}} }\\
	& = \eex{\vy\getsr \vY}{ \log  \prod_j \eps_{\vy_{ \le j}} }\nonumber\\
	& = \eex{\vy\getsr \vY}{ \log \frac{\pr{\vY = \vy}}{\pr{\vX = \vy}}}\nonumber\\
	&= \log 1/\pr{\vX \in \cY}. \nonumber
	\end{align}
	Let  $\eps'= \pr{\vX \in \cY} $.  We conclude that
	\begin{align*}
	\eex{\vx\getsr \vX}{\Hall_{\vY,\cJ}(\vx)}  &= \sum_{\vx\in \Supp(\vX) \setminus \cY} \pr{\vX=\vx} \cdot \Hall_{\vX,\cJ}(\vx) + \sum_{\vx\in \cY} \pr{\vX=\vx} \cdot \Hall_{\vX,\cJ}(\vx)\\
	&\le (1 - \eps')k + \eps'  \cdot  \sum_{\vx\in \cY} \pr{\vY=\vx} \cdot \Hall_{\vX,\cJ}(\vx)\\
	&\le  (1 - \eps')k + \eps'  \cdot (	\eex{\vy\getsr \vY}{\Hall_{\vY,\cJ}(\vy)}  - \log 1/\eps')\\
	&\le (1 - \eps')k + \eps'  \cdot (\Hmax(\vX_{\cJ}) - \log 1/\eps')\\
	&\le (1 - \eps)k + \eps  \cdot (\Hmax(\vX_{\cJ}) - \log 1/\eps)\\
	&\le k +\eps \cdot  \Hmax(\vX_\cJ) + 1.
	\end{align*}
\end{proof}

The next two lemmas are only needed when measuring the \emph{max} accessible entropy of a generator, as we do  in \cref{sec:MaxAE},  and are not used in the main body of the paper. The first lemma generalizes \cref{prop:ShanonoToSample} to  sequence of random variables.

\begin{lemma}\label{prop:ShanonoToSampleSeqAdv}
	Let $\vX= (X_1,\ldots,X_t)$  be a sequence of random variables and let $\cJ \subseteq [t]$. Then, 
	\begin{enumerate}
		\item $\eex{\vx\getsr \vX}{2^{\Hall_{\vX,\cJ}(\vx) }} \le  \Hmax(\vX_\cJ)$.
		
		\item $\ppr{x \getsr  X}{\Hall_{\vX,\cJ}(\vx) > \log \frac1\eps + \Hmax(X_\cJ)} <\eps$, for any $\eps >0$.
	\end{enumerate}
\end{lemma}
\begin{proof}

	The second item follows from the first one as in the proof of \cref{prop:ShanonoToSample}. We prove the first item by induction on $t$ and $\size{J}$. The case $t=1$ is immediate, so we assume for all $(t',\cJ')$ with $(t',\size{\cJ'}) < (t,\size{\cJ})$ and prove it for $(t,\cJ)$. Assume that $1 \in \cJ$ (the case $1\notin \cJ$ is analogous) and let $\vX_{-1} = (X_2,\ldots,X_t)$ and  $\cJ_{-1}= \set{i-1\colon i\in \cJ\setminus\set{1}}$. Compute 
	\begin{align*}
	\eex{\vx\getsr \vX}{2^{\Hall_{\vX,\cJ}(\vx) }} &= \sum_{x_1\in \Supp(X_1)} 2^{-\Hall_{X_1}(\vx_1)} \cdot 2^{\Hall_{X_1}(x_1)} \cdot \eex{\vx \getsr \vX_{-1}|_{X_1 =x_1}}{2^{\Hall_{\vX_{-1}|_{X_1=x_1},\cJ_{-1}}(\vx)}}\\
	&\le  \sum_{x_1\in \Supp(X_1)}   1\cdot \size{\Supp((\vX_{-1})_{\cJ_{-1}}|_{X_1 = x_1})}\\
	&=  \sum_{x_1\in \Supp(X_1)}   \size{\Supp(\vX_{\cJ \setminus \set{1}}|_{X_1 = x_1})}\\
	&= \size{\Supp(\vX_\cJ)}.
	\end{align*}
\end{proof}

\paragraph{Sub-additivity.}
The chain rule for Shannon entropy yields that 
$$\HSh(X= (X_1,\ldots,X_t)) = \sum_i \HSh(X_i | X_1,\ldots,X_{i-1}) \le \sum_i \HSh(X_i).$$
The following lemma shows that a variant of the above also holds for sample-entropy.

\begin{lemma}\label{prop:subaddativiy2}
	For random variables  $\vX=(X_1,\ldots,X_t)$, it holds that
	\begin{enumerate}
		\item $\eex{\vx \getsr \vX}{2^{\Hall_\vX(\vx) - \sum_t \Hall_{X_i}(\vx_i)}} \le 1$, and
		
		\item $\ppr{\vx \getsr  \vX}{H_\vX(\vx) > \log \frac1\eps + \sum_{i\in [t]} H_{X_i}(\vx_i)} <\eps$, for any  $\eps >0$.
		
	\end{enumerate}
\end{lemma}
\begin{proof}
	As in \cref{prop:ShanonoToSample}, the second part follows from the first by Markov's inequality. For the first part, compute
	\begin{align*}
	\eex{\vx \getsr \vX}{2^{\Hall_\vX(\vx) - \sum_t \Hall_{X_i}(\vx_i)}}&= \sum_{\vx \in \Supp(\vX)} \pr{\vX =\vx} \cdot \frac{\prod_{i\in [t]} \pr{X_i =\vx_i}}{\pr{\vX =\vx} }\\
	& = \sum_{\vx \in \Supp(\vX)}  \prod_i \pr{X_i =\vx_i}\\
	&\le 1.
	\end{align*}
\end{proof}

\subsection{Hashing}\label{sec:FunctionFamilies}
We will use two  types of  (combinatorial) ``hash'' functions.

\subsubsection{Two-Universal Hashing}\label{sec:Two-Uni}
\begin{definition}[Two-universal function family]\label{def:twoUniversalFamily} A function family $\h = \set{h\colon\Dom \mapsto \Rng}$ is {\sf two-universal} if   $\forall x\neq x' \in \Dom$, it holds that $\ppr{h\getsr \h}{h(x) = h(x')} \le  1/\size{\Rng}$.
\end{definition}
An example of  such a function family is the set $\h_{s,t} = \zo^{s\times t}$ of Boolean matrices,  where for  $h\in \h_{s,t}$ and $x\in \zo^s$, we let $h(x) = h \times x$  (\ie the matrix vector product over $\GF_2$). Another canonical example is $\h_{s,t} = \zo^s$ defined by  $h(x) \eqdef h\cdot x$ over $\GF(2^s)$, truncated to its first $t$ bits.

A useful application  of two-universal hash functions  is to convert a source   of high \Renyi entropy to  a (close to) uniform distribution.
\begin{lemma}[Leftover hash lemma \cite{ImpagliazzoLeLu89,ImpagliazzoZ89}]\label{lem:leftover}
	Let $X$ be a random variable over $\zn$ with $\HRen(X) \geq k$,  let $\h = \set{g\colon\zn \mapsto \zm}$ be two-universal, and let $H \getsr \h$. Then 
	$\SD((H,H(X)),(H,\Uni_m)) \leq \frac12\cdot  2^{(m -k)/2}$.
\end{lemma}

\subsection{Many-wise Independent Hashing}\label{sec:Many-Ind}
\begin{definition}[$\ell$-wise independent function family]\label{def:t-wiseFamily} A function family $\h = \set{h\colon\Dom \mapsto \Rng}$ is {\sf $\ell$-wise independent} if   for any distinct $x_1,\ldots,x_\ell \in \Dom$, it holds that $(H(x_1),\ldots,H(x_\ell))$ for $H \getsr \h$ is uniform over $\Rng^\ell$.
\end{definition}
The canonical example of  such an $\ell$-wise independent function  family is  $\h_{s,t,\ell} = (\zo^s)^\ell$ defined by  $(h_0,\ldots,h_{\ell-i})(x) \eqdef \sum_{0\le i\le \ell-1} h_i \cdot x^i$ over $\GF(2^s)$, truncated to its first $t$ bits. 

It is easy to see that, for $\ell>1$, an $\ell$-wise independent function  family is two-universal, but $\ell$-wise independent function families,  in particular with  larger value of $\ell$, have  stronger guarantees on their output distribution compared with two-universal hashing. We will state, and use,  one such guarantee in the construction of statistically hiding commitment schemes presented in \cref{sec:GapToCom}.

\subsection{One-way Functions}\label{subsec:OWF}
We recall the standard definition of one-way functions.

\begin{definition}[one-way functions]\label{def:OWF}
	A polynomial-time computable $f\colon \zo^n\mapsto \zo^\ast$ is  {\sf one-way} if for every \ppt $\Ac$
	\begin{align}\label{eq:OWF}
	\ppr{y\getsr f(U_{s(n)})}{\Ac(1^n,y)\in f^{-1}(y)} = \negl(n).
	\end{align}
\end{definition}

Without loss of generality, \cf \cite{HaitnerHR11}, it can be assumed that $s(n) =n$ and  $f$ is length-preserving (\ie $\size{f(x)} = \size{x}$).

In \cref{sec:ATEtoOWF} we prove that the existence of inaccessible entropy generators implies that of one-way functions.  The reduction uses the seemingly weaker notion of distributional one-way functions. Such a function is easy to compute, but it is hard to compute uniformly random preimages of random images.

\begin{definition}\label{def:dowf}
	A polynomial-time  computable  $f\colon\zn\to\zo^{\ell(n)}$ is  {\sf distributional one-way}, if $\exists p\in\poly$ such that
	$$\SD\left((x,f(x))_{x\getsr \zn},(\Ac(1^n,f(x)),f(x))_{x\getsr \zn}\right)\geq \frac{1}{p(n)}$$
	for any \pptm $\Ac$ and large enough $n$.
\end{definition}

Clearly, any one-way function is also a distributional one-way function. While the other direction is not necessarily always true, \citet{ImpagliazzoLu89} showed that the existence of distributional one-way functions  implies that of (standard) one-way functions. In particular, they proved that if one-way functions do not exist, then any efficiently computable function has an inverter of the following form.
\begin{definition}[$\xi$-inverter]\label{sef:gammaInverter}
	An algorithm $\Inv$ is an {\sf $\xi$-inverter} of $f\colon \zn\mapsto \zo^{\ell(n)}$ if the following holds:
	\begin{align*}
	\ppr{x\getsr \zn; y = f(x)}{\SD\left((x')_{x'\la f^{-1}(y)},(\Inv(1^n,y))\right) > \xi} \leq \xi.
	\end{align*} 
\end{definition}
\begin{lemma}[{\cite[Lemma 1]{ImpagliazzoLu89}}, full details in {\cite[Theorem 4.2.2]{ImpagliazzoPHD}}]\label{lemma:NoDistOWF}
	Assume one-way functions do not exist. Then for any polynomial-time computable function $f\colon\zn\mapsto \zo^{\ell(n)}$ and $p\in \poly$, there exists a \pptm algorithm $\Inv$ such that the following holds for infinitely many $n$'s:  on security parameter $1^n$, algorithm $\Inv$ is a $1/p(n)$-inverter of $f_n$  (\ie $f$ is restricted to $\zn$).
\end{lemma}

\citet{ImpagliazzoLu89} only gave a proof sketch for the above lemma. The full proof can be found in \cite[Theorem 4.2.2]{ImpagliazzoPHD} (see more details in \cite{BermanHT18}).
\section{Inaccessible Entropy Generators}\label{sec:realvsacc}
In this section we formalize the notion of a block-generator, and the real and accessible entropy of such generators. As discussed in the introduction, these entropies and the gap between them (\ie the inaccessible entropy of the generator) play a pivotal role in our work.


We begin by informally recalling the definition of a block-generator from the introduction. Let $\Gc\colon \zn \mapsto (\zs)^m$ be an $m$-block generator over $\zn$ and let $\Gc(1^n) = (Y_1,\ldots,Y_m)$ denote the output of $\Gc$ over a uniformly random input. The \emph{real entropy} of $\Gc$ is the (Shannon) entropy in $\Gc$'s output blocks:  for each block $Y_i$ we take its entropy conditioned on the previous blocks $Y_{<i} = (Y_1,\ldots,Y_{i-i})$.  The \emph{accessible entropy} of an arbitrary, adversarial $m$-block generator $\Gs$, with the same block structure as of $\Gc$,  is the entropy of the block of $\Gs$ conditioned not only on the previous blocks but also on the \emph{coins used by $\Gs$ to generate the previous blocks}. The generator $\Gs$ is allowed to flip fresh random coins to generate its next block, and this is indeed the source of entropy in the block (everything else is fixed). We only consider generator $\Gs$ whose messages are \emph{always}  consistent with $\Gc$: the support of $\Gs$'s messages is contained in that of $\Gc$.

Moving to the formal definitions, we first define  an $m$-block generator and  then define the real and accessible entropy of such a generator.   The generators below  might get  ``public parameter'' chosen uniformly, and we measure the entropy of the generator conditioned on this parameter.\footnote{The public parameter extension  is not used for our reduction from one-way functions to statistically  hiding commitments presented in this work, but  this natural extension  was already found useful in the   work of \citet{BitanskyHKY18}.}

\begin{definition}[Block-generators]\label{def:blockGenrator}
	Let $n$ be a security parameter, and let $\pp=\pp(n)$, $s=s(n)$ and $m=m(n)$. An {\sf $m$-block generator} is a function $\Gc \colon \zo^{\pp} \times \zo^{s} \mapsto (\zs)^{m}$, and it is {\sf efficient} if its running time on input of length $\pp(n)+ s(n)$ is polynomial in $n$.
	
	We   call parameter $n$ the {\sf security parameter}, $\pp$ the public parameter length, $s$ the {\sf seed length},  $m$  the  {\sf number of blocks}, and  $\ell(n) = \max_{(z,x)\in \zo^{\pp(n)} \times \zo^{s(n)},i\in [m(n)]} \size{\Gc(z,x)_i}$ the {\sf maximal block length} of $\Gc$.  
\end{definition}

\subsection{Real Entropy}\label{SHC:def:RealE}
Recall that we are interested in lower bounds on the real entropy of a block-generator. We define two variants of real entropy: real Shannon entropy and real min-entropy. We connect these two notions through the notion of real sample-entropy. In other words, for a fixed  $m$-tuple output of the generator, we ask ``how surprising were the blocks output by $\Gc$ in this tuple?'' We then get the real Shannon entropy by taking the expectation of this quantity over a random execution and the min-entropy by taking the minimum (up to negligible statistical distance). An alternative approach would be to define the notions through the sum of conditional entropies (as we do in the intuitive description in the introduction). This approach would yield closely related definitions, and in fact exactly the same definition in the case of Shannon entropy (see \Cref{lem:realShannondef}).

\begin{definition}[Real sample-entropy]\label{def:RealSamEnt}
Let $\Gc$ be an $m$-block generator over $\zo^{\pp} \times \zo^{s}$,  let $n\in\N$, let  $Z_n$ and $X_n$ be uniformly distributed over $\zo^{\pp(n)}$ and $\zo^{s(n)}$ respectively, and let $\vY_n= (Y_{1},\ldots,Y_m) =  \Gc(Z_n, X_n)$.  For $n\in\N$ and $i\in [m(n)]$, define the {\sf real sample-entropy of $\vy\in \Supp(Y_{1},\ldots,Y_i)$ given $z\in\Supp(Z_n)$} as
	$$\Hrealsam_{\Gc,n}(\vy| z) = \sum_{j=1}^i \Hsam_{Y_j|Z_n,Y_{<j}}(\vy_j|z,\vy_{<j}).$$
\end{definition}
\noindent We omit the security parameter from the above notation when clear from the context.
  
\begin{definition}[Real entropy]
	Let $\Gc$ be an $m$-block generator, and  let  $Z_n$ and $\vY_n$ be as in $\cref{def:RealSamEnt}$. Generator $\Gc$ has {\sf real entropy at least $k = k(n)$}, if
	$$\eex{(z,\vy) \getsr (Z_n,\vY_n)}{\Hrealsam_{\Gc,n}(\vy| z)} \geq k(n)$$
	for every $n\in \N$.
	
	The generator $\Gc$ has {\sf real min-entropy at least $k(n)$ in its \ith block} for some $i = i(n)\in [m(n)]$, if
	$$\ppr{(z,\vy) \getsr (Z_n,\vY_n)}{\Hsam_{Y_i|Z_n,Y_{<i}}(\vy_i|z,\vy_{<i}) < k(n)} = \negl(n).$$

We say the above bounds	 are {\sf invariant to the public parameter} if they hold for any fixing of the public parameter $Z_n$.\footnote{In particular, this is the case when there is no public parameter, \ie $\pp= 0$.}
\end{definition}

We observe that the real Shannon entropy simply amounts to measuring the standard conditional Shannon entropy of $\Gc$'s output blocks.

\begin{lemma} \label{lem:realShannondef}
	Let  $\Gc$,  $Z_n$ and $\vY_n$ be as in \cref{def:RealSamEnt} for some $n\in \N$. Then
	$$\eex{(z,\vy) \getsr (Z_n,\vY_n)}{\Hrealsam_{\Gc,n}(\vy|z)} =  \Hall(\vY_n|Z_n).$$
\end{lemma}
\begin{proof}
Omit $n$ for clarity. Applying  \cref{prop:ShanonoToSampleSeq}(1)  to $\vX = (Z,\vY)$ and $\cJ = \set{2,\ldots,m+1}$ yields that $\eex{(z,\vy) \getsr (Z,\vY)}{\Hrealsam_\Gc(\vy|z)} = \sum_{i\in [m]} \HSh(Y_i|Z,Y_{<i})$, where  by the chain rule for Shannon entropy ,	$\sum_{i\in [m]} \HSh(Y_i|Z,Y_{<i}) =  \Hall(\vY|Z)$.
\end{proof}

\subsection{Accessible Entropy}\label{SHC:def:IE}
Recall that we are interested in \emph{upper bounds} on the accessible entropy of a block-generator. As in the case of real entropy, we define  this   notion through the notion of accessible sample-entropy. For a fixed execution of the adversary $\Gs$, we ask how surprising were the messages sent by $\Gs$, and then get accessible Shannon entropy by taking the expectation of this quantity over a random execution.


%



\begin{definition}[Online block-generator]
	Let $n$ be a security parameter, and let $\pp = \pp(n)$ and $m=m(n)$. An $m$-block  {\sf online} generator is a function $\Gs \colon \zo^\pp \times (\zo^{v})^{m} \mapsto  (\zs)^{m}$ for some $v =v(n)$, such that the \ith output block of $\Gs$ is a function of (only) its first $i$  input blocks. We denote the {\sf transcript} of $\Gs$ over random input by  $T_{\Gs}(1^n) = (Z,R_1,Y_1,\ldots,R_m,Y_m)$, for  $Z \getsr \zo^\pp$, $(R_1,\ldots,R_m) \getsr (\zo^{v})^{m}$ and $(Y_1,\ldots,Y_m) =\Gs(Z,R_1,\ldots,R_m)$.
\end{definition}
That is, an online block-generator is a special type of block-generator that tosses fresh random coins before outputting each  new block. In the following, we let $\Gs(z,r_1,\ldots,r_i)_i$  stand for $\Gs(z,r_1,\ldots,r_i,r^\ast)_i$ for arbitrary $r^\ast \in (\zo^{v})^{m -i}$ (note that the choice of $r^\ast$ has no effect on the value of $\Gs(z,r_1,\ldots,r_i,r^\ast)_i$).

\begin{definition}[Accessible sample-entropy]\label{def:AccessibleSampleEntropy}
	Let $n$ be a security parameter, and let $\Gs$ be an online $m=m(n)$-block online  generator. The {\sf accessible sample-entropy of $\vt =(z,r_1,y_1,\ldots,r_m,y_m)\in \Supp(Z,R_1,Y_1\ldots,R_m,Y_m) = T_{\Gs}(1^n)$} is defined by
	
	$$\Haccsam_{\Gs,n}(\vt) = \sum_{i=1}^{m} \Hall_{Y_i|Z,R_{<i}}(y_i|z,r_{<i}).$$
\end{definition}
\noindent 
Again, we omit the security parameter from the above notation when clear from the context.

As in the case of real entropy, the expected accessible entropy of a random transcript can  be expressed in terms of the standard conditional Shannon entropy.
\begin{lemma}\label{lem:accessible-as-conditional}
	Let $\Gs$ be an online $m$-block generator  and let  $(Z,R_0,R_1,Y_1,\ldots,R_m,Y_m) = T_\Gs(1^n)$ be its transcript.	Then, 
	$$\eex{\vt\getsr T_{\Gs}(Z_n,1^n)} {\Haccsam_{\Gs}(\vt)} = \sum_{i\in[m]} \Hall(Y_i|Z,R_{<i}).$$
\end{lemma}
\begin{proof}
Follows similar lines to  that of \Cref{lem:realShannondef}. 
\end{proof}
The following  observation relates the  accessible entropy to the  accessible sample-entropy of likely events.

\begin{lemma}\label{lem:MaxAvgAE}
	Let $\Gs$ be an online $m$-block generator of maximal block length $\ell$ and  let  $T= T_{\Gs}(1^n)$. If $\pr {\Haccsam_{\Gs}(T) > k} <\eps$, then $\ex{\Haccsam_{\Gs}(T)} \le  (1 - \eps)k + \eps(  m \ell - \log 1/\eps)  < k +  \eps m \ell + 1$. 	
\end{lemma}
\begin{proof}
	Follows by \cref{prop:HighAEContributionSeq} with $\vX = \vT_\Gs(1^n,Z_n) = (Z_n,Y_1,R_1,\ldots, Y_n,R_n)$, for  $(Z_n,R_1,Y_1,\ldots, R_n,Y_n) = \vT_\Gs(1^n,Z_n)$ and $\cJ = \set{2,4,6, \ldots}$ (\ie the indices of the output blocks of $\Gs$).
\end{proof}

The above definition is only interesting when placing restrictions on the generator's actions \wrt the underlying generator $\Gc$. (Otherwise, the accessible entropy of $\Gs$ can be arbitrarily large by outputting arbitrarily long strings.) In this work, we focus on efficient generators that are consistent  \wrt  $\Gc$. That is, the support of their output  is contained in that of $\Gc$.\footnote{In the more complicated notion of accessible entropy considered in \cite{HaitnerReVaWe09}, the ``generator" needs to \emph{prove} that its output blocks are in the support of $\Gc$, by providing an input of $\Gc$  that would have generated the same blocks. It is also allowed there for   a generator to fail to prove the latter with some probability, which requires  a measure of  accessible  entropy that discounts entropy that may come from failing.}

\begin{definition}[Consistent  generators]\label{def:nonFailingGen}
	Let $\Gc$ be a block-generator over $\zo^{\pp(n)} \times \zo^{s(n)}$. A block (possibly online) generator $\Gc'$ over $\zo^{\pp(n)} \times \zo^{s'(n)}$ is {\sf $\Gc$ consistent} if, for every $n\in \N$, it holds that  $\Supp(\Gc'(U_{\pp(n)},U_{s'(n)})) \subseteq \Supp(\Gc(U_{\pp(n)},U_{s(n)}))$.
\end{definition}

\begin{definition}[Accessible entropy]\label{def:accessible-entropy}
	A block-generator $\Gc$ has {\sf accessible entropy at most $k = k(n)$} if, for  every efficient $\Gc$-consistent, online generator $\Gs$ and all large enough $n$,
	$$\eex{\vt\getsr T_\Gs(1^n)} {\Haccsam_\Gs(\vt)} \leq k.$$
\end{definition}

We call a  generator whose real entropy is noticeably higher than its  accessible entropy an inaccessible entropy generator.

\begin{remark}[Maximal accessible entropy]
An alternative  to the  (average) accessible entropy discussed above is to consider the \emph{maximal} accessible sample-entropy of a  random  execution of the generator (ignoring negligible events).  The resulting notion is somewhat harder to work with, but when applicable it typically yields  more efficient constructions (time wise and communication wise). Indeed, working with maximal  accessible entropy  would have yielded a  statistically  hiding commitment based on one-way functions that is more computation and communication efficient than the one we construct in     \cref{sec:GapToCom}. Yet, preferring simplicity  over efficiency, in the main body of the paper we chose to work with the simpler notion of (average) accessible entropy, and formally define and prove basic facts on the maximal accessible entropy notion   in \cref{sec:MaxAE}.
\end{remark}

\section{Inaccessible Entropy Generator from One-way Functions}\label{sec:IAEGenFromOWF}
In this section we show how to build an  inaccessible entropy generator from any one-way function. In particular, we prove the following result:
\begin{construction}\label{ctr:AEG}
	For $f\colon\zn \mapsto \zn$, define  the $(n/\log n) +1$  block-generator $\Gc^f$ over $\zn$,\footnote{We assume for simplicity that $n/\log n \in \N$. Otherwise we ``pad" $f$'s output.} by
	$$\Gc(x) = f(x)_{1,\ldots,\log n},f(x)_{\log n +1,\ldots,2\log n},\ldots,f(x)_{n-\log n +1,\ldots,n},x $$
\end{construction}
Namely, the first $n/\log n$ blocks of $\Gc^f(x)$ are the bits of  $f(x)$ partitioned  into $1\log n$-bit  sequences, and its final block is $x$.

\begin{theorem}[Inaccessible entropy generators from  one-way functions]\label{thm:AEGfromOWF}
	If  $f\colon\zn \mapsto \zn$ is  a one-way function, then the efficient block-generator  $\Gc=\Gc^f$ defined in \cref{ctr:AEG} has accessible entropy $n-\omega(\log n)$.
\end{theorem}
\begin{proof}	
Suppose the theorem does not hold and let $\Gs$  be  an efficient,   $\Gc$-consistent online  block-generator  such that
	\begin{align}\label{eq:AEGfromOWF:1}
	\ex{\Haccsam_{\Gs}(\tT)} \ge n-c\cdot \log n
	\end{align}
	for some $c>0$, and infinitely many $n$'s. In the following,  fix $n\in \N$ for which the above equation holds,  and omit it from the notation when its value is clear from the context. Let $m = n/\log n$ and let $\vv$ be a bound on the number of coins used by $\Gs$ in each round. The inverter $\Inv$ for $f$ is defined as follows:
	\begin{algorithm}[Inverter $\Inv$ for $f$ from the accessible entropy generator $\Gs$]\label{alg:Inv}~
		\begin{description}
			\item[Input:] $z\in \zn$ 
			\item[Operation:] 
		\end{description}
		\begin{enumerate}

			\item For $i = 1$ to $m$: \label{alg:Inv:step:1}
			\begin{enumerate}
				\item Sample $r_i \getsr \zo^\vv$ and let $y_i = \Gs(r_1,\ldots,r_i)_i$.
				\item If $y_{i\cdot \log n +1,\ldots,(i+1)\cdot \log n} = z_i$, move to next value of $i$.\label{alg:Inv:step:1:c}
				
				\item Abort after  $n^3$ failed attempts for sampling  a good $r_i$.
			\end{enumerate}
			\item Sample $r_{m+1} \getsr \zo^\vv$ and output $\Gs(r_1,\ldots,r_{m+1})_{m+1}$.\label{alg:Inv:step:2}  \footnote{Choosing $r_{m+1}$ uniformly at random is merely for the sake of the analysis.  Setting, for instance,   $r_{m+1} = 0^v$  induces the same inversion probability.}
		\end{enumerate}
		
	\end{algorithm}
	Namely, $\Inv(y)$ does the only natural thing that  can be done with $\Gs$; it tries to make, via rewinding,  $\Gs$'s first $m$ output blocks equal to $y$, knowing that if this happens then, since  $\Gs$ is $\Gc$-consistent, $\Gs$'s $(m+1)$ output block is a preimage of $y$. 
	
	It is clear that  $\Inv$ runs in polynomial time, so we will finish the proof by showing that  it inverts $f$ with high probability.  We prove the latter by showing that transcript distribution  induced  by  the (standalone execution)  of $\Gs(1^n)$ is close in KL-divergence to the   execution of $\Gs$ embedded (emulated) in $\Inv(f(U_n))$.  Since  the last output of $\Gs$ is always the preimage of the image point defined by  its first $m$ output elements, it follows that it is also the case, with high probability, for embedded    execution  $\Gs$, and thus $\Inv(f(U_n))$ finds a preimage with high probability.

	Let $\tT = (\tR_1,\tY_1,\ldots,\tR_{m+1},\tY_{m+1}) = \vT_\Gs$, and recall that $\tT$ is associated  with a random execution of $\Gs$ on security parameter $n$ by
	\begin{itemize}
		\item $\tR_i$ -- the random coins of $\Gs$ in the \ith round, and
		\item $\tY_i$ -- $\Gs$'s \ith output block.
	\end{itemize}
	Let $\hT = (\hR_1,\hY_1,\ldots,\hR_{m+1},\hY_{m+1})$ denote the value of $\Gs$'s coins and  output blocks, set by the end of  \Stepref{alg:Inv:step:2}  in a random execution of the \emph{unbounded} version of $\Inv$ (\ie \Stepref{alg:Inv:step:1:c}  is removed) on input $f(U_n)$,  setting it arbitrary in case the execution of (the unbounded) $\Inv$ does not end.\footnote{The unboundedness change is only an intermediate step in the proof that does not significantly change  the inversion probability of $\Inv$, as shown below.}    The heart of the proof lies in the following claim:
	
	\begin{claim}\label{claim:AEGfromOWF}
	$\kld{\tT}{\hT} \le  c\cdot \log n$.
	\end{claim}
	We prove \cref{claim:AEGfromOWF} below but first use it for proving the theorem.	Let  $w(r_1,y_1,\ldots,r_{m+1},y_{m+1})$ be the indicator for $y_{m+1} \in f^{-1}(y_{\le m})$. Clearly, 
		\begin{align}
		\pr{w(\tT)} =1
		\end{align}
	Let  $e(r_1,y_1,\ldots,r_{m+1},y_{m+1})$ be the indicator that for all $i\in [m]$ it holds that 
	$\pr{\tY_i = y_i \mid \tR_{<i} = r_{<i}} \ge 1/2mn$.
	Since $\tY_i$ takes at most $n$ values,  a straightforward union bound argument yields that 
	\begin{align}
	\pr{e(\tT)} \ge 1/2
	\end{align}
	By the above two facts, we conclude that  $\pr{(w\land e)(\tT)  = w(\tT) \land e(\tT)} \ge 1/2$.	
	 Hence, by \cref{claim:AEGfromOWF} and the data-processing property of KL-divergence, 
	 $\kld{(w\land e)(\tT)}{(w\land e)(\hT)} \le  c\cdot \log n$.  Thus, a simple calculation yields that
	\begin{align}
	\pr{(w\land e)(\hT)} \ge 1/ n^{c+1}
	\end{align}
	for large enough $n$. Since a transcript $t$ with $(w\land e)(t)=1$ is generated by (bounded) $\Inv$ with probability at  $(1-m\cdot 2^{-n})$ times the probability it is generated by the unbounded variant of $\Inv$, we deduce that 
	$$\ppr{y\getsr f(U_n)}{\Inv(y) \in f^{-1}(y)}  \ge  \pr{(w\land e)(\hT)}  -  m\cdot 2^{-n} \ge 1/ n^{c+1} - n\cdot 2^{-n}$$
	for large enough $n$, contradicting the one-wayness of $f$.
\end{proof}

	\paragraph{Bounding $\kld{\tT}{\hT}$ (proving \cref{claim:AEGfromOWF}).}
	
	\begin{proof}[Proof of \cref{claim:AEGfromOWF}]
   By definition,  
	\begin{align}\label{eq:AEGfromOWF:2}
	\tR_i|_{\tY_{  i} =y_i,\tR_{<i} = r_{<i}} \equiv \hR_i|_{\hY_{ i} =y_i,\hR_{<i} = r_{<i}}
	\end{align}
	 for every $i\in [m+1]$, $y_i\in \Supp(\tY_{  i})$ and $r_{<i}\in \Supp(\tR_{<i})$. In addition,  
	\begin{align}\label{eq:AEGfromOWF:3}
	\tY_{m+1}|_{\tR_{ \le m} =r_{<m}} \equiv \hY_{m+1}|_{\hR_{ \le m} =r_{<m}}
	\end{align}
	 for every $r_{<m} \in \Supp(\tR_{ \le m})$.  Compute
	
	\begin{align}\label{eq:AEGfromOWF:3:5}
	\lefteqn{\kld{\tT}{\hT}}\\
	&=\sum_{i=\in \set{2,4,\ldots,2m+2}}  \eex{t\getsr \tT_{<i}}{\kld{\tT_i|_{\tT_{<i}=t}}{\hT_i|_{\hT_{<i}=t}}} + \sum_{i=\in \set{1,3,\ldots,2m+1}}  \eex{t\getsr \tT_{<i}}{\kld{\tT_i|_{\tT_{<i}=t}}{\hT_i|_{\hT_{<i}=t}}}\nonumber\\
	&=\sum_{i=1}^{m+1} \eex{r\getsr\tR_{<i}}{\kld{\tY_i|_{\tR_{<i}=r}}{\hY_i|_{\hR_{<i}=r}}} + \sum_{i=1}^{m+1} \eex{(r,y)\getsr(\tR_{<i},\tY_i)}{\kld{\tR_i|_{\tR_{<i}=r,\tY_i =y}}{\hR_i|_{\hR_{<i}=r},\hY_i=y}}\nonumber\\
	&=\sum_{i=1}^m \eex{r\getsr\tR_{<i}}{\kld{\tY_i|_{\tR_{<i}=r}}{\hY_i|_{\hR_{<i}=r}}}  +0+0\nonumber.
	\end{align}
	The first 	equality holds by chain-rule of KL-divergence, and since KL-divergence is invariant to permutation. The last equality holds by \cref{eq:AEGfromOWF:2,eq:AEGfromOWF:3}. It follows that
	\begin{align}	\label{eq:AEGfromOWF:4}
	\kld{\tT}{\hT} &=\sum_{i=1}^m-\Hall(\tY_i\mid \tR_{<i}) -  \sum_{i=1}^m \eex{r_{<i}\getsr \tR_{< i}}{\log \left(\pr{\hY_i = y_i \mid \hR_{<i}=  r_{<i} }\right)}\\
	&=\sum_{i=1}^m-\Hall(\tY_i\mid \tR_{<i}) -  \sum_{i=1}^m \eex{y_{< i}\getsr\tY_{<  i}}{\log \left(\pr{\hY_i = y_i \mid \hY_{<i} = y_{<i} }\right)}\nonumber\\
	&=\sum_{i=1}^m-\Hall(\tY_i\mid \tR_{<i}) - \eex{y_{\le m}\getsr\tY_{\le m}}{\sum_{i=1}^m \log \left(\pr{\hY_i = y_i \mid \hY_{<i} = y_{<i} }\right)}\nonumber\\
	&=\sum_{i=1}^m-\Hall(\tY_i\mid \tR_{<i}) - \eex{y_{\le m}\getsr\tY_{\le m}}{ \log \left(\pr{\hY_{\le m} = y_{\le m} }\right)}\nonumber\\
	&=\Hall(\tY_{m+1}\mid \tR_{\le m})  -\ex{\Haccsam_{\Gs}(\tT)}   - \eex{y_{\le m}\getsr\tY_{\le m}}{ \log \left(\pr{\hY_{\le m} = y_{\le m} }\right)}.\nonumber
	\end{align}	
 The first inequality holds  by \cref{eq:AEGfromOWF:3:5} and definition of KL-divergence. The  second  equality holds since  $\hY_i$ and $\hR_{<i}$ are independent conditioned on  $\hY_{<i}$.  The last   equality is by the definition  of $\Haccsam_{\Gs}$. Since $\Hall(\tY_{m+1}\mid \tR_{\le m}) = \eex{y_{\le m}\getsr\tY_{\le m}}{ \log \size{f^{-1}(y_{\le m})}}$ and $\pr{\hY_{\le m} = y_{\le m}}= \size{f^{-1}(y_{\le m})}/2^n$, \cref{eq:AEGfromOWF:4} yields  that 
	\begin{align*}
	\kld{\tT}{\hT} &\le -\ex{\Haccsam_{\Gs}(\tT)}  +n \le c\cdot \log n.
	\end{align*}
	The second inequality is by the assumption on $\Gs$ accessible entropy (\cref{eq:AEGfromOWF:1}). 
	\end{proof}

\section{Manipulating Real and Accessible Entropy}\label{sec:manipulatingAE}
In this section, we  develop tools  to manipulate the real and accessible entropy of a block-generator. These tools are later used in \cref{sec:GapToCom} as a first step toward using the one-way function inaccessible entropy generator constructed in \cref{sec:IAEGenFromOWF} to construct  statistically hiding commitments. The tools considered  below are rather standard ``entropy manipulations": entropy equalization (\ie picking  a random variable at random from a set of random variables to get a new random variable whose entropy is the average entropy) and direct product, and their effect on the real entropy of random variables is rather clear. Fortunately, these manipulations have  the same effect also on the  \emph{accessible  entropy} of  a block-generator.

\subsection{Entropy Equalization via Truncated Sequential Repetition}
This tool concatenates several independent executions of an $m$-block  generator, and then  truncates, at random, some of the first and final output blocks of the concatenated string. Assuming that the (overall) real entropy of the original generator is at least $\kreal$, then the  real entropy of each block of the resulting  generator is at least $\kreal/m$.  This per-block knowledge of the  real entropy  is very handy when considering  applications of inaccessible entropy generators,  and in particular for  constructions of statistically hiding commitments (see \cref{sec:GapToCom}).  The price of this manipulation is that we ``give away"  some real entropy (as we do not output all blocks), while we cannot claim that the same happens to the accessible entropy.  Hence,  the additive gap between the real and accessible entropy of the resulting generator gets smaller. Yet, if we do enough repetition, this loss is insignificant.

\begin{definition}\label{def:EqRealEnt}
	For security parameter $n$, let $m=m(n)$, let   $\pp = \pp(n)$, $s=s(n)$, $\ee=\ee(n)$, and let $\pp' = \pp'(n)= \ee(n) \cdot \pp(n)$ and $s' = s'(n)  = \log(m(n)) + \ee(n) \cdot s(n)$. Given an $m$-block generator $\Gc$ over $\zo^{\pp} \times \zo^{s}$, define the $\left((\ee-1)\cdot m\right)$-block generator  $\Gbl$  over $(\zo^{\pp})^\ee \times ([m] \times (\zo^{s})^\ee)$ as follows: on  input $ (z_1,\ldots,z_\ee, j,(x_1,\ldots,x_{\ee}))$, it sets  $\vy = (y_1,\ldots,y_{\ee  m}) = (\Gc(z_1,x_1),\ldots,\Gc(z_\ee,x_{\ee}))$, and outputs $((j,y_j),y_{j+1},\ldots,y_{(\ee-1) m +j  -1 })$.
\end{definition}

That is, $\Gbl$  truncates the first $j-1$  and  last $m+1-j$ blocks of $\vy$. It then  outputs the remaining $(\ee-1)\cdot m$  blocks of $\vy$ one by one, while appending $j$, indicating the truncation location,   to the first block. 

\begin{lemma}\label{lem:EqRealEnt}
	For security parameter $n$, let $m =m(n)$ be a power of $2$, let $\pp = \pp(n)$, $s=s(n)$, let $\Gc$ be an efficient $m$-block generator over $\zo^{\pp}\times \zo^{s}$, and let $\ee = \ee(n)$ be a polynomially computable and bounded integer function. Then $\Gbl$ defined according to \cref{def:EqRealEnt} is an efficient\footnote{Since $m$ is a power of $2$, standard techniques can be used to change the input domain of $\Gbl$ to  $\zo^{s'}$ for some polynomial-bounded and polynomial-time computable $s'$, to make it  an efficient  block-generator  according to \cref{def:blockGenrator}.} $\left((\ee-1)\cdot m\right)$-block generator  satisfying the following properties:
	\begin{description}
		\item[Real entropy:] If $\Gc$ has real entropy at least $\kreal = \kreal(n)$, then each block of $\Gbl$  has real entropy at least $\kreal/m$. Furthermore, if the bound on the real entropy $\Gc$ is invariant to the public parameter, then  this  is also the case for the above bound on $\Gbl$.
		
		\item[Accessible entropy:]   If $\Gc$ has accessible entropy  at most $\kmax = \kmax(n)$, then $\Gbl$ has accessible  entropy at most
		$$\kmaxp \eqdef (\ee-2)  \cdot \kmax + 2\cdot \Hmax(\Gc(U_\pp,U_s)) +  \log(m).$$  
	\end{description}
\end{lemma}
That is, each of the $(\ee-2)$ non-truncated executions of $\Gc$ embedded in $\Gbl$ contributes its accessible entropy to the overall accessible entropy of  $\Gbl$. In addition,  we pay the max-entropy of  the two truncated executions of $\Gc$ embedded in $\Gbl$ and the max-entropy of the index $j$.

\def\EEGen
{
	Consider the following efficient $\Gc$-consistent generator.
	\begin{algorithm}[Generator $\Gs$]
		\item Input: public parameter $z\in \zo^\pp$.
		
		\item Operation:
		
		\begin{enumerate}
			\item Sample $v\getsr \set{2,\ldots,\ee -1}$.  Let $\vz= (z_1,\ldots,z_\ee)$ for $z_v = z$, and $z_i$,  for $i\neq v$,  sampled uniformly in  $\zo^{\pp}$. 
			
			We will refer to the part of $\vz$ sampled by the generator as $\vz_{-v}$, and (abusing notation) assume $\vz = (z,\vz_{-v})$.
			
			\item Start a random execution of $\Gbs(\vz)$.  After $\Gbs$ locally outputs its first block $(j,\cdot)$, continue the execution of $\Gbs$ while outputting, block by block,  the output blocks of $\Gbs$ indexed by $\set{f =  (v-1)m +2 -j,\ldots,f+m-1}$. 
		\end{enumerate}   
	\end{algorithm}
	It is clear that  $\Gs$ is indeed an efficient $\Gc$-consistent generator. We will show that the accessible entropy of $\Gs$ violates  the assumed  bound on the accessible entropy of $\Gc$.
}

\begin{proof}
	To avoid notational clutter let $\Gb = \Gbl$.

	\paragraph{Real entropy.}
	 Fix $n\in \N$ and omit it from the notation when clear from the context. Let $\mt= (\ee -1) m$,  let $\vZ = (Z_1,\ldots,Z_\ee) \getsr(\zo^\pp)^\ee$ and  $\Yl = \Gbl(\vZ,U_{s'}= (J,X_1,\ldots,X_{\ee}))$.  Let $\Ypl = (\Gc(Z_1,X_1),\ldots,\Gc(Z_\ee,X_{\ee}))$, and finally  for $i\in [\ee m]$, let $\hYpl_i= (J,\Ypl_i)$ if $J=i$, and $\hYpl_i=\Ypl_i$ otherwise. 
	 
	 Fix $i \in [\mt]$. By the chain rule for Shannon entropy, it holds that 
	\begin{align}
	\Hall(\Yl_i\mid Z,\Yl_{<i}) &=  \Hall(\hYpl_{i+J -1} \mid \vZ,\hYpl_{J,\ldots, i+J-2})\\
		&\ge  \Hall(\hYpl_{i+J -1} \mid \vZ, \Ypl_{J,\ldots, i+J-2},J).\nonumber\\
	&=  \Hall(\Ypl_{i+J -1} \mid \vZ, \Ypl_{J,\ldots, i+J-2},J).\nonumber
	\end{align}

		Let $(Y_1,\ldots,Y_{m}) = \Gc(Z = U_{\pp},U_{s})$,  and let $m \bmod m$ be  $m$ (rather than $0$). It follows that 
		\begin{align*}
		\Hall(\Ypl_{i+J -1} \mid Z, \Ypl_{< i+J-1},J)& = \eex{j\getsr J}{	\Hall(\Ypl_{i+j -1} \mid \vZ,\Ypl_{< i+j-1})}\\
		&=  \eex{j\getsr J}{\Hall(Y_{i+j -1 \bmod m} \mid Z, Y_{< i+J-1 \bmod m})}\nonumber\\
		&= \Ex_{i'\getsr [m]}[\Hall(Y_{i'}\mid Z, Y_{ < i'})]\nonumber\\
		&= \frac1{m} \cdot  \sum_{i'\in [m]} \Hall(Y_{i'}\mid Z,Y_{<i'})\nonumber\\
		& \ge \kreal/m,\nonumber
		\end{align*}
		yielding that $\Hall(\Yl_i\mid \vZ, \Yl_{<i})\ge \kreal/m$. The second  equality follows from the fact that, for any $t \in [\ee m]$, $(Z_{\ceil{t/m}},\Ypl_{t' = \floor{t/m}\cdot m + 1},\ldots,\Ypl_t)$ is independent of $(\vZ_{< \ceil{t/m}},\vZ_{> \ceil{t/m}},\Ypl_{< t'})$, and has the same distribution as  $(Z,Y_{1},\ldots,Y_{t \bmod m})$. The third equality holds since $(i+J -1 \bmod m)$ is uniformly distributed in $[m]$.

	It readily follows from the above proof that if the bound on the real entropy $\Gc$ is invariant to the public parameter, then  the above $\kreal/m$ bound on the  real entropy of $\Gbl$ is also  invariant to the public parameter.

	\paragraph{Accessible entropy.} 
	 Let $\Gbs$ be  an efficient  $\Gb \ (= \Gbl)$-consistent generator, and assume
	\begin{align}\label{eq:EqRealEnt:0}
	k_ \Gbs= \eex{\vt \getsr  \Tbs}{\Haccsam_{\Gbs}(\vt)} > \kmaxp
	\end{align}
	for $\Tbs = T_{\Gbs}(1^n)$.	We show that  \cref{eq:EqRealEnt:0} yields that  a  random sub-transcript of  $\Tbs$  contributes more than  $\kmax$ bits of  accessible entropy, and use it to construct a cheating generator for $\Gc$ whose  accessible entropy is greater than $\kmax$, in contradiction to the assumed accessible entropy of $\Gc$.

	Let $ (\vZ,\tRR_1,\tYY_1,\ldots,\tRR_\mt,\tYY_\mt) = \Tbs$ and let $J$ be the first part of  $\tYY_1$. Since $\Gbs$ is $\Gb$-consistent, $\tYY_1$ is   of the form $(J,\cdot)$. Fix $j\in [m]$, let   $(\vZ, \tRR_1^j,\tYY_1^j,\ldots,\tRR_\mt^j,\tYY^j_\mt) = \Tbs^j = \Tbs|_{J=j}$ and let $\I = \I(j)$ be the indices of the output blocks coming from the truncated executions of $\Gc$ in $\Gb$ (\ie $\I = \set{1,\ldots,m+1-j} \cup \set {\mt+2-j,\ldots,\mt}$). It is easy to see that these  blocks do not contribute  more entropy than  twice the max-entropy of $\Gc$. Indeed  \cref{prop:ShanonoToSampleSeq}, letting $\vX=(\vZ,\tYY_1^j,\tRR_1^j,\ldots,\tYY_\mt^j,\tRR_\mt^j)$ and   $\cJ$  being the indices of the  blocks of $\I$ in $\vX$, yields that 
	\begin{align}\label{EqRealEnt:1}
	\eex{(\vz,r_1,y_1,\ldots,r_{\mt},y_\mt) \getsr  \Tbs^j}{\sum_{i\in \I} \Hall_{\tYY_i^j|\vZ,\tRR_{<i}^j}(y_i|\vz,r_{<i})} \le  2\cdot \Hmax(\Gc(U_\pp,U_s))
	\end{align}

	Consider  $(\vz,r_1,y_1 = (j,\cdot),\ldots,r_{\mt},y_\mt) \in \Supp(\Tbs)$. Since $J$ is determined by  $(Z,R_1)$, for $i >1$ it holds that $\Hall_{\tYY_i|\vZ,\tRR_{<i}}(y_i|\vz,r_{<i}) = \Hall_{\tYY_i^{j}|\vZ,\tRR_{<i}^j}(y_i|\vz,r_{<i})$, whereas  for $i=1$,  $\Hall_{\tYY_1|\vZ}(y_1|\vz) =\Hall_{J|\vZ}(j|\vz) + \Hall_{\tYY_1^{j}| \vZ}(y_1|\vz)$, where $\eex{(\vz,j) \getsr (Z,J)}{ \Hall_{J|\vZ}(j|\vz)} \le \log m$. These observations, and \cref{eq:EqRealEnt:0,EqRealEnt:1}, yield that
	\begin{align}\label{eq:EqAccEnt:2}
	\eex{(\vz,r_1,y_1= (j,\cdot),\ldots,r_{\mt},y_\mt) \getsr  \Tbs}{\sum_{i\in [\mt] \setminus \I(j)} \Hall_{\tYY_i|\vZ,\tRR_{<i}}(y_i|\vz,r_{<i})  }  &\ge k_ \Gbs   - 2 \cdot \Hmax(\Gc(U_\pp,U_s)) - \log m\\
	&> (\ee -2)\cdot \kmax, \nonumber
	\end{align}
and we conclude that
	\begin{align}\label{eq:EqAccEnt:3}
	&\eex{(\vz,r_1,y_1=(j,\cdot),\ldots,r_{\mt},y_\mt) \getsr  \Tbs; v\getsr \set{2,\ldots,\ee -1}}{\sum_{i \in \set{f= (v-1)m +2 -j,\ldots,f+m-1}} \Hall_{\tYY_i|\vZ,\tRR_{<i}}(y_i|\vz,r_{<i})} > \kmax
	\end{align}

	\EEGen

	 Let $(Z,\tR_1 = (V,\vZ_{-V},F, R'_{< F},R'_F),\tY_1,\ldots,\tR_m,\tY_m) = T_\Gs$ be the transcript of $\Gs(Z)$, for $Z\getsr \zo^\pp$,  $(V,\vZ_{-V},F)$ being the value of $(v,\vz_{-v},f)$ sampled in the first step of  $\Gs$, and $R'_{< F}$  and $R'_F$ being the randomness used by the emulated execution of $\Gbs(Z,\vZ_{-V})$  done  in the second step of $\Gs$, for outputting its first $(F-1)$ blocks and the $F$'th block, respectively.  Let  $(z,\tr_1 = (v,\vz_{-v},f,r'_{< f},r'_f),y_1,\ldots,\tr_m,y_m) \in \Supp(T_\Gs)$. It is easy to verify that
	
	\begin{align}
	\Hall_{\tY_1|Z,V,Z_{-V},F,R'_{<F}}(y_1 \mid z,v,\vz_{-v},f,r'_{<f}) = \Hall_{\tYY_{f}|\vZ,\tRR_{<f}}(y_1\mid (z,\vz_{-f}), r_{< f})
	\end{align}
	and that  for $i>1$: 
	\begin{align}
	\Hall_{\tY_i|Z,\tR_{<i}}(y_i \mid z,\tr_{<i}) = \Hall_{\tYY_{f + i}|\vZ,\tRR_{<f + i}}(y_i \mid (z,\vz_{-f}), (r'_{< f},r_f',\tr_2,\ldots,\tr_{i-1}))
	\end{align}

It follows that
	\begin{align*}
\lefteqn{\eex{(\vz,r_1,y_1=(j,\cdot),\ldots,r_{\mt},y_\mt) \getsr  \Tbs; v\getsr \set{2,\ldots,\ee -1}}{\sum_{i \in \set{f= (v-1)m +2 -j,\ldots,f+m-1}} \Hall_{\tYY_i|\vZ,\tRR_{<i}}(y_i|\vz,r_{<i})}}\\
 &= \eex{ (z,\tr_1 = (v,\vz_{-v},f,r'_{< f},r'_f),y_1,\ldots,\tr_m,y_m) \getsr  T_\Gs} {\Hall_{\tY_1|Z,V,Z_{-V},F,R'_{<F}}(y_i \mid z,v,\vz_{-v},f,r'_{<f}) + \sum_{i=2}^m \Hall_{\tY_i|Z,\tR_{<i}}(y_i \mid z,\tr_{<i}) } \\
&=\Hall(\tY_1 \mid Z,V,Z_{-V},F,R'_{<F}) +  \sum_{i=2}^m \Hall(\tY_i \mid Z,\tR_{<i})\\
&\le \Hall(\tY_1 \mid Z) +  \sum_{i=2}^m \Hall(\tY_i \mid Z,\tR_{<i})\\
&= \eex{\vt \getsr  T_\Gs} {\Haccsam_\Gs(\vt)}.
\end{align*}
The inequality is by the chain rule of Shannon entropy and the last equality by \cref{lem:accessible-as-conditional}. Thus,  by \cref{eq:EqAccEnt:3}, 	the accessible entropy of $\Gs$ is greater than $\kmax$, in contradiction to the assumed accessible entropy of $\Gc$.
\end{proof}

\subsection{Gap Amplification Via Direct Product}
This manipulation simply takes  the direct product of a generator. The effect of this manipulation is  two-fold. The first  effect is that  the overall real entropy of a $\vv$-fold direct product repetition of a generator $\G$ is  $\vv$ times the  real entropy of $\G$. Hence, if $\G$'s real entropy is larger than its accessible entropy,  this gap gets multiplied  by $\vv$ when we perform direct product.  The second effect of such repetition is that per-block real entropy is turned into per-block  \emph{min-entropy}. The price of this manipulation is a slight decrease in the per block  min-entropy of the resulting generator, compared to  the sum of the per block real entropies of the independent copies of the generators  used to generate it. (This loss is due to the move from Shannon entropy to min-entropy, rather than from the direct product itself.) But  when performing sufficient repetitions, this loss can be ignored.

\begin{definition}\label{def:PR}
	Let $m =m(n)$, $\pp = \pp(n)$ and $s = s(n)$, and  $\vv= \vv(n)$. Given an    $m$-block generator $\Gc$ over $\zo^\pp \times \zo^s$, define the $m$-block   generator $\Gt$  over $(\zo^\pp)^\vv \times (\zo^s)^\vv$ as follows: on input  $(\vz,\vx) $, the \ith block of $\Gt$ is $(\Gc(\vz_1,\vx_1)_i,\ldots,\Gc(\vz_\vv,\vx_{\vv})_i)$. 
\end{definition}

\begin{lemma}\label{lem:PR}
	For security parameter $n$, let $m =m(n)$, let  $\vv= \vv(n)$ be polynomial-time   computable and bounded integer functions, and let $\Gc$ be an efficient\footnote{Again, standard techniques can be used to change the input domain of $\G$ to  $\zo^{s'(n)}$ for some polynomial-bounded and polynomial-time computable $s'$, to make it an efficient  block-generator  according to \cref{def:blockGenrator}.}  $m$-block generator. Then $\Gt$,  defined according to \cref{def:PR}, is an efficient $m$-block generator that satisfies the following properties:
	\begin{description}
		\item[Real entropy:] If   the \ith  block of $\Gc$ has  real min-entropy at least $\kreal =\kreal(n)$, then the \ith  block of  $\Gt$ has real min-entropy at least $\krealp(n) = \vv \cdot \kreal - O( (\log n + \ell)\cdot\log n \cdot \sqrt{\vv})$, for  $\ell = \ell(n)$ being the  length of the \ith block.  If the bound on the real entropy $\Gc$ is invariant to the public parameter, then  this  is also the case for the above bound on $\Gt$. 
		
		\item[Accessible entropy:]  If $\Gc$ has accessible  entropy at most $\kmax = \kmax(n)$, then $\Gt$ has accessible entropy at most $\vv \cdot \kmax$.
	\end{description}
\end{lemma}

	\def\PRGen{
	
	Consider the following efficient  $\Gc$-consistent generator.
	
	\begin{algorithm}[Generator $\Gs$]
		\item Input: public parameter $z\in \zo^\pp$.
		
		\item Operation:
		
		\begin{enumerate}
			\item Sample $j\getsr [\vv]$.  Let $\vz= (z_1,\ldots,z_\ee)$ for $z_v = z$, and $z_i$,  for $i\neq v$,  sampled uniformly in  $\zo^{\pp}$. 
			
			We will refer to the part of $\vz$ sampled by the generator as $\vz_{-v}$, and (abusing notation) assume $\vz = (z,\vz_{-v})$.
			
			\item Start  a random execution of $\Gbs(\vz)$ and output the \jth entry of each output block.
			
		\end{enumerate}   
	\end{algorithm}
	It is clear that  $\Gs$ is indeed an efficient $\Gc$-consistent generator. We will show that the accessible entropy of $\Gs$ violates  the assumed  bound on the accessible entropy of $\Gc$.
}

\begin{proof}
	The bound on real entropy follows readily from  \Cref{lem:flatteningCond} by taking $\eps =  2^{-\log^2 n}$, and noting that the support size of each block of $\Gc$ is at most $\ell\cdot 2^{\ell}$. Therefore, we focus on establishing the bound on accessible entropy. Let $\Gb = \Gt$, let $\Gbs$  be  an efficient, nonuniform,  $\Gb$-consistent generator, and assume
	\begin{align}\label{eq:PR}
	k_\Gbs = \eex{\vt \getsr  \Tbs}{\Haccsam_{\Gbs}(\vt)} > \kmaxp = \vv \cdot \kmax
	\end{align}
	for $(\vZ,\vR_1,\vY_1,\ldots,\vR_m,\vY_m) = \Tbs = T_{\Gbs}(1^n)$. 	We show that \cref{eq:PR} yields that  a  random ``column'' of  $\Tbs$  contributes more than  $\kmax$ bits of  accessible entropy, and use it to construct a cheating generator for $\Gc$ whose  accessible entropy is greater than $\kmax$, in contradiction to  the assumed accessible entropy of $\Gc$.

	Let $(\vZ,\vR_1,\vY_1,\ldots,\vR_m,\vY_m)= \Tbs$. Since $\Gbs$ is   $\Gb$-consistent, each  $Y_i$ is of the form $(Y_{i,1},\ldots, Y_{i,\vv})$. It follows that 
	\begin{align}\label{eq:PrAccEnt:2}
	k_\Gbs &=  \sum_{i=1}^m \Hall (\vY_i \mid \vZ,\vR_{\le i})\\
	&  =  \sum_{i\in [m]} \sum_{j\in  [\vv]} \Hall (Y_{i,j} \mid \vZ,\vR_{<i},\vY_{i,1},\ldots,\vY_{i,j-1}  )\nonumber\\
	&  \le  \sum_{i\in [m]} \sum_{j\in  [\vv]} \Hall (Y_{i,j} \mid \vZ,\vR_{<i},Y_{<j} )\nonumber\\
	&=  \sum_{j\in  [\vv]}  k_j \nonumber
	\end{align}
	for $k_j =  \sum_{i\in [m]} \Hall (Y_{i,j} \mid \vZ,\vR_{<i} )$. 
	The first equality is by  \cref{lem:accessible-as-conditional}, and the  inequality is by  the chain rule of Shannon entropy.

	\PRGen

	 Let $(Z,\tR_1= (J,\vZ_{-J},R_1'),\tY_1,\ldots,\tR_m,\tY_m) = T_\Gs$, for  $(J,\vZ_{-J})$ be the value of $(j,\vz_{-f})$ sampled in the first step of $\Gs$  with $R_1'$ being  the randomness used by the emulated $\Gbs$ to  generate its first output bock.  Let  $(z,\tr_1 = (j,\vz_{-j},r'_1),y_1,\ldots,\tr_m,y_m) \in \Supp(T_\Gs)$. It is easy to verify that
	 
	 \begin{align}
	 \Hall_{\tY_1|Z,J,\vZ_{-J}}(y_i \mid z,j,\vz_{-j}) = \Hall_{Y_{1,j}|\vZ}(y_1\mid (z,\vz_{-f}))
	 \end{align}
	 and that  for $i>1$: 
	 \begin{align}
	 \Hall_{\tY_i|Z,\tR_{<i}}(y_i \mid z,\tr_{<i}) = \Hall_{Y_{i,j}|\vZ,\tRR_{<i}}(y_i \mid (z,\vz_{-f}), (r_1',\tr_2,\ldots,\tr_{i-1}))
	 \end{align}
	 It follows that
	 	\begin{align*}
	\lefteqn{ \eex{j \getsr [\vv]}{k_j}}\\
	 &{\eex{ (\vz,r_1,y_1,\ldots,r_{m},y_m) \getsr  \Tbs; j\getsr [\vv]}{\sum_{i \in[m]} \Hall_{Y_{i,j}|\vZ,\tR_{<i}}(y_i|\vz,r_{<i})}}\\
	 &= \eex{(z,\tr_1 = (j,\vz_{-j},r'_1),y_1,\ldots,\tr_m,y_m) \getsr  T_\Gs} {\Hall_{\tY_1|Z,J,Z_{-J}}(y_i \mid z,j,\vz_{-j}) + \sum_{i=2}^m \Hall_{\tY_i|Z,\tR_{<i}}(y_i \mid z,\tr_{<i}) } \\
	 &=\Hall(\tY_1 \mid Z,J,Z_{-J}) +  \sum_{i=2}^m \Hall(\tY_i \mid Z,\tR_{<i})\\
	 &\le \Hall(\tY_1 \mid Z) +  \sum_{i=2}^m \Hall(\tY_i \mid Z,\tR_{<i})\\
	 &= \eex{\vt \getsr  T_\Gs} {\Haccsam_\Gs(\vt)}.
	 \end{align*}
	The inequality is by the chain rule of Shannon entropy and the last equality by \cref{lem:accessible-as-conditional}. Thus,  by \cref{eq:PrAccEnt:2}, 	the accessible entropy of $\Gs$ is greater than $\kmax$, in contradiction to the assumed accessible entropy of $\Gc$.
\end{proof}

\section{Entropy Gap to Commitment}\label{sec:GapToCom}
\newcommand{\Ssh}{\Sc_{\SB}^\ast}
\newcommand{\lem}{v}
\newcommand{\Low}{\cl}
\newcommand{\rim}{{r_{\le i}}}
\newcommand{\Rim}{{R_{\le i}}}
\newcommand{\RIm}{{R_{\le L}}}
\newcommand{\li}{{\le i}}
\newcommand{\lI}{{\le I}}
\newcommand{\low}{\mathsf{low}}

In this section we construct statistically hiding commitments  from inaccessible entropy generators. Combined with the main result of \cref{sec:IAEGenFromOWF}, the above yields a construction of  statistically hiding commitments  from one-way functions, reproving the result of \cite{HaitnerNgOnReVa09}.  

\medskip

We start by recalling the definition of  statistically hiding commitment schemes.
\vspace{-0.2in}
\paragraph{Statistically hiding commitment schemes.}
A {\em commitment scheme} is the cryptographic analogue of a safe. It is a two-party protocol between a {\em sender} $\Sc$ and a {\em receiver} $\Rc$ that consists of two stages. The {\em commit stage} corresponds to putting an object in a safe and locking it; the sender ``commits'' to a private message $m$. The {\em reveal stage} corresponds to unlocking and opening the safe; the sender ``reveals'' the message $m$ and ``proves'' that it was the value committed to in the commit stage (without loss of generality by revealing coin tosses consistent with $m$ and the transcript of the commit stage).

\begin{definition}\label{def:com}
	A {\sf (bit) commitment scheme}\footnote{We present the definition for bit commitment. To commit to multiple bits, we may simply run a bit commitment scheme in parallel.} is an efficient two-party protocol $\Com = (\Sc,\Rc)$  consisting of two stages. Throughout, both parties receive the security parameter $1^n$ as input.
	\begin{quote}
		
		{\sc Commit}. The sender $\Sc$ has a private input $b \in \zo$, which it wishes to commit to the receiver $\Rc$, and a
		sequence of coin tosses $\sigma$. At the end of this stage, both parties receive as common output a \emph{commitment} $z$.
		\medskip
		
		{\sc Reveal}. Both parties receive as input a commitment $z$. $\Sc$ also receives the private input $b$ and coin tosses $\sigma$ used in the commit stage. After the interaction of $(\Sc(b,r),\Rc)(z)$, $\Rc$ either outputs a bit, or the reject symbol $\perp$.
		
	\end{quote}
	
	The  commitment is {\sf receiver public-coin} if the messages the receiver sends are merely the coins it flips at each round.

	For the sake of this tutorial, we focus on commitment schemes with a \emph{generic} reveal scheme: the commitment  $z$ is simply the transcript of the commit stage, and in the noninteractive reveal stage,  $\Sc$ sends $(b,\sigma)$ to $\Rc$, and $\Rc$ outputs $b$ if  $\Sc$, on input $b$ and randomness $\sigma$, would have acted as the sender did in $z$; otherwise, it outputs $\perp$. 
\end{definition}

Commitment schemes have two security properties. The {\em hiding} property informally states that, at the end of the commit stage, an adversarial receiver has learned nothing about the message $m$, except with negligible probability. The {\em binding} property states that, after the commit stage, an adversarial sender cannot produce valid openings for two distinct messages (\ie to both $0$ and $1$), except with negligible probability. Both of these security properties come in two flavors---{\em statistical}, where we require security even against a computationally unbounded adversary, and {\em computational}, where we only require security against feasible (\eg
polynomial-time) adversaries.

Statistical security is preferable to computational security, but it is impossible to have commitment schemes that are both statistically hiding
and statistically binding. In this tutorial, we focus on  constructing statistically hiding
(and computationally  binding) schemes, which are  closely connected to  the notion of inaccessible entropy generators.

\begin{definition}[commitment schemes]
	A commitment scheme $\Com=(\Sc,\Rc)$ is {\sf statistically hiding} if the following holds:
	\begin{quote}
		{\sc Completeness}. If both parties are honest, then for any bit $b \in \zo$ that $\Sc$ gets as private input, $\Rc$ accepts and outputs $b$ at the end of the reveal stage.
		
		\medskip
		
		{\sc Statistical  Hiding}. For  every  unbounded strategy $\Rs$, the distributions $\view_{\Rs}((\Sc(0),\Rs)(1^n))$ and $\view_{\Rs}((\Sc(1),\Rs)(1^n))$ are statistically  indistinguishable, where $\view_{\Rs}(e)$ denotes the collection of all messages exchanged and the coin tosses of $\Rs$ in $e$. 
		
		The commitment is  {\sf honestreceiver} statistically hiding, if the above is only guaranteed for $\Rs= \Rc$.
		\medskip
		
		{\sc Computational   Binding}. A \ppt  $\Ss$ succeeds in the following game (breaks the commitment) only with negligible probability in $n$:
		\begin{itemize}
			\item $\Ss = \Ss(1^n)$ interacts with an honest $\Rc = \Rc(1^n)$ in the commit stage, on security parameter $1^n$, which yields a commitment $z$.
			\item $\Ss$ outputs two messages $\tau_0,\tau_1$ such that $\Rc(z,\tau_b)$ outputs
			$b$, for both $b\in \zo$.
		\end{itemize}
		\Com is {\sf $\delta$-binding} if no \ppt  $\Ss$ wins the above game with probability larger than  $\delta(n) + \negl(n)$.
	\end{quote}
\end{definition}

We now discuss the intriguing connection between  statistically hiding commitment  and inaccessible entropy generators. Consider a statistically hiding commitment scheme
in which the sender commits to a message of length $k$, and suppose we run the protocol with the message $m$ chosen uniformly  at random
in $\zo^k$. Then, by the statistically  hiding property, the {\em real entropy} of the message $m$ after the commit stage is $k-\negl(n)$. On the other hand, the computational binding property states that the {\em accessible entropy} of $m$ after the commit stage is at most $\negl(n)$. This is only an intuitive connection, since we have not discussed real and accessible entropy for protocols, but only for generators. Such definitions can be found in \cite{HaitnerReVaWe09}, and for them it can be proven that statistically hiding commitments imply protocols in which the real entropy is much larger than the accessible entropy. Here our goal is to establish the converse, namely that a generator with a gap between its real and accessible entropy implies a statistically hiding commitment scheme.  The extension of this fact for protocols can   be found in \cite{HaitnerReVaWe09}.

\begin{theorem}[Inaccessible entropy generator to  statistically hiding commitment]\label{theorem:GapToCom}
	Let $k= k(n)$, $\pp=\pp(n)$, $s=s(n)$ be  polynomial-time computable functions and let $p = p(n)\in \poly$. Let  $\Gc$ be an efficient  $m=m(n)$-block generator over $\zo^{\pp} \times \zo^{s}$, and assume that  $\Gc$'s real Shannon entropy is at least $k$ and  its accessible entropy is at most $(1-1/p)\cdot k$. Then  for any polynomial-time computable $g= g(n)  \ge \max\set{\Hmax(G(U_\pp,U_s)),\log m}$, there exists an  $O(m pg/ k)$-round,   statistically hiding and computationally binding commitment scheme. Furthermore, if the bound on the real entropy is invariant to the public parameter, then  the commitment is receiver public-coin.

\end{theorem}

Given per $n$  polynomial-size advice, the commitment round complexity can be reduced to $O(m)$;  see \cref{remark:NonUniformCom} for details. In \cref{SHC:ConstantRound} we  use this fact to prove that an inaccessible  generator with a  \emph{constant} number of blocks yields a  constant round commitment.

Combining the  above theorem with \cref{thm:AEGfromOWF} reproves the following fundamental  result:
\begin{theorem}[One-way functions to statistically hiding commitment]\label{thm:OWFtoSHC}
	Assume there exists  one-way function $f\colon\zn \mapsto \zn$.  Then there exists an $O(n^2/\log^2 n)$-round, receiver public-coin  statistically hiding commitment scheme. 
\end{theorem}
Given  per $n$ polynomial-size advice, the round complexity of the above commitment can be reduced to $O(n/\log n)$, matching the lower bound for such fully black box constructions of \cite{HaitnerHRS15}.  

The heart of the proof of \cref{theorem:GapToCom} lies in the following lemma.

\begin{lemma}\label{lemma:GapToCom}
	Let $k(n) \ge 4n$ be   a polynomial-time computable function, let $\Gc$ be an efficient  $m$-block generator, and assume one-way functions exist. Then  for  every efficiently computable $p(n) \ge  1/\poly(n)$   there exists a polynomial-time, $O(m)$-round, receiver public-coin, commitment scheme $\Com$   with the following properties: 
	\begin{description}
		\item[Hiding:] If each block of $\Gc(U_{\pp(n)},U_{s(n)})$ has  real min-entropy at least $k(n)$, then \Com is statistically hiding.  Furthermore, if the bound on the real entropy is invariant to the public parameter, then  the commitment is receiver public-coin.
		
		\item[Binding:]

		If for  every efficient $\Gc$-consistent, online generator $\Gs$ and all large enough $n$,
		$$\ppr{\vt\getsr T_\Gs(1^n)} {\Haccsam_\Gs(\vt) > m(n) (k(n)  - 3n)} \le 1 - 1/p(n),$$
		
		then \Com is computationally binding. 
		
	\end{description}
\end{lemma}
We prove  \cref{lemma:GapToCom} in \cref{SHC:ProvingGapToCom}, but  first use it to prove \cref{theorem:GapToCom}.


\paragraph{Proving \cref{theorem:GapToCom}.}

\begin{proof}[Proof of \cref{theorem:GapToCom}]
	We prove \cref{theorem:GapToCom} by  manipulating the real and accessible entropy of  $\Gc$ using the tools described in \cref{sec:manipulatingAE}, and then applying \cref{lemma:GapToCom}  on the resulting generator. 

	\paragraph{Truncated sequential repetition: real entropy equalization.}  	In this step we use  $\Gc$ to define  a generator  $\Gc^{\brk\ee}$ that   each of whose  blocks has the  same amount of  real entropy: the average of the real entropy of the blocks of $\Gc$.   In relative terms, the entropy gap of $\Gc^{\brk \ee}$  is essentially  that of $\Gc$.  
	
	We assume \wlg that $m(n)$ is a power of  two.\footnote{Adding $2^{\ceil{\log m(n)}} -m(n)$ final blocks of constant value transforms a  block-generator  to one whose block complexity is a power of two, while maintaining the same amount of real and accessible entropy.}  Consider the  efficient  $m'= m'(n)=(\ee-1)\cdot m)$-block generator $\Gc^{\brk \ee}$ resulting by  applying  truncated sequential repetition (see \cref{def:EqRealEnt}) on $\Gc$ with parameter $\ee = \ee(n)=\max\set{4,\ceil{12gp/k}} \le \poly(n)$. By \cref{lem:EqRealEnt}: 	
	\begin{itemize}
		\item  \emph{Each  block}  of $\Gc^{\brk \ee}$ has real entropy  at least $k'= k'(n)  = k/m$.
		
		\item   The accessible entropy   of $\Gc^{\brk \ee}$ is at most  
		\begin{align*}
		a'= a'(n) &= (\ee-2)   \cdot (1-1/p)\cdot k  +  \log m + 2\Hmax(G(U_\pp,U_s)))\\
		&\le (\ee-2)  \cdot (1-1/p)\cdot k   + 3g\\
		&= (\ee-2)  \cdot (1-1/2p)\cdot k - (\ee -2) \cdot k/2p +  3g\\
		&\le (\ee-2)  \cdot (1-1/2p)\cdot k  - \ee\cdot k/4p +  3g\\
		&\le  (\ee-2)  \cdot (1-1/2p)\cdot k\\
		&<  m'  \cdot (1-1/2p)\cdot k'.
		\end{align*}
	\end{itemize}

	\paragraph{Direct product: converting real entropy to   min-entropy   and gap amplification.} In this step we use   $\Gc^{\brk \ee}$ to define a generator  $(\Gc^{\brk \ee})^{\seq \vv}$ that each of whose blocks has the same amount of real  min-entropy, about $\vv$ times the per-block real entropy of  $\Gc^{\brk \ee}$. The accessible entropy of  $(\Gc^{\brk \ee})^{\seq \vv}$  is at most  $\vv$ times the accessible entropy of $\Gc^{\brk \ee}$.  
	
	We assume \wlg  that the  output blocks are of all of the same length $\ell = \ell(n)\in \Omega(\log n)$.\footnote{A standard padding technique can be used to  transform a block-generator to one whose   blocks are all of the same length, without changing its real  and its accessible entropy.}   Let $\vv = \vv(n) = \max\set{32np/k',\ceil{c\cdot \left(\log(n) \ell p)/k'\right)^2}}$ for $c>0$ to be determined by the analysis.  Consider the efficient  $m'$-block generator $(\Gc^{\brk \ee})^{\seq \vv}$,  generated by taking the direct product of $\Gc^{\brk \ee}$ according to \cref{def:PR}. By  \cref{lem:PR}:
	\begin{itemize}
		\item Each block of $(\Gc^{\brk \ee})^{\seq \vv}$ has real \emph{min-entropy}  at least  $k'' = k''(n) =  \vv \cdot k' - O\left(\log (n)\cdot  \ell \cdot  \sqrt{\vv} \right)$.
		
		\item The  accessible entropy of $(\Gc^{\brk \ee})^{\seq \vv}$  is at most $a'' = a''(n) =  \vv \cdot a'$.
	\end{itemize}
	Hence for large enough $n$, it holds that 
	\begin{align}
	m' \cdot k'' - a''&\ge  m' \cdot \left(\vv \cdot k' - O\left(\log (n)\cdot  \ell \cdot  \sqrt{\vv} \right)\right) - v\cdot a'  \nonumber\\
	&>  m' \cdot \left(\vv \cdot k' - O\left(\log (n)\cdot  \ell \cdot  \sqrt{\vv} \right)\right) - v\cdot m'  \cdot (1-1/2p)\cdot k'\nonumber\\
	&=  m' \cdot v \cdot \left(k'/2p - O( \log (n) \cdot \ell/\sqrt{\vv}) \right)\nonumber\\
	&\ge  m' \cdot v \cdot k'/4p \nonumber\\
	&\ge 4m' n.\nonumber
	\end{align}
	The penultimate inequality holds by taking a large enough value of $c$ in the definition of $\vv$. 
	Hence,  by an averaging argument  for any efficient $(\Gc^{\brk \ee})^{\seq \vv}$-consistent, online generator $\Gs$ and all large enough $n$, it holds that
	\begin{align}
	\ppr{\vt\getsr T_\Gs(1^n)} {\Haccsam_\Gs(\vt) > m' (n)(k''(n)  - 3n)} \le 1 - 1/p'(n)
	\end{align}
	for $p'(n) = m'(n)(k''(n) - 3n)/n$.

By \cref{lemma:IEGtoOWF}, the  existence   of $\Gc$ implies that of one-way functions. 	Hence, we can apply   \cref{lemma:GapToCom} with $(\Gc^{\brk \ee})^{\seq \vv}$,  $k =k''$ and $p = p'$, to get  the claimed  $\left(m'= m\cdot (\ee-1) = O(m \cdot gp/ k)\right)$-round,   statistically hiding and computationally binding commitment.

Finally, it readily  follows from  the above proof that  if the bound on the real entropy of $\Gc$ is invariant to the public parameter, then so is that of   $(\Gc^{\brk \ee})^{\seq \vv}$. Hence, \cref{lemma:GapToCom}  yields  that in this case the resulting commitment is receiver public-coin.
\end{proof}

\begin{remark}[Comparison with the construction of next-block  pseudoentropy generators to pseudorandom generators]
	It is interesting to  see the similarity between the  manipulations we apply above on the  inaccessible entropy generator $\Gc$ to construct statistically hiding commitment, and those applied by \citet{HaitnerReVa13,VadhanZheng2012} on the next-block  pseudoentropy generator to construct a pseudorandom generator.  The manipulations applied in both constructions are essentially the same and achieve similar goals: to convert real entropy to per-block min-entropy whose overall sum is significantly larger than the accessible entropy in the above, and  to convert next-block  pseudoentropy to per-block pseudo-min-entropy whose  overall sum is significantly larger than the real entropy in  \cite{HaitnerReVa13,VadhanZheng2012}. This fact, together with the similarity in the initial steps of constructing the above generators from one-way functions  (inaccessible entropy generator above and next-block pseudoentropy generator in \cite{VadhanZheng2012}) yields that the  constructions of statistically hiding commitment schemes and pseudorandom generators from  one-way functions are surprisingly similar.     
\end{remark}

\begin{remark}[Omitting the entropy equalizing step]\label{remark:NonUniformCom}
		If the amount of real entropy of each block of   $\Gc$ is \emph{efficiently computable}, the entropy equalizing step in the   proof of \cref{theorem:GapToCom} above is not needed.  Rather, we can take a direct product  of $\Gc$  itself to get the desired generator.  This argument yields  an $\Theta(m)$-round, nonuniform (the parties use a nonuniform polynomial-size advice per security parameter)  commitment scheme,  assuming  the bound on the accessible entropy of  $\Gc$  holds for nonuniform generators. When combined with \cref{thm:AEGfromOWF}, the latter yields a $\Theta(n/\log n)$-round nonuniform commitment statistically hiding scheme from any nonuniform one-way function, matching the lower bound of \cite{HaitnerHRS15}.\footnote{The bound of \cite{HaitnerHRS15} is stated for uniform commitment schemes, but the  same bound  for nonuniform commitment schemes readily follows from their proof.}

	If the  generator's number of blocks is  \emph{constant}, the knowledge of the per-block entropy is not needed. Rather,  the above  reduction can be applied to  all possible values for the real entropy of the blocks  (up to some $1/\poly$ additive accuracy level), yielding   polynomially many commitments that are all binding and at least one of them is hiding. Such commitments can then be  combined in a standard way to get a single scheme that is statistically hiding and computationally binding. See \cref{SHC:ConstantRound} for details.

\end{remark}

\subsection{Proving \cref{lemma:GapToCom}}\label{SHC:ProvingGapToCom}
To prove \cref{lemma:GapToCom}, we use a random  block of $\Gc$ to mask the committed bit. The guarantee about the real  entropy of $\Gc$ yields that  the resulting commitment is hiding, whereas the guarantee about  $\Gc$'s  accessible entropy yields that  the commitment is weakly (\ie  $\Omega(1/mp)$) binding. This  commitment is then  amplified via parallel repetition, into  a full-fledged computationally binding and statistically hiding  commitment. In more detail, the construction of the aforementioned weakly binding commitment scheme goes as follows: $\Rc$ samples the public parameter $z$ and sends it to $\Sc$, and $\Sc$  starts (privately) computing a  random execution of  $\Gc(z, \cdot)$. At the \ith round, $\Rc$ tells $\Sc$ whether to send it the \ith  block of $\Gc$ to $\Rc$, or  to use the \ith block $y_i$ as the sender input for  a (constant round)  ``hashing" subprotocol. This subprotocol  has the following properties:
\begin{itemize}
	\item Conditioned on  $y_1,\ldots,y_{i-1}$ and the hash value of $y_i$ (\ie the transcript of the hashing protocol), the (real) min-entropy of $y_i$ is still high (\eg $\Omega(n)$), and
	
	\item if the accessible entropy of $\Gc$ in the \ith block is lower than $k - 2n$ (\ie given an adversarial  generator view, the support size of $y_i$ is smaller than $2^{k - 2n}$), then $y_i$ is \emph{determined} from the point of view of (even a cheating) $\Sc$ after sending the hash value.
\end{itemize}
Next,   $\Sc$ ``commits" to its secret  bit $b$ by masking it (via XORing) with a  bit extracted (via an inner product with a random string) from $y_i$, and the commit stage halts. 

The hiding of the above scheme  follows from the guarantee about the min-entropy of $\Gc$'s blocks. The $1/mp$-binding of the scheme follows since the bound on the accessible  entropy of $\Gc$ yields that with some probability, the accessible entropy of at least one of $\Gc$'s blocks is low, and thus the sender is bounded to a single bit if  the receiver  has chosen this block to use  for the commitment.

The aforementioned hashing protocol is defined and analyzed  in \cref{sec:GapToCom:ComHash},  the weakly binding commitment is defined in \cref{sec:GapToCom:weakcommitment}, and we put it all together to prove the lemma in \cref{sec:GapToCom:PutTogether}.\footnote{A simplified version of the   somewhat complicated hashing protocol and the resulting weak commitment defined below is given in \cite{HaitnerVadhan17}.}

\subsubsection{Strongly Binding Hashing Protocol}\label{sec:GapToCom:ComHash}
A building block  of our hashing protocol is the  following variant of ``weakly binding'' (interactive) hashing protocol of \citet{DingHRS04}.

\medskip 


\noindent
Let   $\h^1$ and  $\h^2$ be function families over $\zo^\ell$.
\begin{protocol}[Weakly binding hashing protocol 
	$(\Sc_{\WB},\Rc_{\WB})^{\h_1,\h_2}$]\label{protocol:StatHashing} ~
	
	\item[$\Sc_{\WB}$'s private input:] $x\in \zo^\ell$
	
	\begin{enumerate}
		\item $\Rc_{\WB}$ sends $h^1 \getsr \h^1$ to $\Sc_{\WB}$.
		
		\item $\Sc_{\WB}$ sends $y^1 = h^1(x)$ back to $\Rc_{\WB}$.
		
		\item $\Rc_{\WB}$ sends $h^2 \getsr \h^2$  to $\Sc_{\WB}$.
		
		\item $\Sc_{\WB}$ sends $y^2 = h^2(x)$ back to $\Rc_{\WB}$.
	\end{enumerate}
\end{protocol}
We will use two properties of the above protocol. The first, which we will use for hiding, is that if $\Sc_{\WB}$ sends only $k'$ bits to  $\Rc_{\WB}$ and $\Sc_{\WB}$'s input $x$ comes from a distribution of min-entropy significantly larger than $k'$, the input of $\Sc_{\WB}$  has   high min-entropy conditioned on $\Rc_{\WB}$'s view of the protocol (with high probability). On the other hand,
the following binding property, which we refer to as ``weak'' to distinguish it from the binding property of the final protocol,  states that if $x$ has max-entropy smaller than $\len$ (\ie is restricted to come from a set of size at most $2^\len$) and $\h_1$ and $\h_2$ are ``sufficiently'' independent  and their total output length is sufficiently larger than $k$, then after the interaction ends,  $x$ will be uniquely determined, except with exponentially small  probability.

The proof of the following fact, proved here for completeness,  follows similar lines to the proof of \cite[Theorem 5.7]{DingHRS04}:
\begin{lemma}[$(\Sc_\WB,\Rc_\WB)$ is weakly binding] \label{proposition:StatHashing}
Let $\h^1\colon \zo^\ell\mapsto \zo^{\len}$ and $\h^2\colon \zo^\ell\mapsto \zn$, and let $\Sc_{\WB}^\ast$ be an (unbounded) adversary playing the role of $\Sc_{\WB}$ in $(\Sc_{\WB},\Rc_{\WB}) = (\Sc_{\WB},\Rc_{\WB})^{\h^1,\h^2}$. Assuming  $\h^1$ ad $\h^2$ are  $t$-wise and pairwise independent hash function families,  respectively, and that $n \ge 4 (1+ \log t)$, then  the following holds for any  $2^{\len}$-size
set $\cl\subseteq \zo^\ell$:  

Let $H^1,H^2,Y^1,Y^2$ and $X = (X_0,X_1)$, be the values of $h^1,h^2,y^1,y^2$ and the final output of $\Sc_{\WB}^\ast$, in  a random execution of $(\Sc_{\WB}^\ast,\Rc_{\WB})$,  then

	$$\Pr[X_0 \neq X_1 \in \cl \   \land \ \forall j\in \zo \colon\ H^1(X_j)= Y^1 \land H^2(X_j)= Y^2] < 2^{k - \floor{t/2}} + 2^{-n/2}.$$
\end{lemma}
Namely,  with save but exponentially small probability, there are no two distinct items in $\cl$ that are consistent with the protocol transcript (\ie with the two hash values).
\begin{proof}
For $x\in \zl$ and $y\in \zo^k$, let $I^{x,y}$ be the indicator for $H^1(x)= y$. The one-wise independence of  $\h^1$ yields that $\mu \eqdef \ex{I^{x,y}} = 2^{-k}$. Let $A^y = \sum_{x\in \cl} I^{x,y}$. Since the $I^{x,y}$'s are $t$-wise independent,  \cite[Corollary 6]{KatzKo05} yields that for any $\delta>0$:
\begin{align}
\pr{A^y > \delta \mu} < \left(\frac{t^2}{e^{2/3} \delta^2 \mu^2}\right)^{\floor{t/2}}
\end{align}
Taking  $\delta = 2t/\mu$,  we get that $\pr{A^y > 2t} < 2^{-\floor{t/2}}$, and by a union bound 
\begin{align}
	\pr{\exists y\in \zo^k \colon A^y > 8k} < 2^k \cdot 2^{-\floor{t/2}}= 2^{k-\floor{t/2}}
\end{align}
Since, by assumption, $n \ge 4 (1+ \log t)$,  we deduce that 
\begin{align}
	\pr{\exists y\in \zo^k \colon A^y > 2^{n/4}} \le 2^{k-\floor{t/2}}
\end{align}
It follows that  with save but probability $ 2^{k-\floor{t/2}}$, there are at most $2^{n/4}$ elements of $\cl$ that are consistent with $H^1$ and $Y^1$. 

The pairwise independent of $\h^2$ yields that $\pr{H^2(x) = H^2(x')} = 2^{-n}$ for any $x\neq x' \in \zl$. Assume there  are at most $2^{n/4}$ elements of $\cl$ that are consistent with $H^1$ and $Y^1$.  By a union bound, with save but probability $2^{-n} \cdot (2^{n/4})^2= 2^{-n/2}$, none of these pairs collides  \wrt $H^2$. We conclude that binding  is violated with probability at most $ 2^{k-\floor{t/2}}+ 2^{-n/2}$.
\end{proof}

Our strongly binding hashing protocol is obtained by adding \emph{universal one-way hash functions} on top of the above protocol.\footnote{The following protocol is of similar flavor to (and indeed inspired by)  the protocol used by \citet{HaitnerRe07} in their transformation of ``two-phase" commitment to statistically hiding commitment. In fact, their protocol  can be seen as a special case of \cref{protocol:Hashing}, designed to work for singleton sets $\cl_v$' (see \cref{lemma:Hashing}).}

\begin{definition}[universal one-way hash functions \cite{NaorYu89}]\label{def:UOWHF}
	An efficient function family  $\FFam = \set{\FFam_n = \set{f \colon \zo^{\ell(n)} \mapsto
			\zo^{m(n)}}}_{n\in \N}$ is {\sf universal one-way (hash)} if the following holds.
	\begin{description}
		\item [Compression.] $\ell(n) > m(n)$.
		
		\item [Target Collision Resistance.] The probability that a \ppt $\Ac$
		succeeds in the following game is negligible in $n$:
		\begin{enumerate}
			\item $(x,\state)\getsr \Ac(1^n)$
			\item $f\getsr \FFam_n$
			\item $x'\getsr \Ac(x,\state,f)$ and $\Ac$ succeeds whenever $x' \neq x$ and
			$f(x')= f(x)$.
		\end{enumerate}
	\end{description}
\end{definition}
By \citet{Rompel90} (full proof in \cite{KatzKo05}; see also \cite{HaitnerHolReVaWe10}),  and the length reduction of \cite[Lemma 2.1]{NaorYu89}, the existence of  one-way functions implies that of a family of universal one-way hash functions for any poly-time computed and bounded length function $\ell$.\footnote{The Target Collision Resistance
	property of \cref{def:UOWHF} is somewhat stronger than the
	one given in \cite{KatzKo05} (and somewhat weaker than the original
	definition in~\cite{NaorYu89}). The strengthening is in allowing $\Ac$
	to transfer additional information, \ie  $\state$, between the
	selection of $x$ and finding the collision. We note that the proof
	in~\cite{KatzKo05} holds also \wrt  our stronger definition
	(and even \wrt the original definition of~\cite{NaorYu89}).}

\begin{theorem}[\cite{Rompel90, NaorYu89, KatzKo05}]\label{thm:pwf-to-uowhf}
	Assume that one-way functions exist. Then, for any  poly-time computed and bounded $\ell(n) > n$, there exists a family of universal one-way hash functions mapping strings of length $\ell(n)$ to strings of length $n$.
\end{theorem}

\noindent
Let    $\h^1 = \set{\h^1_n}$,   $\h^2 = \set{\h^2_n}$  and $\FFam = \set{\FFam_n}$ be  function families over  $\zo^{\ell(n)}$.

\begin{protocol}[Strongly  binding hashing protocol $(\Sc_\SB,\Rc_\SB)^{\h_1,\h_2, \FFam}$]\label{protocol:Hashing}
	
	\item[Common input:] $1^n$.

	\item[$\Sc_\SB$'s private input:] $x\in \zo^{\ell(n)}$.
	
	\begin{enumerate}
		\item The two parties interact in  $(\Sc_{\WB}(x),\Rc_{\WB})^{\h^1_{n},\h^2_n}$, with $\Sc_\SB$ and $\Rc_\SB$ taking the role of  $\Sc_{\WB}$ and $\Rc_{\WB}$ respectively.
		
		\item $\Rc_\SB$ sends $f\getsr \FFam_n$ to $\Sc_\SB$.
		\item $\Sc_\SB$ sends $w = f(x)$ back to $\Rc_\SB$.
	\end{enumerate}
\end{protocol}
It is clear that if $\h^1$, $\h^2$ and $\FFam$ are efficient families, then the above protocol is efficient. 
We prove following ``strong'' binding property. 

\begin{lemma}[$(\Sc_\SB,\Rc_\SB)$ is strongly binding]\label{lemma:Hashing}
	Let    $\h^1 = \set{\h^1_n}$,   $\h^2 = \set{\h^2_n}$  and $\FFam = \set{\FFam_n}$ be  efficient function families, mapping strings of length $\ell(n)$ to strings of length $\len(n)$, $n$ and $n$, respectively, and let $\Sc_{\SB}^\ast$ be a \ppt adversary playing the role of $\Sc_{\SB}$ in $(\Sc_{\SB},\Rc_{\SB})= (\Sc_{\SB},\Rc_{\SB})^{\h^1,\h^2,\FFam}$. Assuming $\h^1_n$ is $4k(n)$-wise  independent,  that $\h^2_n$ is pairwise independent,  $\FFam$ is a universal one-way hash family and that  $k(n) \in \omega(\log n)$, then the following holds for any  set ensemble  $\set{\set{\cl^v_n\subseteq \zo^{\ell(n)}}_{v\in \zo^\ast}}_{n\in \N}$ with $\size{\cl^v_n} \le 2^{\len(n)}$ for any $v,n$:

	 Let $H = (H^1,H^2), Y= (Y^1,Y^2), F, W$, $V$ and  $X = (X_0,X_1)$, be the values of $h^1,h^2,y^1,y^2,f$ and $w$, the value $\Sc_{\WB}^\ast$ outputs  \emph{before} the interaction starts, and the final output of $\Sc_{\WB}^\ast$, in  a random execution of $(\Sc_{\WB}^\ast,\Rc_{\WB})(1^n)$, then 

	\begin{align*}
	\pr{X_0 \neq X_1 \ \land  \ \set{X_0,X_1} \cap \cl^v_n  \neq \emptyset \ \land \atop   \forall j\in \zo\colon   H(X_j)= Y \land \ F(X_j) = W} = \negl(n).
	\end{align*}
 
\end{lemma}
Namely, a cheating  efficient sender committing  to a small set $\cl^v_n$ before the interaction starts   cannot find two distinct strings that are consistent with  the interaction and (even) one of which  is in $\cl^v_n$. This is a strengthening of the binding of \cref{protocol:StatHashing} (see \cref{proposition:StatHashing}), which only guarantees that there are no \emph{two} items in a  predetermined small set that are consistent  with the transcript.  Note, however, that the weak binding of  \cref{protocol:StatHashing} holds unconditionally (against any cheating strategy), where the binding  of the above protocol is only guaranteed to hold against an efficient adversary.  For our application of  constructing statistically hiding and computationally commitment, having binding against an efficient adversary suffices.
 
\begin{proof}
	Assume towards a contradiction that there exists a \ppt $\Sc_{\SB}^\ast$ that for infinitely  many $n$'s violates the ``binding'' of $(\Sc_{\SB},\Rc_{\SB})$ with success probability at least $1/p(n)$, for some $p\in \poly$. Consider the following efficient algorithm for violating the target collision resistance of $\FFam_n$.
	\begin{algorithm}[Collision finder $\ColFinder$]\label{alf:UOWHFInverter}~
		
		\paragraph{Committing stage.}
		\begin{description}
			\item[Input:] security parameter $1^n$.
		\end{description}
		
		\begin{enumerate}
			
			\item Emulate a random execution of $(\Sc_{\SB}^\ast,\Rc_{\SB})(1^n)$ until the end of the embedded execution of $(\Sc_{\WB},\Rc_{\WB})$, and denote the state of the emulated protocol by $\state$.
			
			\item Continue  the execution $(\Sc_{\SB}^\ast,\Rc_{\SB})$  until it ends, and let $\set{x_0,x_1}$ be the two values $\Sc_{\SB}^\ast$ outputs at the end of the emulation. 
			
			\item Output $(x,\state)$, for $x \getsr\set{x_0,x_1}$.

		\end{enumerate}
		
		\paragraph{Finding collision.}
		\begin{description}
			\item[Input:] $f\in \FFam_n$, $x\in \zo^\ell$ and $\state\in \zo^\ast$.
		\end{description}
		\begin{enumerate}
			\item Emulate a random execution of $(\Sc_{\SB}^\ast,\Rc_{\SB})(1^n)$ conditioned on $\state$ and $f$, and let $x_0$ and $x_1$ be the output of $\Sc_{\SB}^\ast$ in the end of the emulation.
			
			\item Output $x' \in \set{x_0,x_1} \setminus \set{x}$.
		\end{enumerate}
	\end{algorithm}
	Fix $n$  such that $\Sc_{\SB}^\ast$ breaks the binding in  $(\Sc_{\SB}^\ast,\Rc_{\SB})(1^n)$ with probability at least $ 1/p(n)$.  For a given execution of $(\Sc_{\SB}^\ast,\Rc_{\SB})(1^n)$, let $v$ be the string output by $\Sc_{\SB}^\ast$ before the interaction starts, and let $z$ be the value of the  element of $\cl^v_n$ that is consistent with the embedded execution of   $(\Sc_{\WB},\Rc_{\WB})$, setting it to $\perp$ if the number of consistent elements is not one.  Since $n \in \omega(\log k)$, otherwise $\FFam$ cannot be target collision resistant, \cref{proposition:StatHashing} yields that the probability $\Sc_{\SB}^\ast$  breaks the binding and $z\neq \perp$, is  at least $1/p(n) - \negl(n) > 1/2p(n)$. 
	
	Let $\State$ be the value of $\state$ in a random execution of $\ColFinder(1^n)$. For $s \in \Supp(\State)$, let $q(s)$ be the probability that  $\Sc_{\SB}^\ast$ breaks the binding in  $(\Sc_{\SB}^\ast,\Rc_{\SB})(1^n)$ and $z\neq \perp$, conditioned that  its state after the execution of  $(\Sc_{\WB},\Rc_{\WB})$ is $s$.  It is  easy to verify that conditioned on $\State = s$, it holds that $\ColFinder(1^n)$ finds a collision in $\FFam_n$ (\ie $x\neq x'$ and $f(x)=f(x')$) with probability at least $q(s)^2/2$. Hence, $\ColFinder(1^n)$ finds a collision with probability at least $\ex{q(\State)^2/2}$. By the Jensen inequality, the latter is at least $\ex{q(\State)}^2/2 \ge 1/8p(n)^2$, in contradiction to the  target collision resistance of $\FFam_n$.
\end{proof}

\subsubsection{Constructing Weakly Binding Commitment}\label{sec:GapToCom:weakcommitment}
We are finally ready to define the  weakly binding commitment.  Let $\Gc\colon \zo^{\pp(n)} \times \zo^{s(n)} \mapsto (\zo^{\ell(n)})^{m(n)}$ be an  $m$-block generator.  Let    $\h^1 = \set{\h^1_n}$,   $\h^2 = \set{\h^2_n}$  and $\FFam = \set{\FFam_n}$ be  function families, mapping strings of length $\ell(n)$ to strings of length $k(n)-3n$, $n$, and $n$, respectively.  The weakly binding commitment is defined as follows:

\begin{protocol}[Weakly binding, receiver public-coin  commitment scheme $\Com = (\Sc,\Rc)$]\label{protocol:WeakBinding}~
	\begin{description}
		\item[Common input:] security parameter $1^n$
		
		\item[$\Sc$'s private input:] $b\in \zo$

		\item[Commit stage:]~

		\begin{enumerate}
			
			\item $\Rc$ samples $z\getsr \zo^{\pp(n)}$ and send it to $\Sc$.
			
			\item  $\Sc$ starts (internally) an execution of $\Gc(z,x)$ for $x \getsr \zo^{s(n)}$.
			
			\item For  $i=1$ to $m$, the parties do the following:
			\begin{enumerate}

			\item The two parties interact in  $(\Sc_\SB(y_i = \Gc(z,x)_i),\Rc_\SB)^{\h^1,\h^2,\FFam}(1^n)$, with $\Sc$ and $\Rc$ taking  the roles of  $\Sc_\SB$ and $\Rc_\SB$, respectively.
			
			\item $\Rc$ flips a coin $c_i$ to be one with probability $1/(m+1-i)$, and sends it to $\Sc$.
			
			\medskip
			\it If $c_i = 0$, $\Sc$ sends $y_i$ to $\Rc$.
			
			\medskip
			Otherwise,  
			
			\begin{enumerate}

				\item $\Sc$ samples  $u\getsr  \zo^{\ell(n)}$ and sends $(\iprod{u,y_i} \xor b,u)$ to $\Rc$, for  $\iprod{\cdot,\cdot}$ being inner product modulo $2$.

				\item The parties end the execution.
			\end{enumerate}
			
		\end{enumerate}

			
		\end{enumerate}		
	\end{description}
\end{protocol}
Assuming  $\Gc$  is efficient and that $\h_1$, $\h_2$ and  $\FFam$  are efficiently computable (\ie sampling and evaluation time are polynomial in $n$),  then clearly \Com is an  efficient (poly-time computable)  correct public-message commitment scheme. It is left to  prove the hiding and binding properties of \Com.  

 In the following  let $\is$ be the round for which $c_i$ takes the value $1$.  Note that $\is$ is uniform over $[m]$.

\begin{claim}[Statistically hiding]\label{claim:GapToCom:Hiding}
	Assume  each block of $\Gc(1^n)$ has  real min-entropy at least $k(n)$. Then \Com is  honest-receiver statistically hiding. If the bound on the real min-entropy is invariant to the public parameter, then  \Com is  statistically hiding.
\end{claim}
\begin{proof}
	Fix $n\in \N$ and omit it from the notation when clear from the context. For $i\in [m]$, let $Y_i$ denote the \ith block of $\Gc(Z,X)$ for $(Z,X) \la \zo^\pp \times\zo^s$. By assumption,  $\ppr{\vy \getsr Y_{1,\ldots,i} }{\Hall_{Y_\is|Z,Y_{<i}}(\vy_\is | z,\vy_{<i})< k} =\negl(n)$. Thus  \cref{lemma:smoothEntropies} yields that there exists a random variable $Y_i'$ such that
	\begin{enumerate}
		\item  $(Z,Y_{<i},Y_i')$ is statistically  indistinguishable  from $(Z,Y_{<i},Y_i)$, and
		\item  $\Hmin(Y_i|_{Z=z,Y_{<i} = \vy}) \ge k$ for every $(z,\vy)  \in \Supp(Z,Y_{<i})$.
	\end{enumerate}	
	Let $\Rs$ be an arbitrary algorithm playing the  role of $\Rc$ in \Com that samples $z$ as instructed (\ie uniformly in $\zo^\pp$). Let  $\tY_1,\ldots, \tY_\is$ be the first $\is$ computed by $\Sc$  in a random execution of $(\Sc,\Rs)$, and let $V^\Rs$ be $\Rs$'s view  right after $\Sc$ sent $\tY_{\is -1}$ (all variables are  arbitrarily set if the execution has aborted). Since, by assumption,  $\Rs$ samples $z$ uniformly and since $V_\is^\Rs$ is a probabilistic function of the public parameter $Z$ and  $\tY_{< \is}$, there exists  a random variable $\tY_\is'$ such that
	\begin{enumerate}
		\item $(V^{\Rs},\tY_\is)$ is statistically indistinguishable from $(V^{\Rs},\tY_\is')$, and 
		\item $\Hmin(\tY_\is'|_{V^{\Rs} = v}) \ge k$, for every non-aborting view $v \in \Supp(V^{\Rs})$.
	\end{enumerate}
	
	Let $W$ be the messages sent by $\Sc$ in the embedded execution of the   interactive hashing  $(\Sc,\Rs)$. Since $\size{W} =k- n$, by \cref{lemma:smoothEntropies,prop:CondNotReduce} there exists a random variable $\tY_\is''$ such that
	\begin{enumerate}
		\item $(V^{\Rs},W,\tY_\is)$ is $(\negl(n) + 2^{-\Omega(n)})$-close to $(V^{\Rs},W,\tY_\is'')$, and
		\item $\Hmin(\tY_\is''|_{V^{\Rs} = v,W=w}) \ge n/2$, for every  non-aborting view $v \in \Supp(V^{\Rs})$ and $w\in \Supp(W)$.
	\end{enumerate}
Let  $V^\Rs_b$ denotes $\Rs$'s view at the \emph{end} of  the commit stage of  $(\Sc(b),\Rs)$. By the above observation,  the leftover hash lemma (\cref{lem:leftover}) and the two-universality of the family $\set{h_u(y) = \iprod{u,y} \colon u\in \zn}$, it holds that $V^\Rs_0$ and $V^\Rs_1$ are  statistically indistinguishable.

It is clear by the above analysis that if the bound on the real min-entropy is invariant to the public parameter, then  the hiding holds for any $\Rs$ (that  might choose the public parameter arbitrarily).
\end{proof}

\begin{claim}[Weak computational binding]\label{claim:GapToCom:BindingWeak}
	Assume $\Gc$  is efficient and that  for  every efficient $\Gc$-consistent, online generator $\Gs$ and all large enough $n$,
	$\ppr{\vt\getsr T_\Gs(1^n)} {\Haccsam_\Gs(\vt) > m (k  - 3n)} \le  1 - 1/p$. Assume further    that $\FFam$ is a family of universal one-way hash functions, that $\h_1$ and $\h_2$ are efficiently computable, and are  $(k(n)  - 3n)$-wise and pairwise independent, respectively, and that $k(n) \ge 4n$,  then \Com is $(1 - 1/3mp)$-binding.
\end{claim}

The proof of \cref{claim:GapToCom:BindingWeak} immediately follows from the next two claims.

\begin{definition}[Non-failing senders]
	A sender $\Ss$ is called {\sf non-failing} \wrt a commitment scheme $(\Sc,\Rc)$ if the following holds. Let $Z$ be the transcript of the commit stage of $(\Ss,\Rc)(1^n)$, and let $\Sigma$ be the first decommitment string  that $\Ss$ outputs in the (generic) reveal stage. Then $\pr{\Rc(Z,\Sigma)= \perp} = 0$. 
\end{definition}
That is, a non-failing sender never fails to  justify its actions in the commit stage.

\begin{claim}[Weak computational binding against non-failing senders]\label{claim:GapToComUni:BindingNonFailing}
	Let  $\Gc$, $\Com$, $\FFam$, $\h_1$ and  $\h_2$  be as in \cref{claim:GapToCom:BindingWeak}. Then \Com is $(1-1/2mp))$-binding against {\sf non-failing} senders. 
\end{claim}

\begin{claim}\label{claim:GapToComUni:FronNonFailingToArbitrary}
	Assume a  receiver public-coin  commitment scheme is $\alpha$-binding  against non-failing senders. Then it is $(\alpha + \negl)$-binding.
\end{claim}

\paragraph{Proving \cref{claim:GapToComUni:BindingNonFailing}.}
\begin{proof} Assume towards a contradiction that there exists a non-failing \ppt sender  $\Ss$ that breaks the $(1 - 1/2mp)$-binding of \Com.  We use $\Ss$ to construct an efficient, $\Gc$-consistent generator $\Gs$ that breaks the assumed bound on the accessible entropy of $\Gc$. We assume for  simplicity that $\Ss$ is deterministic.

	Fix $n \in \N$ for which $\Ss$ breaks the binding with probability at least $1 - 1/2mp(n)$, and omit $n$ from the notation when clear from the context. The following  generator   uses the  ability of $\Ss$ to  break the binding  of the embedded hashing protocol at all rounds, induced by its high probability of breaking the binding, to output  high sample-entropy transcript.
	
	\begin{samepage}
		\begin{algorithm}[$\Gs$---High entropy  generator from cheating sender $\Ss$]\label{alg:HEGen}
			\item[Security parameter:] $1^n$.
			\item [Input:] public parameter $z$.
			
			\item[Operation:]~

			\begin{enumerate}
						
				\item  Start a random execution of  $(\Ss,\Rc)(1^n)$ with $\Rc$'s first message  set to $z$  and $\is =m$.
				  
				\item For  $i=1$ to $m-1$: output the value of $y_i$ sent by $\Ss$ at round $i$ (as the \ith output block). 
				
				
				\item  Continue the emulation of $(\Ss,\Rc)$  until its  end. Let $\sigma=(\cdot,r)$ be the first decommitment string output by $\Ss$. Output $\Gc(r)_m$ as the \tth{m} output block.

			\end{enumerate} 
		\end{algorithm}
	\end{samepage}
	
	The efficiency of $\Gs$ is clear, and since $\Ss$ is non-failing,  it is also clear that $\Gs$ is $\Gc$-consistent.  In the rest of the proof we show that the computational binding  of $(\Sc_{\SB},\Rc_{\SB})^{\h^1,\h^2,\FFam}$ yields that $\Gs$  violates the assumed bounds on the accessible entropy of $\Gs$. 
	
	Let $T =  (Z,R_1,Y_1,\ldots,R_m,Y_m) =T_\Gs(1^n)$. That is,  $R_i$  are  the coins $\Rc$ uses in the \ith round of the above emulation, \ie its coins  used in the \ith invocation of $(\Sc_\SB,\Rc_\SB)$.

	For  $\vt = (z,r_1,y_1,\ldots) \in \Supp(T)$ and $i\in [m]$, let  $\Low_{\vt,i}$ be the set of all low-entropy \ith block of $\Gs$ given $r_{<i}$. That is, 
	\begin{align}\label{eq:LisSmall}
	\Low_{\vt,i}= \Low_{z,r_{<i}} \eqdef \set{y\colon \Hall_{Y_i|Z,R_{<i}}(z,y|r_{<i}) \le k - 3n}
	\end{align}
	We conclude the  proof  by showing that
	\begin{align}\label{eq:BindingNonFailing:1}
	\ppr{\vt =  (\ldots,y_i,\ldots) \la T}{\exists i\in [m] \colon  y_i \in  \Low_{\vt,i}} < 1/p
	\end{align}
in contradiction to the assumed bound on the accessible entropy of $\Gc$.

	
	Assuming \cref{eq:BindingNonFailing:1} does not hold,  we show that the assumption about the success probability of $\Ss$  yields an algorithm for breaking the computational binding of $(\Sc_\SB,\Rc_\SB)= (\Sc_{\SB},\Rc_{\SB})^{\h^1,\h^2,\FFam}$.  The idea is that when $y_i \in \Low_{\rim}$, then  for breaking the commitment for $\is=i$, the cheating sender $\Ss$ has to break the binding of $(\Sc_\SB,\Rc_\SB)$ \wrt the small, by definition, set $\Low_{\rim}$.  

	For $i\in [m]$ and  $(z,\rim)\in \Supp(Z,\Rim)$, consider the execution of  $(\Ss,\Rc)(1^n)$  induced by $z$, $\rim$ and $\is =i$:  the coins used  by $\Rc$ in the first \tth{j} execution of $(\Sc_\SB,\Rc_\SB)$, for $j\in [i]$, are set to $r_j$,  $c_1=\ldots =c_{i-1} =0$  and $c_i=1$.

	 Let $\hY_{z,\rim,0} = \Gc(z,s_0)_i$, for $\tau_0=(s_0,\cdot,0),\tau_1=(s_1,\cdot,1)$ being  the two  strings output by $\Ss$ at the end of the above interaction. Note that since  $\Ss$ is non-failing,  $s_0$ is always consistent with the interaction, and in particular $\hY_{z,\rim,0}$ is well defined. Similarly, if $\tau_1$ is a valid decommitment, let $\hY_{z,\rim,1} = \Gc(z,s_1)_i$; otherwise, let $\hY_{z,\rim,1} = \hY_{z,\rim,0}$.

	Let $\low(\vt)$ be the smallest  value of $i$ for which  $y_i \in \Low_{z,\rim}$, set to $\perp$ if there is no such $i$, and let $L = \low (T)$. The assumption  that \cref{eq:LisSmall} does not holds implies that $\pr{L \neq \perp} \ge 1/p$.  Thus, 
	\begin{align}\label{eq:LisSmall2}
	\pr{\hY_{Z,\RIm,0}\neq \hY_{Z,\RIm,1} \in \zo^\ell \mid   L \neq \perp}\ge 1/2
	\end{align}
	Indeed, if \cref{eq:LisSmall2} does not hold, then  $\Ss$ fails to break the commitment with probability at least $\pr{L= \is} \cdot 1/p = \pr{L \neq \perp} \cdot 1/m \cdot 1/2 \ge  1/2mp$, in contradiction to the  assumed success probability of $\Ss$. It follows that  
	\begin{align}\label{eq:BindingNonFailing:2}
	\pr{\hY_{Z,\RIm,0}\neq \hY_{Z,\RIm,1} \in \zo^\ell \land Y_L \in \Low_\RIm} \ge  \pr{L \neq \perp} \cdot 1/2  \ge 1/2p
	\end{align}
	We conclude the proof by using the above observation to  define an algorithm for breaking the  computational binding of $(\Sc_{\SB},\Rc_{\SB})$. 
	
	\begin{algorithm}[Algorithm $\Ssh$ for breaking the   binding of $(\Sc_{\SB},\Rc_{\SB})$.]\label{alg:BindingNonFailing}~
		
		\begin{description}
			\item[Input:] security parameter $1^n$.
		\end{description}
		
		\begin{enumerate}
			\item Sample $\vt = (z,r_1,\ldots) \la T_\Gs(1^n)$ and $i\getsr [m=m(n)]$. Output $v= r_{<i}$.
			
			\item Emulate $(\Ss,\Rc)(1^n)$ for its first $(i-1)$  rounds, with $\Rc$'s first message set to $z$,   and  the coins used  by $\Rc$ in the first \tth{j} execution of $(\Sc_\SB,\Rc_\SB)$, for $j\in [i-1]$, are set to $r_j$, and  $c_1=\ldots =c_{i-1} =0$.
			
			\item Interact in $(\Sc_{\SB},\Rc_{\SB})^{\h^1,\h^2,\FFam}(1^n)$, by forwarding $\Rc$'s  messages to $\Ss$, and $\Ss$'s answers back to $\Rc$.\label{step:BindingNonFailing:H}
			
			\item Send $c_i = 1$ to  $\Ss$. Let  $(\tau_0=(s_0,\cdot,0),\tau_1=(s_1,\cdot,1))$ be  the two decommitment strings  output by $\Ss$. 
			
			\item Set $y_{i,0} = \Gc(z,s_0)_i$. If  $\tau_1$ is a valid decommitment set $y_{i,1} = \Gc(z,s_1)_i$; otherwise, set $y_{i,1} = y_{i,0}$.
			
			\item If $i=m$, let $y_i = y_{i,0}$.
			
			  Otherwise,
			  
			  \begin{enumerate}
			  	\item Rewind $\Ss$ to its state just before it received  the message $c_i = 1$ above.
			  	\item Send $c_i = 0$ to  $\Ss$. Let $y_i$ be the next message sent by $\Ss$.
			  \end{enumerate}
			\item Output   $x_0 = y_i$ and  $x_1 \la \set{y_{i,0},y_{i,1} }$.

		\end{enumerate}
	\end{algorithm}
   Since $\Ss$ is non-failing, the pair $(x_0,x_1)$ output by $\Ssh$ is always consistent with its  interaction with $\Rc_{\SB}$ (happens in \Stepref{step:BindingNonFailing:H}).   In addition, for infinitely many $n$'s, it holds that
   
	\begin{align*}
	\pr{x_0 \in \Low_v \land x_0 \neq x_1} & \ge  \pr{i = \low(\vt) } \cdot  \pr{x_0 \in \Low_{z,\rim} \land x_0 \neq x_1 \mid i = \low(\vt) } \\
	&\ge \frac1{mp} \cdot  \frac12 \cdot  \pr{x_0 \in \Low_{z,\rim}\mid i = \low(\vt)} \cdot \pr{y_{i,0} \neq y_{i,1}\mid i = \low(\vt)} \\
	&\ge \frac1{2mp} \cdot  \frac12 \cdot 1 \cdot \pr{y_{i,0} \neq y_{i,1}\mid i = \low(\vt)} \\
	&\ge \frac1{4mp} \cdot  \frac12= \frac1{8mp}.
	\end{align*}
	The last inequality is due to \cref{eq:BindingNonFailing:2}.  Since by definition $\size{\Low_{\rim}} \le 2^{k(n) - 3n}$, algorithm $\Ssh$ violates the  soundness of  $(\Sc_{\SB},\Rc_{\SB})^{\h^1,\h^2,\FFam}$   guaranteed by \cref{lemma:Hashing}. 
\end{proof}

\paragraph{Proving \cref{claim:GapToComUni:FronNonFailingToArbitrary}.}
\begin{proof}
	Let $\Com= (\Sc,\Rc)$  be a  receiver public-coin commitment scheme, and assume there exists an efficient  cheating sender $\Ss$ that breaks the binding of $\Com$ with probability at least $\alpha(n) + 1/p(n)$, for some $p\in \poly$ and infinitely many $n$'s. We construct an efficient non-failing sender $\Sss$ that breaks the  binding of \Com with probability $\alpha(n) + 1/2p(n)$, for infinitely many $n$'s.  It follows that if $\Com$ is $\alpha(n)$-binding for non-failing senders, then it is $\left(\alpha(n) + \negl(n)\right)$-binding.

	We assume for simplicity that $\Ss$ is deterministic, and define the non-failing sender  $\Sss$ as follows: $\Sss$ starts acting as $\Ss$, but before forwarding the \ith message $y_i$  from $\Ss$ to $\Rc$, it first makes sure it will be able to ``justify'' this message --- to output an input for $\Sc$ that is consistent with $y_i$, and the message $y_1,\ldots,y_{i-i}$ it sent in the previous rounds.  To find such a justification string, $\Sss$ continues, in its head, the interaction between  the emulated $\Ss$ and $\Rc$ until its end, using fresh coins for the receiver's messages. Since the  receiver is public-coin, this efficient  random continuation has  the same distribution as  a (real) random continuation of $(\Ss,\Rc)$ has. The sender $\Sss$ applies such random continuations  polynomially many times, and if following  one of them $\Ss$ outputs a valid decommitment string (which by definition is a valid justification string), it keeps it for future use, and outputs $y_i$ as its \ith message. Otherwise (\ie it failed to find a justification string for $y_i$), $\Sss$ continues as the honest $\Sc$ whose   coins and input bit are set to the  justification string $\Sss$ found in the previous round. 

	Since $\Sss$ maintains the invariant that it can always justify its messages, it can also do that at the very end of the commitment stage, and thus outputting this string makes it a non-failing sender. In addition, note that $\Sss$  only fails to find a justification string if $\Ss$ has a very low probability to open the commitment at the end of the current interaction, and thus very low probability to cheat.  Hence, deviating from $\Ss$ on such transcripts will only slightly decrease the cheating probability of $\Sss$ compared with that of $\Ss$.

	Assume for concreteness that $\Rc$ sends the first message in $\Com$. The non-failing sender $\Sss$ is defined as follows:

	\begin{algorithm}[Non-failing sender $\Sss$ from failing sender $\Ss$]
		\item[Input:] $1^n$
		
		\item[Operation:]~

		\begin{enumerate}

			\item Set $w = (0^{s(n)},0)$, for $s(n)$ being a bound on the  number of coins used by $\Sc$, and set  $\Fail= \false$.
			
			
			\item Start an execution of $\Ss(1^n)$.

			\item Upon getting the \ith message $q_i$ from $\Rc$, do:

			\begin{enumerate}

				\item If $\Fail = \false$,
				
				\begin{enumerate}\setlength\itemsep{3pt}
					\item Forward $q_i$ to $\Ss$, and continue the execution of $\Ss$ until it sends its \ith message.
					
					
					\item // {\sf Try and get a justification string for this \ith message.}\vspace{.05in}
					
					Do the following for $3np(n)$ times:\label{step:loop}
					
					\begin{enumerate}
						
						\item Continue the execution of $(\Ss,\Rc)$ until its end, using uniform random messages for $\Rc$. 
						
						\item Let $z'$ and $w'$  be the transcript  and  first message  output by $\Ss$, respectively,  at the end of this execution.
						
						\item Rewind $\Ss$ to its state right after sending its \ith message.
						
						
						
						\item // {\sf Update the justification   string.}\vspace{.05in}

						If $\Rc(z',w') \neq \perp$.  Set $w= w'$ and break the loop.
					\end{enumerate}

					\item If the maximal number of attempts has been reached, set $\Fail = \true$. 
				\end{enumerate} 
				
				
				\item // {\sf Send the \ith message to $\Rc$. 
					
					If $\Fail = \false$, this will be the message sent by $\Ss$ in Step $3(a)$. Otherwise, the string will be computed according to the justification string found in a previous round.}\vspace{.05in}
				
				Send $a_i$ to $\Rc$, for $a_i$ being  the \ith message that $\Sc(1^n,w)$ sends to $\Rc$ upon getting the first $i$ messages sent by $\Rc$.
			\end{enumerate}

			\item If $\Fail =  \false$, output the same value that $\Ss$ does at the end of the execution.

			Otherwise, output $w$.		
			
		\end{enumerate} 
	\end{algorithm}
	It is clear that $\Sss$ is non-failing and runs in polynomial time. It is left to show that it breaks the binding of $\Com$ with high enough probability. We do that by coupling a random  execution of $(\Sss,\Rc)$ with that of $(\Ss,\Rc)$, by letting $\Rc$ send the same, uniformly chosen, messages in both executions. We will show that the probability that $\Ss$ breaks the binding, but $\Sss$ fails to do so, is at most  $1/3p(n) +  m\cdot 2^{-n}$, for $m$ being the round complexity of $\Com$. If follows that, for infinitely many $n$'s, $\Sss$ breaks the binding of $\Com$ with probability $\alpha(n) + 1/2p(n)$.
	
	Let  $\delta_i$ denote  the  probability of $\Ss$ to break the binding after sending its \ith message,  where the probability is over the messages to be sent by $\Rc$ in the next rounds.  By  definition of $\Sss$, the probability that $\delta_i \ge 1/3p(n)$ for all $i\in [m]$, and yet $\Sss$ set $\Fail =\true$, is at most $m\cdot 2^{-n}$.  We conclude that the probability that $\Sss$ does not break the commitment,  and yet $\Ss$ does, is at most  $1/2p(n) +  m\cdot 2^{-n}$.
\end{proof}

\subsubsection{Putting it Together}\label{sec:GapToCom:PutTogether}
Given the above, we prove  \cref{lemma:GapToCom} as follows:

\begin{proof}[Proof of \cref{lemma:GapToCom}]
We use efficient $\ell$-wise function family $\h^1 = \set{\h^1_n}$ and pairwise  function family  $\h^2 = \set{h^2_n}_{n\in \N}$ mapping strings of length $\ell(n)$ to strings of length $k(n)-3n$ and $n$, respectively (see \cite{CarterWe79,CarterWe81} for  constructions of such families). Since, by assumption,   one-way functions exist,  we use \cref{thm:pwf-to-uowhf} to construct universal hash function families $\FFam$ mapping strings of length $\ell(n)$ to strings of length $n$.

	\cref{claim:GapToCom:Hiding,claim:GapToCom:BindingWeak}  yield that the  invocation of  \cref{protocol:WeakBinding} with the generator $\Gc$ and the above function families is  an $O(m)$-round,  receiver public-coin commitment scheme $\Com$  that is  honest-receiver statistically hiding if the real entropy of $\Gc$ is sufficiently large, and is $(1-\Theta(1/mp))$-binding if the generator  accessible entropy is sufficiently small.  Let $t = \log(n)^2 pm$ be an efficiently computable function and let  $\Com^{\seq {t}} = (\Sc^{\seq {t}},\Rc^{\seq {t}})$ be the $t$ parallel  repetition of \Com: an execution of $(\Sc^{\seq {t}}(b),\Rc^{\seq {t}})(1^n)$ consists of  $t$-fold   parallel and independent executions of $(\Sc(b),\Rc)(1^n)$.  It is easy to see that since \Com is honest receiver statistically hiding, so is  $\Com^{\seq {t}}$. Finally,  since $\Com$ is $(1-\Theta(1/pm))$-binding,  by \cite{HastadPPW10} (recall that \Com is receiver public coin) $\Com^{\seq {t}}$ is  computationally binding. By \cite[Cor 6.1]{HaitnerHoKaKoMoSh09Jor}, this yields an  $O(m)$-round statistically hiding protocol.
	
	Assuming the bound on the real entropy of $\Gc$ is invariant to the public parameter,   	\cref{claim:GapToCom:Hiding} yields that $\Com^{\seq {t}}$ is already statistically hiding. Thus, we do not have to use  the (non-public-coin) reduction of \cite{HaitnerHoKaKoMoSh09Jor}, and immediately get a receiver public coin  statistically hiding commitment.

\end{proof}

\subsection{Constant-Round Commitments}\label{SHC:ConstantRound}
In this section we prove that an   inaccessible entropy generator of constant number of blocks yields a constant-round statistically hiding commitment.
 
\begin{theorem}[Inaccessible entropy generator to  statistically hiding commitment, constant-round version]\label{theorem:GapToComConst}
	Let  $\Gc$ be an efficient  block-generator  with a constant number of blocks.  Assume   $\Gc$'s real Shannon entropy 
	is at least $k(n)$ for some efficiently computable function $k$, and that its accessible entropy is bounded by $k(n) - 1/p(n)$ for some  $p\in \poly$. Then  there exists a constant-round  statistically hiding and computationally binding commitment scheme. Furthermore, if the bound on the real entropy is invariant to the public parameter, then  the commitment is receiver public-coin.
\end{theorem}

The heart of the proof of \cref{theorem:GapToComConst} lies in the following lemma. In the following we use the natural generalization of commitment schemes for nonuniform protocols: the correctness, binding and hiding hold for adversaries seeing the nonuniform advice.  

\newcommand{\tk}{\widetilde{k}}
\begin{lemma}\label{lemma:GapToComConst}
	Let $\Gc$ be an efficient  block-generator  with a constant number of blocks $m$ and block length $\ell = \ell(n)$, and assume one-way functions exist. Then  for  every efficiently computable $p(n) \ge  1/\poly(n)$ there exists a polynomial-time, $O(m)$-round,  commitment scheme $\Com$   such that the following holds for any  polynomial size $\set{\tk_n = (\tk_n(1),\ldots,\tk_n(m))}_{n\in \N}$. 
	\begin{description}
		
		\item [Correctness:] $\Com(1^n,\tk_n)$ is  correct.
		
		\item[Hiding:] If for each $n\in \N$ and $i\in [m]$, either $\tk_n(i)=0$ or  the \ith block of $\Gc(U_{\pp(n)},U_{s(n)})$ has  real min-entropy at least $\tk_n(i) \ge 3n$, then $\Com(1^n,\tk_n)$ is statistically hiding.  Furthermore, if the bound on the real entropy is invariant to the public parameter, then  the commitment is receiver public-coin.
		
		\item[Binding:]

		If for  every efficient $\Gc$-consistent, online generator $\Gs$ and all large enough $n$,
		$$\ppr{\vt\getsr T_\Gs(1^n)} {\Haccsam_\Gs(\vt) > \left(\sum_{i\in [m]} \tk_n(i)\right) - 3mn} \le 1 - 1/p(n),$$
		
		then $\Com(1^n,\tk_n)$ is computationally binding. 
		
	\end{description}
\end{lemma}
\begin{proof}
The proof of 	\cref{lemma:GapToComConst} follows the same line as the proof of 	\cref{lemma:GapToCom}, where in the \ith round, the  $\ell(n)$-wise independent hash function outputs $\tk_n(i) -n$ bits (rather than $k(n)-n$ as in   \cref{lemma:GapToCom}). If $\tk_n(j)=0$, the \jth block is skipped. The correctness and hiding are clear,  and the binding holds since the nonuniform advice $\tk_n$ is of logarithmic size and thus  only improves the probability of outputting a high entropy block by a (fixed) polynomial factor (an improvement that can  be accommodated by increasing the number of repetitions of the weakly binding commitment). 
\end{proof}
Given \cref{lemma:GapToComConst}, the proof of \cref{theorem:GapToComConst} is similar to that of \cref{theorem:GapToCom} adapted to exploit the constant number of blocks. 
\begin{proof}[Proof of \cref{theorem:GapToComConst}]

	Let $\pp = \pp(n)$ and $s = s(n)$ be the pubic parameter and input length of $\Gc$. We assume \wlg  that the  output blocks are of all of the same length $\ell = \ell(n)\in \Omega(\log n)$. We also assume for simplicity that $k$, the bound on the real entropy of $\Gc$  is an integer function.    For $n\in \N$, let $Z_n \getsr \zo^\pp$ and let $(Y_1(n),\ldots,Y_m(n)) = \Gc(Z_n,U_{s(n)})$, and for $i\in [m]$, let $f_n(i) = \floor{\Hall_{Y_i(n)| Z,Y_{<i}(n)}}_{1/2mp}$, for $\floor{x}_\delta \eqdef \floor{x/\delta}\cdot \delta$, and let $f = f(n) = \sum_{i \in [m]} f_n(i)$. By definition, $f(n) \ge k(n)(1 - 1/2p(n))$.

	In the following we omit $n$ when clear from the context. Let $\Gc^{\seq \vv}$ be the   direct product of $\Gc$ (see \cref{def:PR}),  with $\vv = \vv(n) = \max\set{32n mp,\ceil{c\cdot \left(\log n\cdot \ell mp \right)^2}}$, for $c>0$ to be determined by the analysis. By  \cref{lem:PR}:
	\begin{itemize}
		\item The \ith block of $\Gc^{\seq \vv}$ has real \emph{min-entropy}  at least  $f_n(i) =  \vv \cdot f_n(i) -  c'\log (n)\cdot  \ell \cdot  \sqrt{\vv}$ for some universal constant $c'$. In particular (by taking large enough $c$ in the definition of $\vv$),   $f_n(i) > 0$ implies  $f_n'(i)\ge 3n$.
		
		\item The  accessible entropy of $\Gc^{\seq \vv}$  is at most $a' =  \vv \cdot (k - 1/p)$.
	\end{itemize}
	Let  $f' = f'(n) = \sum_{i \in [m]} f'_n(i)$. The above yields that for large enough $n$,  
	\begin{align}
	f' - a'&\ge  \vv \cdot \left (k (1- 1/2p) - O\left(\log (n)\cdot  \ell/\sqrt{\vv} \right) -  k(1-1/p)\right)\\
	&\ge   \vv k /4p\nonumber\\
	&\ge 4m n.\nonumber
	\end{align}
	The penultimate inequality holds by taking a large enough value of $c$ in the definition of $\vv$. 
	Hence  by an averaging argument,  for any efficient $\Gc^{\seq \vv}$-consistent, online generator $\Gs$ and all large enough $n$, it holds that
	\begin{align}
	\ppr{\vt\getsr T_\Gs(1^n)} {\Haccsam_\Gs(\vt) > f'- 3mn} \le 1 - 1/p'(n)
	\end{align}
	for $p'(n) = (f'(n)  - 3mn)/n$.

	For $n\in \N$, let $\ck_n$ be the set of all possible values for $k_n$,  \ie all tuples of the form $\tk_n= (\tk[1],\ldots,\tk[m])$ with $\tk[i] \in \set{0,1/2mp,\ldots,\floor{k(n)}_{1/2mp}}$ and $\sum_{i\in [m]} \tk[i] \in [k(n)(1-1/2p(n),k(n)]$.  Since $m$ is constant, $ \size{\ck_n} \in \poly(n)$. For $\tk_n= (\tk[1],\ldots,\tk[m]) \in \ck_n$, define $\tk'_n= (\tk'[1],\ldots,\tk[m]')$ by $\tk'[i] = \min\set{0,\vv \tk'[i] - c'\log (n)\cdot  \ell \cdot  \sqrt{\vv}}$, for $c'$ being the constant from \cref{lem:PR}.

 By \cref{lemma:IEGtoOWF}, the existence  of $\Gc$ implies that of one-way functions. 	Hence, for any large enough $n$ and  $\tk_n \in \ck_n$,  we have that $\Com(1^n,\tk'_n)$  is:
 
 \begin{enumerate}
 	\item correct,
 	
 	\item computationally binding, and 
 	
 	\item statistically hiding   if $\tk_n = k_n$.
 \end{enumerate}
 Consider the commitment scheme that on security parameter $n$, the parties invoke $\Com(1^n,\tk'_n)$ in parallel, for all  choices of $\tk_n \in \ck_n$, where the  committed values  used by the sender are $ \size{\ck_n}$-out-of-$ \size{\ck_n}$ shares of the value the sender wishes to commit to (\ie their XOR is the committed value). By the above observation, the resulting commitment is the desired constant-round statistically hiding commitment.

Finally, it readily  follows from the above proof that  if the bound on the real entropy of $\Gc$ is invariant to the public parameter, then so is that of   $\Gc^{\seq \vv}$, and the resulting commitment is receiver public-coin.
\end{proof}

\section{One-Way Functions are Necessary for  an Accessible Entropy Generator}\label{sec:ATEtoOWF}
In \cref{sec:IAEGenFromOWF}, we  proved  that  the existence of one-way functions implies that of an inaccessible entropy  generator. The following theorem states that  the converse direction is also  true. 
\begin{theorem}\label{lemma:IEGtoOWF}
	Let $\Gc$ be an efficient block generator with real entropy $k(n)$. If  $\Gc$  has accessible entropy   at most $k(n) - 1/p(n)$ for some $p\in\poly$, then one-way functions exist.
\end{theorem}

\begin{proof}
	The proof is  by reduction. We assume that one-way functions do not exist, and show that this  implies that $\Gc$ does not have  a noticeable gap between its real and accessible entropy.  For simplicity, we assume that $\Gc$ gets no public parameter. Consider the  efficient function that outputs the first $i$ blocks of  of $\Gc$:
	\begin{align}
g(x,i)= \Gc(x)_{1,\ldots,i}
	\end{align}
	We assume for ease of notation that the seed length of $\Gc$ on security parameter $n$  is just $n$ (\ie $s(n) = n$). Hence,  $g$ is defined over $\zn \times [n]$. 
	
	 Let  $m$ and $\ell$ be the block complexity and maximal block length of $\G$ respectively, and let $\alpha(n)  = 1/\left(20 m^3(n)p^2(n)(\ell(n)+n)\right)$.
	 Assuming  one-way functions do not exist,  by \cref{lemma:NoDistOWF}   there exists an efficient algorithm $\Inv$ that is an $\alpha$-inverter  for $g$ (see \cref{def:dowf}) on infinitely many $n$'s.  We assume \wlg that $\Inv$ either outputs a valid preimage of $g$, or $\perp$. Consider the following on-line generator $\Gs$:

	\begin{algorithm}[Online generator $\Gs$]~
		\item[Input: ] $1^n$.
		
		\item[Operation:]~
		\begin{enumerate}
			\item Let $\Fail = \false$ and  $x_0 = 0^n$.
			\item For $i=1$ to $m(n)$:
			\begin{enumerate}

			\item If  $\Fail$, set $x_i = x_{i-1}$. 
			
			Otherwise, 
			\begin{enumerate}
				
				\item  Let $(x_i,\cdot) \la \Inv(\Gc(x_{i-1})_{1,\ldots,i-1})$.
				
				\item If $x_i =\perp$, set $\Fail =\true$ and $x_i = x_{i-1}$.
	 
			\end{enumerate}
		
		 \item Output $\Gc(x_i)_i$.
		 
		 	\end{enumerate}
		\end{enumerate}
	\end{algorithm}
	 It is clear that $\Gs$ is efficient and $\Gc$-consistent. In the following we show that  for infinitely many $n$'s, the accessible entropy of $\Gs$ is at least  $k(n) - 1/p(n)$, contradicting the assumed bound on the accessible entropy  of $\Gc$.
	 
	 Let $\I \subseteq \N$ be the infinite sequence of input lengths on which $\Inv$ is an $\alpha$-inverter of $g$. Fix $n\in \I$ and omit it from the notation when clear from the context. Let $Y= (Y_1,\ldots,Y_m) = \Gc(U_n)$ and let $(R_1,\tY_1,\ldots,R_m,\tY_m) = \Ts = T_{\Gs}(1^n)$.   We first prove that  $\Hall(\tY = (\tY_1,\ldots,\tY_m))$ is almost as high as the real entropy of $\Gc$.
	 
	 \begin{claim}\label{claim:IEGtoOWF:1}
	 $\Hall(\tY) \ge \Hall(Y)-1/3p$. 
	 \end{claim}
  \begin{proof}
  	Since  $\Inv$ is an $\alpha$-inverter, a simple coupling argument yields that 
  	 \begin{align}\label{eq:SD}
  	 \delta \eqdef   \SD(Y, \tY) \leq 2m \alpha < 1/6pm\ell.
  	 \end{align}
  	Hence, \cite[Fact 3.3.9]{Vadhan-thesis} yields that  
	 \begin{align}\label{eq:RealEntofGs}
	 \Hall(\tY)  - \Hall(Y)  \ge - \delta\cdot \log( \size{\Supp(\tY \cup Y)} - h_2(\delta),
	 \end{align}
	 for $h_2(x)$ being the \Renyi entropy of the Boolean random variable taking the value $1$ wp $x$. Since $\Gs$ is non-failing, it holds that   $\Supp(\tY) \subseteq \Supp(Y)$, and by assumption,   $\size{\Supp(Y)} \le 2^{m\ell}$. It follows that 
	 \begin{align}
	  \Hall(\tY) - \Hall(Y) &\ge -\delta  m \ell -  h_2(\delta)\\
	  &\ge  - \delta  m \ell   - 2\delta \nonumber\\
	  &\ge - (m \ell  +2)/4pm\ell\nonumber\\
	  & \ge -1/3p.\nonumber
	 \end{align}
	 The second inequality holds since $h_2(\delta) < 2\delta$, and the third one since  \wlg $\ell m \ge2$.
 \end{proof}
	 
	   Recall that by definition,
	  
	   \begin{align}
	   \Haccsam_{\Gs}(\Ts) = \eex{(r_1,y_1,\ldots) \la \Ts}{\sum_{i=1}^{m} \Hall_{\tY_i|R_{<i}}(y_i|r_{<i})} 
	   \end{align}
	   and that,  by the chain rule, 
	   
	   \begin{align}
	   \Hall(\tY) = \eex{(r_1,y_1,\ldots) \la \Ts}{\sum_{i=1}^{m} \Hall_{\tY_i|\tY_{<i}}(y_i|y_{<i})}
	   \end{align}
	   We complete the proof showing that with save but very small probability over the choice of $(r_1,y_1,\ldots) \la \Ts$, it holds that $\Hall_{\tY_i|R_{<i}}(y_i|r_{<i})$ is very close to $\Hsam_{Y_i|Y_{<i}}(y_i|y_{<i})$ for every $j\in [m]$, and that the   complementary event does not contribute much to the entropy of $Y$. We do that by  focusing on the set of ``good'' transcripts  $\cs\subseteq \Supp(\Ts)$. The set $\cs$ contains  all   transcripts $\vt=(r_1,y_1,\ldots,r_m,y_m)$ such that 
	     
	   \begin{enumerate}
	   	\item  $\Fail(\vt) = \false$,  for $\Fail(\vt)$ being the event for which  the  flag $\Fail$ is set to $\true$ in the execution of  $\Gs$ reflected in  $\vt$.  
	   	
	   	\item  $\Hall_{\tY_i|R_{<i}}(y_i|r_{<i}) +1/3mp \ge \Hall_{\tY_i|\tY_{<i}}(y_i|y_{<i})$, for every $i\in [m]$.
	   	
	   \end{enumerate}
	  	We first prove that a random transcript is likely to be  in $\cs$. 
	  
	   \begin{claim}\label{claim:IEGtoOWF:2}
	   	$\pr{\Ts\notin \cs} \le 1/4p(m\ell +n)$.
	   \end{claim}
   \begin{proof}
   Let $\cs' = \set{\vt=(r_1,y_1,\ldots)\in \Supp(\Ts) \colon \forall i \in [m] \colon \SD\left((x)_{x\getsr g^{-1}(y_{\le i})},\Inv(y_{\le i})\right)  \le  \alpha \land \Fail(\vt) = \false}$. Since  $\Inv$ is an $\alpha$-inverter,  a simple coupling argument yields that 
   \begin{align}
   \pr{\Ts\notin \cs' } \le 2m\alpha
   \end{align}
   	We conclude the proof by showing that the second property of $\cs$ holds with high probability for  a random transcript of $\cs'$. Fix $\vt=(r_1,y_1,\ldots,r_m,y_m) \in \cs'$.   By definition, it holds that
   	\begin{align}
   	\pr{\tY_i = y_i \mid R_{<i} = r_{<i}} \le \pr{\tY_i = y_i \mid \tY_{<i} = y_{<i}} + \alpha
   	\end{align}
   	for every $\in [m]$. In particular, if $\pr{\tY_i = y_i \mid \tY_{<i} = y_{<i}} \ge 3pm\alpha$, then
   		\begin{align}
   	\Hall_{\tY_i|R_{<i}}(y_i|r_{<i}) +1/3mp \ge \Hall_{\tY_i|\tY_{<i}}(y_i|y_{<i})
   	\end{align}
   	Thus, we should only care about transcripts for which $\pr{\tY_i = y_i \mid \tY_{<i} = y_{<i}} <  3pm\alpha$ for some $i$.
   	
   	Note that for every possible value of $r_{<i}$, there exists at most a \emph{single}  value $y_i^\ast = y_i^\ast(r_{<i})$ such that  $\pr{\tY_i = y_i^\ast\mid R_{<i} = r_{<i}} > \pr{\tY_i = y_i^\ast \mid \tY_{<i} = y_{<i}}$ (\ie the value $\Gs$ outputs when $\Inv$ fails). Assuming  $\SD\left((x)_{x\getsr g^{-1}(y_{< i})},\Inv(y_{< i})\right)  \le  \alpha $, it holds that 
   	\begin{align}
   	\pr{\Inv(y_{< i}) = y}  \le 3pm\alpha + \alpha < 1/5mp(m\ell +n)
   	\end{align}
   	for every value $y$. Hence,
   	\begin{align}
   	\ppr{\vt=(r_1,y_1,\ldots,r_m,y_m)  \la \Ts}{\vt  \in \cs  \land \exists i\in [m] \colon y_i = y_i^\ast(r_{<i})} \le 1/5p(m\ell +n)
   	\end{align}
   	We conclude that  $\pr{\Ts \in \cs} \ge \pr{\Ts \in \cs'} -1/5p(m\ell+n) \ge 1- 1/4p(m\ell+n)$.
   \end{proof}
We now  use \cref{claim:IEGtoOWF:2} to  show that the expectation of  the sample-entropy $\tY$ is almost intact  when ignoring the contribution of  transcripts not in $\cs$. By  the second part of	  \cref{prop:ShanonoToSampleSeq},
\begin{align}
\ppr{\vt=(r_1,y_1,\ldots)\la \Ts}{\Hall_\tY(y = (y_1,\ldots,y_m) > \ell m  +n } \le 2^{-n}
\end{align}
Hence, the first part  of \cref{prop:HighAEContributionSeq} yields that for some universal constant $c$,
\begin{align}
\eex{\vt=(r_1,y_1,\ldots)\la \Ts}{(\Hall_\tY(y = (y_1,\ldots,y_m) > \ell m  +n)\cdot  \Hall_\tY(y)}  \le 2^{-n}(m \ell +n +n +c) <2^{-n/2}
\end{align}
for large enough $n$. It follows that
\begin{align*}
\eex{\vt=(r_1,y_1,\ldots)\la \Ts}{(\vt\notin \cs) \cdot  \Hall_\tY(y_1,\ldots,y_m)}  \le (\ell m +n)\cdot 1/4p(\ell m +n) + 2^{-n/2} < 1/3p.
\end{align*}
Thus,  
 	   \begin{align}\label{eq:IEGtoOWF1}
 	   \eex{\vt=(r_1,y_1,\ldots)\la \Ts}{(\vt\in \cs) \cdot  \Hall_\tY(y_1,\ldots,y_m)}  \ge \Hall(\tY) - 1/3p
 	   \end{align} 
 	   
 	    We conclude that
	   \begin{align*}
	   \Haccsam_{\Gs}(\Ts)& = \eex{\vt =(r_1,y_1,\ldots) \la \Ts}{\sum_{i=1}^{m} \Hall_{\tY_i|R_{<i}}(y_i|r_{<i})} \\
	   &\ge \eex{\vt =(r_1,y_1,\ldots) \la \Ts}{ (\vt \in \cs) \cdot (\sum_{i=1}^{m} \Hall_{\tY_i|R_{<i}}(y_i|r_{<i}))}\\
	    &\ge \eex{\vt =(r_1,y_1,\ldots) \la \Ts}{ (\vt \in \cs) \cdot (\sum_{i=1}^{m} \Hall_{\tY_i|\tY_{<i}}(y_i|y_{<i}) - 1/3mp)}\\
	    &\ge \eex{\vt=(r_1,y_1,\ldots)\la \Ts}{(\vt\in \cs) \cdot  \Hall_\tY(y_1,\ldots,y_m)}  - 1/3p\\
	    &\ge  \Hall(\tY) -   1/3p  - 1/3p\\
	    & >  \Hall(Y) - 1/p.
	   \end{align*} 
\end{proof}

\section*{Acknowledgements}
We thank Rosario Gennaro, Oded Goldreich and Muthuramakrishnan Venkitasubramaniam  for very helpful discussions.

\medskip

\bibliographystyle{abbrvnat}
\bibliography{crypto}

\appendix
\section{Maximal Accessible Entropy}\label{sec:MaxAE}
In this section we formally define the \emph{maximal (max)} accessible entropy  of a generator and provide basic observations and tools to work with this measure.  Working with accessible max-entropy raises some additional subtleties. In particular, to make  the manipulation defined in \cref{sec:manipulatingAE} applicable to this measure, we need to strengthen the notion of accessible entropy, so that it  takes into consideration  ``preprocessing randomness''.

\begin{definition}[Online block-generator,   preprocessing variant]
	Let $n$ be a security parameter, and let $\pp = \pp(n)$ and $m=m(n)$. An $m$-block  {\sf online} generator is a function $\Gs \colon \zo^\pp \times (\zo^{v})^{m+1} \mapsto  (\zs)^{m}$ for some $v =v(n)$, such that the \ith output block of $\Gs$ is a function of (only) its first $i+1$  input blocks. We denote the {\sf transcript} of $\Gs$ over random input by  $T_{\Gs}(1^n) = (Z,R_0,R_1,Y_1,\ldots,R_m,Y_m)$, for  $Z \ \zo^\pp$, $(R_0,R_1,\ldots,R_m) \getsr (\zo^{v})^{m+1}$ and $(Y_1,\ldots,Y_m) =\Gs(Z,R_0,R_1,\ldots,R_i)$.
\end{definition}

That is,  unlike the definition given in \cref{SHC:def:IE},  the generator's first block is a function of its  first \emph{two} random strings   (\ie $r_0$ and $r_1)$.  The role of the first string ($r_0$) is to allow the generator a (randomized) ``preprocessing stage'' before it starts computing  the output blocks, and the accessible   entropy of the generator  will be measured \wrt a random choice of this \emph{preprocessing randomness}. This preprocessing randomness becomes handy when bounding the accessible entropy of a generator constructed by manipulating (\eg repetition) of  another generator, as done in \cref{sec:manipulatingMaxAE}.\footnote{When considering \emph{nonuniform}  generators, as done in  \cite{HaitnerVadhan17}, there is no need for  the preprocessing randomness, since the generator can fix the best choice for this part of its random coins.}

\begin{definition}[Accessible sample-entropy, preprocessing variant]\label{def:AccessibleSampleEntropyMax}
	Let $n$ be a security parameter, and let $\Gs$ be an online $m=m(n)$-block online  generator. The {\sf accessible sample-entropy of $\vt =(z,r_0,r_1,y_1,\ldots,r_m,y_m)\in \Supp(Z,R_0,R_1,Y_1\ldots,R_m,Y_m) = T_{\Gs}(1^n)$} is defined by
	
	$$\Haccsam_{\Gs,n}(\vt) = \sum_{i=1}^{m} \Hall_{Y_i|Z,R_{<i}}(y_i|z,r_{<i}).$$
\end{definition}
\noindent 

That is,   the surprise  in $y_1$, \ie $\Hall_{Y_1|Z,R_{<1}}(y_1|z,r_{<1})$, is measured also \wrt the value of  $R_0$. This might  increase the sample-entropy of its output blocks when  a non-typical value  for $r_0$ is sampled.

The average accessible entropy of a generator with preprocessing is defined as in \cref{def:accessible-entropy}  \wrt  the above notion of sample-entropy. It is not hard to see that the two quantities are the same (for any generator).  For accessible max-entropy,  the preprocessing quantity defined next is more manipulation friendly, and we do not know that  is equivalent to its non-preprocessing variant.

\begin{definition}[Max accessible entropy,  preprocessing variant]\label{def:MaxAccessible-entropy}
	A block generator $\Gc$ has {\sf accessible max-entropy at most $k$} if for  every efficient $\Gc$-consistent, online generator $\Gs$ and all large enough $n$,
	$$\Pr_{\vt\getsr T_\Gs(1^n)} [\Haccsam_\Gs(\vt) > k] = \negl(n),$$
	for every such $\Gs$.
\end{definition}
We first note that the accessible max-entropy of a generator indeed bounds its accessible entropy.
\begin{lemma}\label{lem:MaxAvgAEAssym}
	Let $\Gc$ be an efficient block generator of accessible max-entropy $k$. Then its  accessible entropy is at most  $k(n) + 1/p(n)$, for any $p\in \poly$.
\end{lemma}
\begin{proof}
	Fix an efficient $\Gc$-consistent, online generator $\Gs$  and   $p\in \poly$.  Let $m$ and $\ell$ be the block complexity and maximal block length of $\Gc$. By assumption, for large enough $n$ it holds that
	$$\Pr_{\vt\getsr T_\Gs(1^n)} [\Haccsam_\Gs(\vt) > k(n)] \le \eps(n)$$
	for $\eps(n) = 1/p(n)\ell(n) m(n)n$. Since $\eps(  m \ell - \log 1/\eps) \le 1/p(n)$, a similar proof to that given in \cref{lem:MaxAvgAE} yields that $\ex{\Haccsam_{\Gs}(T_{\Gs}(1^n))} \le k(n) + 1/p(n)$.
\end{proof}

The following theorem is the accessible max-entropy entropy variant of \cref{thm:AEGfromOWF}, stating  that the accessible max-entropy of the one-way function generator described in \cref{ctr:AEG} is $n-\omega(\log n)$.  
\begin{theorem}[Max Inaccessible entropy generators from  one-way functions]\label{thm:MaxAEGfromOWF}
	If  $f\colon\zn \mapsto \zn$ is  one-way, then the efficient block-generator  $\Gc=\Gc^f$ defined in \cref{ctr:AEG} has accessible {\sf max}-entropy $n-\omega(\log n)$.
\end{theorem}
By \cref{lem:MaxAvgAEAssym},  \cref{thm:MaxAEGfromOWF} implies \cref{thm:AEGfromOWF}, but working with max-entropy its proof is significantly more complicated.
\begin{proof}[Proof of \cref{thm:MaxAEGfromOWF}]	
Suppose \cref{thm:MaxAEGfromOWF}  does not hold, and let $\Gs$  be  an efficient,   $\Gc$-consistent online block-generator  such that
	\begin{align}
	\ppr{\vt\getsr \vT_\Gs(1^n)}{\Haccsam_\Gs(\vt) > n -  c \cdot \log n} > \eps(n)
	\end{align}
	for some $\eps(n) = 1/\poly(n)$, and infinitely many $n$'s. In the following, we fix $n\in \N$ for which the above equation holds,  and omit it from the notation when its value is clear from the context. Let $m = n/\log n+1$ and let $\vv$ be a bound on the number of coins used by $\Gs$ in each round. The inverter $\Inv$ for $f$ is defined as follows:
	\begin{algorithm}[Inverter $\Inv$ for $f$ from the accessible max-entropy generator $\Gs$]\label{alg:MaxInv}~
		\begin{description}
			\item[Input:] $z\in \zn$ 
			\item[Operation:] 
		\end{description}
		\begin{enumerate}

			\item For $i = 1$ to $m-1$: \label{alg:MaxInv:step:2}
			\begin{enumerate}
				\item Sample $r_i \getsr \zo^\vv$ and let $y_i = \Gs(r_1,\ldots,r_i)_i$.
				\item If $y_{i\cdot \log n +1,\ldots,(i+1)\cdot \log n} = z_i$, move to next value of $i$.
				
				\item Abort after  $n^3/\eps$ failed attempts for sampling  a good $r_i$.\label{alg:MaxInv:step:2:c}
			\end{enumerate}
			\item Sample $r_m \getsr \zo^\vv$ and output $\Gs(r_1,\ldots,r_m)_m$.
		\end{enumerate}
		
	\end{algorithm}
\remove{	Namely, $\Inv(y)$ does the only natural thing that  can be done with $\Gs$; it tries to make, via rewinding,  $\Gs$'s first $n$ output blocks equal to $y$, knowing that if this happens then, since  $\Gs$ is $\Gc$-consistent, $\Gs$'s $m$-th output block is a preimage of $y$.}
	
	It is clear that  $\Inv$ runs in polynomial time, so we will finish the proof by showing that 
	$$\ppr{y\getsr f(U_n)}{\Inv(y) \in f^{-1}(y)}  \ge \eps^2/16n.$$

	We prove the above by relating the transcript distribution  induced  by  the standalone execution  of $\Gs(1^n)$ to that induced by  the  execution of $\Gs$ embedded (emulated) in $\Inv(f(U_n))$.  In more detail, we show that  high-accessible-entropy transcripts \wrt  the standalone execution of $\Gc$, \ie $\Haccsam_\Gs(\vt) > n -  c \cdot \log n$, are produced with not much smaller  probability also in the embedded execution. Since whenever $\Inv$ does not abort it  inverts $y$, it follows that the success probability of $\Inv$ is lower bounded by the probability that  $\Gs(1^n)$  outputs a  high-accessible-entropy transcript, and thus is non-negligible.
	
	For intuition about why  the above statement about   high-accessible-entropy transcripts is true, consider the case of a one-way  \emph{permutation} $f$. By definition, high-accessible-entropy transcripts in the standalone execution of $\Gs$ are produced with probability  at most $\poly(n)/2^n$. On the other hand, the  probability that a ``typical" transcript is produced by the emulated execution of $\Gs$ is about $2^{-n}$ : the probability that a random output of $f$ equals the transcript's first $n$ output blocks.  
	
	We now formally prove the above  for arbitrary one-way functions.
	
	\paragraph{Standalone execution  $\Gs(1^n)$.}
	Let $\tT = \vT_\Gs$, and recall that $\tT =  (\tR_1,\tY_1,\ldots,\tR_m,\tY_m)$ is associated  with a random execution of $\Gs$ on security parameter $n$ by
	\begin{itemize}
		\item $\tR_i$ -- the random coins of $\Gs$ in the \ith round, and
		\item $\tY_i$ -- $\Gs$'s \ith output block.
	\end{itemize}
	Recall that for $\vt = (r_1,y_1,\ldots,r_m,y_m) \in \Supp(\tT)$,  we have defined 
	\begin{align*}
	\Haccsam_{\Gs}(\vt) \eqdef \sum_{i\in [m]}  \Hall_{Y_j|R_{<j}}(y_j|r_{<j}).
	\end{align*}
	Compute
	\begin{align}\label{eq:TildeT}
	\ppr{\tT}{\vt}&= \prod_{i=1}^m \ppr{\tY_i|\tR_{<i}}{y_i |r_{<i}}  \cdot \ppr{\tR_i| 
		\tR_{<i},\tY_i}{ r_i|r_{<i},y_i}\\
	&= 2^{- \sum_{i=1}^m \HSh_{\tY_i | \tR_{<i}} (y_i | r_{<i})} \cdot \prod_{i=1}^m \ppr{\tR_i| 
		\tR_{<i},\tY_i}{ r_i|r_{<i},y_i}\nonumber \\
	&= 2^{-\Haccsam_{\Gs}(\vt)} \cdot R(\vt)\nonumber
	\end{align}
	for 
	\begin{align}
	R(\vt) \eqdef\prod_{i=1}^m  \ppr{\tR_i| \tR_{<i},\tY_i}{ r_i|r_{<i},y_i}
	\end{align}

	\paragraph{Execution embedded in $\InvG(f(U_n))$.}
	Let $\hT = (\hR_1,\hY_1,\ldots,\hR_m,\hY_m)$ denote the value of $\Gs$'s coins and  output blocks, sampled in \Stepref{alg:MaxInv:step:2}  of a random execution of the \emph{unbounded} version of $\Inv$ (\ie \Stepref{alg:MaxInv:step:2:c}  is removed) on input $Z = (Z_1,\ldots, Z_{m-1}) = f(U_n)$.  (This unboundedness change is only an intermediate step in the proof that does not  significantly change  the inversion probability of $\Inv$, as shown below.) 
	
	Since $\Gs$ is $\Gc$-consistent, it holds that $(y_1,\ldots,y_{m-1})\in \Supp(f(U_n))$ for every $(r_1,y_1,\ldots,r_m,y_m)\in \Supp(\tT)$. It follows that every $\vt \in \Supp(\tT)$ can be ``produced'' by the unbounded version of $\Inv$, and therefore $\Supp(\tT) \subseteq \Supp(\hT)$. For $\vt = (r_1,y_1,\ldots,r_m,y_m) \in  \Supp(\tT)$,  compute
	\begin{align}\label{eq:HatT}
	\ppr\hT\vt&=  \ppr{\hR_0}{r_0}\cdot \prod_{i=1}^m \ppr{\hY_i | \hR_{<i}}{y_i | r_{<i}} \cdot \ppr{\hR_i| \hR_{<i},\hY_i} {r_i | r_{<i},y_i}\\
	&=  \left(\prod_{i=1}^{m-1} \ppr{Z_i |\hY_{<i}}{y_i | y_{<i}} \cdot \ppr{\hY_i| \hR_{<i},Z_i}{y_i |r_{<i}, y_i}\right )\cdot  \ppr{\hY_{m}| \hR_{<m}}{y_m | r_{<m} } \nonumber\cdot \prod_{i=0}^m \ppr{\hR_i | \hR_{<i},\hY_i}{r_i|r_{<i},y_i}\nonumber\\
	&=  \left(\prod_{i=1}^{m-1} \ppr{Z_i |\hY_{<i}}{y_i | y_{<i}} \cdot 1\right )\cdot  \ppr{\hY_{m}| \hR_{<m}}{y_m | r_{<m} } \cdot \prod_{i=0}^m \ppr{\hR_i | \hR_{<i},\hY_i}{r_i|r_{<i},y_i}\label{eq:HatT:0}\\
	&=  \ppr{f(U_n)}{y_{<m}} \cdot  \ppr{\hY_{m}| \hR_{<m}}{y_m | r_{<m} } \cdot \prod_{i=0}^m \ppr{\hR_i | \hR_{<i},\hY_i}{r_i|r_{<i},y_i}\nonumber\\
	&=  \ppr{f(U_n)}{y_{<m}} \cdot  \ppr{\tY_{m}| \tR_{<m}}{y_m | r_{<m} } \cdot \prod_{i=0}^m \ppr{\hR_i | \hR_{<i},\hY_i}{r_i|r_{<i},y_i}\label{eq:HatT:1}\\
	&=  \ppr{f(U_n)}{y_{<m}}\cdot  \ppr{\tY_{m}| \tR_{<m}}{y_m | r_{<m} } \cdot \prod_{i=0}^m \ppr{\tR_i | \tR_{<i},\tY_i}{r_i|r_{<i},y_i} \label{eq:HatT:2}\\
	&= \ppr{f(U_n)}{y_{<m}} \cdot \ppr{\tY_{m} | \tR_{<m}}{y_m | r_{<m} } \cdot R(\vt)\nonumber,
	\end{align}
	where again, we let $\hY_0$ and $y_0$ stand for the empty strings.
	
	\cref{eq:HatT:0} holds since $\vt \in \Supp(\tT)$ and $\Inv$ is unbounded.   \cref{eq:HatT:1} holds since in both $\tT$ and in $\hT$, the last output block has the same distribution conditioned on all but the last randomness block.  \cref{eq:HatT:2}  holds since when conditioning on the value of the \ith output block, the randomness used to create this block is distributed the same  in $\tT$ and in $\hT$.   
	
	\paragraph{Relating the two distributions.}
	Combining \cref{eq:TildeT,eq:HatT} yields that,   for   $\vt=(r_1,y_1,\ldots,r_m,y_m)\in \Supp(\tT)$, it holds that 
	\begin{align}\label{eq:HatTandTildeS}
	\ppr\hT\vt  =  \ppr\tT\vt \cdot \left(\ppr{f(U_n)}{ y_{<m}} \cdot \ppr{\tY_{m} | \tR_{<m}}{y_m |r_{<m} } \cdot 2^{\Haccsam_{\Gs}(\vt)}\right)
	\end{align}
	In particular, if  $\Haccsam_{\Gs}(\vt) \geq n - c \log n$, then
	\begin{align}\label{eq:HatTandTilde}
	\ppr\hT\vt &\ge \ppr\tT\vt  \cdot \frac{2^n\cdot \ppr{f(U_n)}{y_{<m}} }{n^c} \cdot \ppr{\tY_{m} | \tR_{<m}}{ y_m| r_{<m} }\\
	&= \ppr\tT\vt  \cdot \frac{\size{f^{-1}(y_{<m})}  }{n^c} \cdot \ppr{\tY_{m}| \tR_{<m}}{y_m | r_{<m} }\nonumber.
	\end{align}
	If it is also the case that  $\Hsam_{\tY_m | \tR_{<m}}(y_m | r_{<m}) \leq \log \size{f^{-1}(y_{<m})} + k$ for some $k>0$, then
	\begin{align}\label{eq:HatTandTildeTT}
	\ppr\hT\vt \ge \ppr\tT\vt  \cdot  \frac{\size{f^{-1}(y_{<m})}}{n^c} \cdot \frac{2^{-k}}{\size{f^{-1}(y      _{<m})}} = \frac{\ppr\tT\vt}{ 2^{k} n^c}
	\end{align}

	\paragraph{Lower bounding the inversion probability of $\Inv$.}
	We conclude the proof by  showing  that  \cref{eq:HatTandTildeTT} implies the existence of  a large set of transcripts that (the bounded version of) $\Inv$ performs well upon. 
	
	Let $\cs$ denote the set of transcripts $\vt=(r_1,y_1,\ldots,r_m,y_m)\in \Supp(\tT)$ with
	\begin{enumerate}
		\item $\Haccsam_{\Gs}(\vt) \geq n - c \log n$, 
		\item  $\Hsam_{\tY_m | \tR_{<m}}(y_m | r_{<m}) \leq \log \size{f^{-1}(y_{<m})} +  \log (4/\eps)$, and
		\item  $\Hsam_{\tY_i | \tY_{<i}}(y_i | y_{<i}) \leq \log (4n^2/\eps)$ for all $i\in [m-1]$.
	\end{enumerate}
	The first two properties will allow us to use \cref{eq:HatTandTilde,eq:HatTandTildeTT} to argue that, if   $\cs$ happens  with significant probability \wrt $\tT$, then  this holds also \wrt $\hT$. The last property will allow us to show that this  also holds \wrt  the bounded version of $\Inv$. We start by  showing that $\cs$ happens with significant  probability \wrt $\tT$, then   show that this holds also \wrt $\hT$, and finally  use it  to lowerbound the success probability of $\Inv$.

	By \cref{prop:ShanonoToSample}, 
	\begin{align}\label{eq:SitemThree}
	\ppr{(r_1,y_1,\ldots,r_m,y_m)  \getsr \tT}{\Hsam_{\tY_m | \tR_{<m}}(y_m | r_{<m}) > \log\size{f^{-1}(y_{<m})} + k }< 2^{-k}
	\end{align}
	for  any $k>0$. Since $\size{\Supp(\tY_i)}=n$ for all $i\in [m-1]$,  it follows that 
	\begin{align}\label{eq:Sitemtwo}
	\ppr{(y_1,\ldots,y_m) \getsr (\tY_1,\ldots,\tY_m)}{\exists i\in [m-1] \colon \Hsam_{\tY_i |
			\tY_{<i}}(y_i | y_{<i}) >  v} <  (m-1)\cdot  n  \cdot 2^{-v}
	\end{align}
	for  any $v>0$.
	
	Applying  \cref{eq:SitemThree,eq:Sitemtwo} with $k= \log (4/\eps)$ and $v= \log(4mn/\eps)$, respectively, and recalling that, by assumption,   $\ppr{\vt \getsr \tT}{\Haccsam_{\Gs}(\vt) \geq n - c\log n} \ge \eps$, yields that  
	\begin{align}\label{eq:SisLargeTilda}
	\Pr_\tT[\cs] \geq \eps - \frac{\eps}4 -  \frac{\eps}4  = \eps/2
	\end{align}
	By \cref{eq:HatTandTildeTT} and the first two properties of $\cs$, we have that       
	\begin{align}\label{eq:SisLargeHat}
	\Pr_\hT[\cs] &\geq \frac{\eps}{4n^c} \cdot \Pr_\tT[\cs] \ge \frac{\eps^2}{8n^c}
	\end{align}
	Finally, let $\hT'$ denote the final value of $\Gs$'s coins and output blocks,  induced by the \emph{bounded} version of $\Inv$ (set to $\bot$ if $\Inv$ aborts). The third property of $\cs$ yields that   
	\begin{align}
	\ppr{\hT'}{\vt } \geq \ppr{\hT}{\vt }  \cdot \left(1 - (m-1)\cdot (1- \tfrac{\eps}{4n^2})^{n^3/\eps}\right) \geq \ppr{\hT}{\vt } \cdot (1 - O(m\cdot 2^{-n})) \geq \ppr{\hT}{\vt }  /2
	\end{align}
	for every $\vt \in \cs$. We conclude that 
	\begin{align*}
	\ppr{z\getsr f(U_n)}{\Inv(z)\in f^{-1}(z)} &= \ppr{z\getsr f(U_n)}{\Inv(z) \mbox{ does not abort}}\\
	&\ge   \ppr{\hT'}{\cs }   \\
	& \ge  \frac 12 \cdot \Pr_\hT[\cs]\\
	&   \ge \frac{\eps^2}{16n^c}.
	\end{align*}
\end{proof}

\subsection{Manipulating  Accessible Max-Entropy}\label{sec:manipulatingMaxAE}
In this section we  analyze the effect of the  tools introduced in \cref{sec:manipulatingAE}  on the accessible max-entropy of the generator (rather than on average accessible entropy).  The following statements and proofs are similar to those in \cref{sec:manipulatingAE}, but are somewhat more complicated due to the more complicated nature of accessible max-entropy, and  the preprocessing string  of the online generator is critical to these  proofs.

\subsubsection{Truncated Sequential Repetition}

\begin{lemma}\label{lem:EqRealEntMaxEntropy}
	For security parameter $n$, let $m =m(n)$ be a power of $2$, let $s=s(n)$, let $\Gc$ be an efficient $m$-block generator over $\zo^{s}$, and let $\ee = \ee(n)$ be a polynomially computable and bounded integer function. Then $\Gbl$ defined according to \cref{def:EqRealEnt} is an efficient\footnote{Since $m$ is a power of $2$, standard techniques can be applied to change the input domain of $\Gbl$ to  $\zo^{s'}$ for some polynomial-bounded and polynomial-time computable $s'$, making it an efficient  block-generator  according to \cref{def:blockGenrator}.} $\left((\ee-1)\cdot m\right)$-block generator  such that the following holds:  if $\Gc$ has  accessible max-entropy at most $\kmax = \kmax(n)$, then $\Gbl$ has  accessible max-entropy at most 
		$$\kmaxp \eqdef (\ee-2)  \cdot \kmax + 2\cdot \Hmax(\Gc(U_s)) +  \log(m)  +d$$  
		 for any $d= d(n) \in \omega(\log n)$.
	\end{lemma}
Roughly, each of the $(\ee-2)$ non-truncated executions of $\Gc$ embedded in $\Gbl$ contributes its accessible entropy to the overall accessible entropy of  $\Gbl$. As in the average accessible entropy case, we pay the max-entropy of  the two truncated executions of $\Gc$ embedded in $\Gbl$.  Working with the less friendly measure of accessible max-entropy costs us an additional super-logarithmic loss  $d$, a price we do not pay when working with  average accessible entropy.
 
\begin{proof}
	To avoid notational clutter, let $\Gb = \Gbl$. Let $\Gbs$ be  an efficient  $\Gb \ (= \Gbl)$-consistent generator, and let
	\begin{align}
	\eps = \eps(n) =  \ppr{\vt \getsr  \Tbs}{\Haccsam_{\Gbs}(\vt) > \kmaxp}
	\end{align}
	for $\Tbs = T_{\Gbs}(1^n)$.	Our goal is to show that $\eps$ is negligible in $n$. We do that by showing that a random sub-transcript of  $\Tbs$  contributes more than  $\kmax$ bits of  accessible entropy if the overall accessible entropy of $\Tbs$ is more than $\kmaxp$. We then use this observation to construct a cheating generator for $\Gc$ that achieves accessible entropy greater than $\kmax$  with probability that is negligibly close to $\eps$.

	Let $ (\vZ,\tRR_0,\tRR_1,\tYY_1,\ldots,\tRR_\mt,\tYY_\mt) = \Tbs$ and let $J$ be the first part of  $\tYY_1$ (recall that $\tYY_1$ is of the form $(j,\cdot)$). Fix $j\in [m]$, and let   $(\vZ, \tRR_0^j,\tRR_1^j,\tYY_1^j,\ldots,\tRR_\mt^j,\tYY^j_\mt) = \Tbs^j = \Tbs|_{J=j}$. Let $\I = \I(j)$ be the indices of the output blocks coming from the truncated executions of $\Gc$ in $\Gb$ (\ie $\I = \set{1,\ldots,m+1-j} \cup \set {\mt+2-j,\ldots,\mt}$). 
	
	Our first step  is to show that these blocks do not contribute much more entropy than  the max-entropy of $\Gc(U_n)$. Indeed, by \cref{prop:ShanonoToSampleSeq}, letting $\vX=(\vZ,\tRR_0^j,\tYY_1^j,\tRR_1^j,\ldots,\tYY_\mt^j,\tRR_\mt^j)$ and   $\cJ$  being the indices of the  blocks of $\I$ in $\vX$, it holds that 
	\begin{align}\label{EqRealEntMax:1}
	\ppr{\vt= (\vz,r_0,r_1,y_1,\ldots,r_{\mt},y_\mt) \getsr  \Tbs^j}{\sum_{i\in \I} \Hall_{\tYY_i^j|\vZ,\tRR_{<i}^j}(y_i|\vz,r_{<i}) > 2\cdot \Hmax(\Gc(U_s)) + d/2} \le 2\cdot 2^{-d/2} = \negl(n)
	\end{align}
	Namely, with save but negligible probability, the blocks that relate to the truncated executions of $\Gc$ in $\Gb$ do not contribute much more than their support size to the overall accessible entropy.

	Our next step is to remove the conditioning on $J=j$ (that we have introduced to have the indices of interest fixed, which enabled us to use \cref{prop:ShanonoToSampleSeq}). By  \cref{prop:ShanonoToSample}, for any $\vz \in \Supp(\vZ)$ and $r_0 \in \Supp(\tRR_0)$ it holds that 
	\begin{align}\label{EqRealEntMax:2}
	\ppr{j \getsr  J| \vZ = \vz,\tRR_0 = r_0}{\Hall_{J|\vZ,\tRR_0}(j|\vz,r_0) > \log(m)+ d/2               } \le 2^{-d/2} = \negl(n)
	\end{align}

	Let $(\vz,r_0,r_1,y_1 = (j,\cdot),\ldots,r_{\mt},y_\mt) \in \Supp(\Tbs)$. For $i >1$, it holds that $\Hall_{\tYY_i|\vZ,\tRR_{<i}}(y_i|\vz,r_{<i}) = \Hall_{\tYY_i^{j}|\vZ,\tRR_{<i}^j}(y_i|\vz,r_{<i})$, whereas  for $i=1$, it holds that $\Hall_{\tYY_1|\vZ,R_0}(y_1|\vz,r_0) =\Hall_{J|\vZ,\tRR_0}(j|\vz,r_0) + \Hall_{\tYY_1^{j}| \vZ,\tRR_0}(y_1|\vz,r_0)$. Hence, by \cref{EqRealEntMax:1,EqRealEntMax:2} it holds  that 
	\begin{align}
	&\ppr{\vt= (\vz,r_0,r_1,y_1,\ldots,r_{\mt},y_\mt) \getsr  \Tbs}{\sum_{i\in [\mt] \setminus \I(J)} \Hall_{\tYY_i|\vZ,\tRR_{<i}}(y_i|\vz,r_{<i})  > (\ee -2)\cdot \kmax } \ge \eps - \negl(n)
	\end{align}
	Let $\F(j) = \set{km + 2-j \colon k\in [\ee-2]}$, \ie the indices of the first blocks of the  non-truncated executions of $\Gc$ in $\Gb$, when the first block of $\Gb$ is $(j,\cdot)$. It follows that 
	\begin{align}\label{eq:EqAccEntMax}
	&\ppr{\vt= (\vz,r_0,r_1,y_1=(j,\cdot),\ldots,r_{\mt},y_\mt) \getsr  \Tbs; f\getsr\F(j)}{\sum_{i=f}^{f+m-1} \Hall_{\tYY_i|\vZ,\tRR_{<i}}(y_i|\vz,r_{<i})> \kmax} \ge (\eps - \negl(n))/w
	\end{align}

	\EEGen

	Let $(Z,R_0' = (R_0,F,\vZ_{-F}),R_1,Y_1,\ldots,R_m,Y_m) = T_\Gs$ be the transcript of $\Gs(Z)$. It is easy to verify that
	\begin{align}
	\Hall_{Y_i|Z,R_{<i}}(y_i|z,r'_{<i}) = \Hall_{\tYY_{f + i}|\vZ,\tRR_{<f + i}}(y_i|\vz = (z,\vz_{-f}),r'_{<i} = (r_0',r_1,\ldots,r_{<i}))
	\end{align}
	for every $(z,r_0' = (r_0,f,\vz_{-f}),r_1,y_1,\ldots,r_m,y_m) \in \Supp(T_\Gs)$ 	and $1 < i \le m$. Thus, \cref{eq:EqAccEntMax} yields  that
	\begin{align*}
	\ppr{\vt \getsr  T_\Gs} {\Haccsam_\Gs(\vt) > \kmax} > (\eps - \negl(n))/w.
	\end{align*}
	
	The assumption about the inaccessible entropy of $\Gc$ yields that $\eps$ is a negligible function of $n$, and the proof of the lemma follows.
\end{proof}

\subsubsection{Direct Product }
\begin{lemma}\label{lem:PRMaxAE}
	For security parameter $n$, let $m =m(n)$, let  $\vv= \vv(n)$ be polynomial-time   computable and bounded integer functions, and let $\Gc$ be an efficient\footnote{Again, standard techniques can be applied to change the input domain of $\G$ to  $\zo^{s'(n)}$ for some polynomial-bounded and polynomial-time computable $s'$, making it an efficient  block-generator  according to \cref{def:blockGenrator}.}  $m$-block generator. Then $\Gt$,  defined according to \cref{def:PR}, is an efficient $m$-block generator such that the following holds: if $\Gc$ has accessible max-entropy at most $\kmax = \kmax(n)$, then $\Gt$ has accessible max-entropy at most $\kmaxp(n) = \vv \cdot \kmax  + d\cdot m$, for any   $d= d(n) \in \omega(\log n)$.
\end{lemma}

 As in the truncated sequential repetition \cref{lem:EqRealEntMaxEntropy}, working with the less friendly measure of accessible max-entropy costs us an additional super-logarithmic loss , a price we do not pay when working with  average accessible entropy.

\begin{proof}
	 Let $\Gb = \Gt$, let $\Gbs$  be  an efficient  $\Gb$-consistent generator, and let
	\begin{align}
	\eps = \eps(n) =  \ppr{\vt \getsr  \Tbs}{\Haccsam_{\Gbs}(\vt) > \kmaxp}
	\end{align}
	for $\Tbs = T_{\Gbs}(1^n)$. 	Our goal is to show that $\eps$ is negligible in $n$.
	
	Let $(\vZ,R_0,R_1,Y_1,\ldots,R_m,Y_m)= \Tbs$.  Recall that  for $\vt= (\vz,r_0,r_1,y_1,\ldots,r_m,y_m)\in \Supp(\Tbs)$, we  have defined
	\begin{align}
	\Haccsam_{\Gbs}(\vt) =\sum_{i\in [m]} \Hall_{Y_i \mid \vZ.R_{<i}}(y_i \mid \vz,r_{<i})
	\end{align}
	Since $\Gbs$ is   $\Gb$-consistent, each  $Y_i$ is of the form $(Y_{i,1},\ldots, Y_{i,\vv})$. Hence \cref{prop:subaddativiy2}, taking $\vX=Y_i|_{\vZ= \vz, R_{<i}=r_{<i}}$ yields that 
	\begin{align}
	\ppr{\vt = (\vz,r_0,r_1,y_1,\ldots,r_m,y_m) \getsr\Tbs}{\Hall_{Y_i|\vZ,R_{<i}}(y_i|\vz,r_{<i}) >d + \sum_{j=1}^\vv \Hall_{Y_{i,j}|\vZ,R_{<i}}(y_{i,j}|\vz,r_{<i}) }  \le 2^{-d}  = \negl(n)
	\end{align}
	for every $i\in [m]$. Summing over all $i\in[m]$, we get that
	\begin{align}
	\ppr{\vt = (\vz,r_0,r_1,y_1,\ldots,r_m,y_m) \getsr\Tbs}{\sum_{i\in [m]}\Hall_{Y_i|\vZ,R_{<i}}(y_i|\vz,r_{<i}) >m d + \sum_{i\in [m]}\sum_{j\in [\vv]} \Hall_{Y_{i,j}|\vZ,R_{<i}}(y_{i,j}|\vz,r_{<i}) } = \negl(n)
	\end{align}
	It follows that 
	\begin{align*}
	\ppr{\vt = (\vz,r_0,r_1,y_1,\ldots,r_m,y_m) \getsr\Tbs}{\sum_{i\in [m]}\sum_{j\in [\vv]}\Hall_{Y_{i,j}|\vZ,R_{<i}} (y_{i,j}|\vz,r_{<i})  \ge  \kmaxp - m\cdot d} \ge \eps - \negl(n),
	\end{align*}
	and therefore
	\begin{align}\label{eq:PrAccEntMax}
	\ppr{\vt = (\vz.r_0,r_1,y_1,\ldots,r_m,y_m) \getsr\Tbs, j\getsr [\vv]}{\sum_{i\in [m]} \Hall_{Y_{i,j}|\vZ,R_{<i}}(y_{i,j}|\vz,r_{<i}) >\kmax  = (\kmaxp - m \cdot d)/\vv} \ge \eps - \negl(n)
	\end{align}
	
	\PRGen
	
	Let $(Z,R_0'= (R_0,J,\vZ_{-J}),R_1',Y_1',\ldots,R_m',Y_m') = T_\Gs$. It is easy to verify that
	\begin{align}
	\Hall_{Y_i'|Z,R'_{<i}}(y_i| z,(r_0,j,\vz_{-j}),r_1,\ldots,r_{i-1}) = \Hall_{Y_{i,j}|\vZ,R_{<i}}(y_i|\vz = (z,\vz_{-j}),r'_{<i}= (r_0'.r_1,\ldots,r_{i-1}))
	\end{align}
	for every $(z,r_0' = (r_0,j,\vz_{-j}),r_1,y_1,\ldots,r_m,y_m) \in \Supp(T_\Gs)$ and $1 < i \le m$.
	Thus, \cref{eq:PrAccEntMax} yields  that
	\begin{align*}
	\ppr{\vt \getsr  T_\Gs} {\Haccsam_\Gs(\vt) > \kmax} \ge \eps - \negl(n).
	\end{align*}
	The assumption about the inaccessible entropy of $\Gc$ yields that $\eps$ is negligible in $n$, and the proof of the lemma follows.
\end{proof}

\remove{

\subsection{Working with max entropy}
\Inote{To be removed}
\begin{theorem}[Inaccessible entropy generator to  statistically hiding commitment, max entropy version]\label{theorem:GapToComMax}
	Let $k= k(n)$, $\pp=\pp(n)$, $s=s(n)$, and $\delta= \delta(n)$ be  polynomial-time computable functions. Let  $\Gc$ be an efficient  $m=m(n)$-block generator over $\zo^{\pp} \times \zo^{s}$, and assume that  $\Gc$'s real Shannon entropy is at least $k$, that  its accessible max-entropy is at most $(1-\delta)\cdot k$, and that  $k\delta \in  \omega(\log n/n)$. Then  for any polynomial-time computable $g= g(n) \in \omega(\log n)$ with  $g \ge \Hmax(G(U_\pp,U_s))$, there exists an  $O(m \cdot g/\delta  k)$-round,  receiver public-coin, statistically hiding and computationally binding commitment scheme. Furthermore, the construction is black box, and on security parameter $1^n$, the commitment invokes $\Gc$ on inputs of length $s(n)$.\footnote{Given a, per $n$, polynomial-size advice, the commitment round complexity can be reduced to $O(m)$. See \cref{remark:constantRound} for  details.} 
\end{theorem}

\begin{proof}[Proof of \cref{theorem:GapToComMax}]
	We prove \cref{theorem:GapToCom} by  manipulating the real and accessible entropy of  $\Gc$ using the tools described in \cref{sec:manipulatingAE}, and then applying \cref{lemma:GapToCom}  on the resulting generator.

	\paragraph{Truncated sequential repetition: real entropy equalization.}  	In this step we use  $\Gc$ to define  a generator  $\Gc^{\brk\ee}$ whose  \emph{each block has the  same amount} of  real entropy: the average of the real entropy of the blocks of $\Gc$.   In relative terms, the entropy gap of $\Gc^{\brk \ee}$  is essentially  that of $\Gc$.  
	
	We assume \wlg that $m(n)$ is a power of  two.\footnote{Adding $2^{\ceil{\log m(n)}} -m(n)$ final blocks of constant value transforms a  block-generator  to one whose block complexity is a power of two, while maintaining the same amount of real and accessible entropy.} Consider the  efficient  $m'= m'(n)=(\ee-1)\cdot m)$-block generator $\Gc^{\brk \ee}$ resulting by  applying  truncated sequential repetition (see \cref{def:EqRealEnt}) on $\Gc$ with parameter $\ee = \ee(n)=\max\set{4,\ceil{16g/\delta k}} \le \poly(n)$. By \cref{lem:EqRealEnt}, taking $d= g$, it holds that
	
	\begin{itemize}
		\item  \emph{Each  block}  of $\Gc^{\brk \ee}$ has real entropy  at least $k= k'(n)  = k/m$.
		
		\item   The accessible max-entropy   of $\Gc^{\brk \ee}$ is at most  
		\begin{align*}
		a'= a'(n) &= (\ee-2)   \cdot ((1-\delta)\cdot k  +  \log m + 2 g +g\\
		&\le (\ee-2)  \cdot (1-\delta)\cdot k   + 4g\\
		&\le (\ee-2)  \cdot (1-\delta/2)\cdot k - (\ee -2) \cdot k\cdot \delta/2 +  4g\\
		&\le (\ee-2)  \cdot (1-\delta/2)\cdot k  - \ee\cdot k\cdot \delta/4 +  4g\\
		&\le  (\ee-2)  \cdot (1-\delta/2)\cdot k\\
		&<  m'  \cdot (1-\delta/2)\cdot k'.
		\end{align*}
	\end{itemize}

	\paragraph{Direct product : converting real entropy to   min-entropy   and gap amplification.} In this step we use   $\Gc^{\brk \ee}$ to define a generator  $(\Gc^{\brk \ee})^{\seq \vv}$ whose each  block has the same amount of \emph{min-entropy}---about $\vv$ times the per-block entropy of  $\Gc^{\brk \ee}$. The accessible entropy of  $(\Gc^{\brk \ee})^{\seq \vv}$  is also about $\vv$ times that of $\Gc^{\brk \ee}$.  
	
	We assume \wlg  that the  output blocks are of all of the same length $\ell = \ell(n)\in \Omega(\log n)$.\footnote{Using padding technique one can transform a block-generator to one whose all  blocks are of the same length, without changing its real  and its accessible entropy.}   Consider the efficient  $m'$-block generator $(\Gc^{\brk \ee})^{\seq \vv}$ resulting from applying the  gap amplification transformation (see \cref{def:PR}) on $\Gc^{\brk \ee}$ with $\vv = \vv(n) = \max\set{24n/k'\delta,\ceil{c\cdot \left(\log n\cdot \ell)/k'\delta\right)^2}}$, for $c>0$ to be determined by the analysis. By  \cref{lem:PR} it holds that 
	\begin{itemize}
		\item Each block of $(\Gc^{\brk \ee})^{\seq \vv}$ has real \emph{min-entropy}  at least  $k'' = k''(n) =  \vv \cdot k' - O\left(\log (n)\cdot  \ell \cdot  \sqrt{\vv} \right)$.
		
		\item The  accessible max-entropy of $(\Gc^{\brk \ee})^{\seq \vv}$  is at most $a'' = a''(n) =  \vv \cdot a' +   d \cdot m'$, for $d = d(n) = n/8k\delta$.
	\end{itemize}
	Hence for large enough $n$, it holds that 
	\begin{align}
	m' \cdot k'' - a''&\ge  m' \cdot \left(\vv \cdot k' - O\left(\log (n)\cdot  \ell \cdot  \sqrt{\vv} \right)\right) - \left(v\cdot a' +  d\cdot m'\right) \nonumber\\
	&>  m' \cdot \left(\vv \cdot k' - O\left(\log (n)\cdot  \ell \cdot  \sqrt{\vv} \right)\right) - \left(v\cdot (m'  \cdot (1-\delta/2)\cdot k')+  d\cdot m'\right)\nonumber\\
	&\ge  m' \cdot v  \cdot \left(k'\delta/4  -  d/\vv \right)\label{eq:GapToCom:prf:Max:1}\\
	&\ge  m' \cdot v  \cdot k'\delta/8 \nonumber\\
	&\ge 3m' n.\nonumber
	\end{align}
	Inequality (\ref{eq:GapToCom:prf:Max:1}) holds by taking a large enough value of $c$ in the definition of $\vv$. Namely, the overall real entropy of $(\Gc^{\brk \ee})^{\seq \vv}$ is larger than its accessible  max-entropy  by at least $3m' n$. 
	
	By \cref{lemma:IEGtoOWF}, the existence of $\Gc$ implies that of one-way functions. Thus by \cref{thm:pwf-to-uowhf}, its existence  implies the existence of   a universal one-way function from the string whose length if the block length of $(\Gc^{\brk \ee})^{\seq \vv}$  to strings of length $n/4$. Hence, by applying  \cref{lemma:GapToCom} with $(\Gc^{\brk \ee})^{\seq \vv}$,  $k =k''$ and $p =2$, we get  the claimed  $\left(m'= m\cdot (\ee-1) = O(m \cdot g/\delta  k)\right)$-round,  receiver public-coin,  statistically hiding and computationally binding commitment.
\end{proof}
}

\end{document}